\documentclass[10pt]{article}
\usepackage{amsmath}
\usepackage{amsthm}
\usepackage{amsfonts}
\usepackage{amssymb}
\usepackage{mathrsfs}
\usepackage{enumitem}
\usepackage{mathtools}
\usepackage{algorithm}
\usepackage{algcompatible}
\usepackage{algpseudocode}
\usepackage{latexsym} 
\usepackage{graphicx}
\usepackage{caption}
\usepackage{float}
\usepackage[table]{xcolor}
\usepackage{diagbox}
\usepackage{changepage}
\usepackage{listings}
\usepackage{dsfont}
\usepackage{setspace}
\usepackage{subfigure}
\usepackage{multirow}
\usepackage{scalerel}
\usepackage{tikz}
\usepackage{indentfirst}
\usepackage{natbib}
\usepackage{threeparttable}
\usepackage{booktabs,makecell}
\usepackage{mathpazo}
\usepackage[T1]{fontenc}
\usepackage{times}

\usepackage{hyperref}
\hypersetup{
            colorlinks={true}, 
            linkcolor={blue}, 
            citecolor=blue,
            pdfencoding=auto, 
            psdextra
            }

\usepackage[capitalise,compress]{cleveref} %
\newlist{assumpenum}{enumerate}{1} 
\setlist[assumpenum]{label=(\arabic*), ref=\theassumption~(\arabic*)}
\crefname{assumpenumi}{Assumption}{Assumptions}

\newlist{lemmaenum}{enumerate}{1}  
\setlist[lemmaenum]{label=(\arabic*), ref=\thelemma~(\arabic*)}
\crefname{lemmaenumi}{Lemma}{lemma}

\newlist{propenum}{enumerate}{1}  
\setlist[propenum]{label=(\arabic*), ref=\theproposition~(\arabic*)}

\crefname{assumption}{Assumption}{Assumptions}

\usepackage[matrix,tips,graph,curve]{xy}
\usepackage{graphicx}
\newcommand{\indep}{\rotatebox[origin=c]{90}{$\models$}}



\makeatletter 
\@addtoreset{equation}{section}
\makeatother

\makeatletter
\newcommand*{\rom}[1]{\expandafter\@slowromancap\romannumeral #1@}

\newtheoremstyle{plain}     
  {3pt}{3pt}{}{}{\bfseries}{.}{ }{}       

\theoremstyle{plain}
\newtheorem{theorem}{Theorem}[section]
\newtheorem{lemma}{Lemma}[section]     
\newtheorem{corollary}{Corollary}[section]
\newtheorem{assumption}{Assumption}[section]
\newtheorem{proposition}{Proposition}[section]
\newtheorem{remark}{Remark}[section]

\theoremstyle{definition}
\newtheorem{definition}{Definition}[section]

\newcommand{\argmin}{\operatorname{argmin}}

\begin{document}
\title{Uniform Inference on Quantile Effects under Network Interference\footnote{
We thank Jianfei Cao, Jizhou Liu, Yuya Sasaki, and Yichong Zhang for their helpful comments and discussions.}}

\author{
Zequn Jin\thanks{School of Economics, Shanghai University of Finance and Economics. Emails: \href{mailto:jinzequn@mail.shufe.edu.cn}{\nolinkurl{jinzequn@mail.shufe.edu.cn}}; \href{mailto:zy.zhang@mail.shufe.edu.cn}{\nolinkurl{zy.zhang@mail.shufe.edu.cn}}.}
\and
Gaoqian Xu\thanks{Tippie College of Business, University of Iowa. Email: \href{mailto:gaoqian-xu@uiowa.edu}{\nolinkurl{gaoqian-xu@uiowa.edu}}.}
\and
Zixin Yang\thanks{School of Statistics and Data Science, Shanghai University of Finance and Economics. Email: \href{mailto:yangzixin@mail.shufe.edu.cn}{\nolinkurl{yangzixin@mail.shufe.edu.cn}}.}
\and 
Zhengyu Zhang\footnotemark[2]
}
\date{\today}

\maketitle

\vspace{0.5cm}

\begin{abstract}
This paper studies quantile treatment and spillover effects in network experiments. Average spillover effects reveal how treating a unit’s neighbors affects its outcome on average, but mask the heterogeneity of these effects across the outcome distribution. We define structural quantile effects that compare outcome quantiles between exposure states, characterizing how own treatment and exposure to treated neighbors affect different parts of the outcome distribution. Building on \citet{leung2020treatment}, we first establish the weak convergence of the estimated quantile-effect process under conditions requiring the stabilization of the degree distribution and the network-dependent covariance structure. Our main contribution is to propose uniform confidence bands (UCBs) based on Gaussian approximations conditional on the realized network, avoiding these stabilization requirements. The proposed method is evaluated through extensive
simulation studies and an empirical application to a randomized savings-account experiment in Nepal \citep{prina2015banking}.
\end{abstract}

\vspace{12pt}


\vspace{3pt}


\newpage
\section{Introduction}

Network experiments have become increasingly prevalent in economics and the social sciences, with applications to information diffusion, technology adoption, microfinance, school-based social interventions, and input-subsidy programs \citep{BandieraRasul2006, BakshyEtAl2012, BanerjeeEtAl2013, cai2015social, PaluckEtAl2016, carter2021subsidies}. In such experiments, an individual’s outcome may depend not only on their own treatment assignment but also on the assignments of other units in the network, thereby violating the stable unit treatment value assumption (SUTVA).  A growing literature has therefore developed methods for causal inference under network interference, typically by summarizing the high-dimensional assignment vector through low-dimensional exposure mappings such as own treatment status and the fraction of treated neighbors \citep{manski2013identification,aronow2017estimating,leung2020treatment,leung2022causal}.

To our knowledge, most existing work focuses on average treatment and spillover effects. Yet average effects may mask economically important heterogeneity across the outcome distribution. For example, \citet{carter2021subsidies} find positive network spillover effects of agricultural input subsidies on maize yields in Mozambique, where maize is a major source of household income. However, their average-effect analysis does not reveal how these gains are distributed. In particular, it does not reveal whether farmers at the lower tail of the yield distribution benefit from technology diffusion through social networks. Motivated by this question, this paper studies quantile treatment and spillover effects under network interference and develops uniform inference for the entire quantile-effect curve.

 Building on the structural quantile model of \cite{matzkin2003nonparametric} and \cite{chernozhukov2005iv}, we assume that an individual’s potential outcome depends on their own treatment, the number of treated neighbors, the number of friends (node degree), and unobserved characteristics. Our parameters of interest are comparisons of quantile response functions between exposure states at each quantile level. These comparisons capture how variations in own treatment or friend exposure change the entire distribution of economic outcomes. In an experimental setting, the quantile response functions can be identified by the conditional quantile of observed outcomes among units with the same own treatment status, number of treated friends, and node degree, following \citet{leung2020treatment}.

The identification leads to a simple estimator based on conditional quantile regression (QR). The inferential targets are quantile effects at different quantile levels, so valid inference  requires uniform control over the quantile effect curve. This is a nontrivial extension of large-network inference for average treatment and spillover effects in  \cite{leung2020treatment}. Moreover, relative to the existing literature on quantile treatment effect (QTE), our inference problem is nonstandard for two reasons. First, the conditioning variables in our QR are network-dependent. For example, the number of treated friends is tied to node degree and mechanically depends on neighboring treatment assignments. Second, unobservable characteristics may exhibit network dependence, since linked individuals or individuals with common neighbors may share common information and local shocks. Therefore, standard weighted or multiplier bootstrap procedures developed for QR and QTE inference \citep{belloni2019conditional,qu2019uniform,chiang2019robust,zhang2020quantile} do not directly apply here without modification.

A natural route for asymptotic inference is to embed the observed network in a growing sequence of graphs and derive a limit distribution for the estimator along that sequence \citep[e.g.,][]{leung2020treatment,leung2022causal,kojevnikov2021limit}.  To obtain such a limit, this route requires the graph sequence and the induced dependence to stabilize asymptotically.  For example, \citet{leung2020treatment} assumes convergence of the empirical node-degree distribution and  a probability limit for the network-dependent asymptotic variance.\footnote{Assumption 2 in \citet{leung2020treatment} imposes convergence of the degree distribution, and Assumption 6 assumes the existence of a probability limit for the network-dependent asymptotic variance.} He estimates the average treatment and spillover effects in a single large network, establishes their asymptotic normality, and conducts inference using consistent estimators of the asymptotic variance that account for network dependence. As a benchmark, \cref{subsec:nonparametric-QTE} follows this route and extends it to quantile effects under network interference, establishing the weak convergence of quantile effect estimators. However, this approach requires convergence of the degree distribution and the network-dependent covariance structure, both of which are difficult to verify from a single observed network.

Our main contribution is the construction of uniform confidence bands (UCBs) for quantile effects under network interference, using Gaussian approximations conditional on the realized network. A key advantage of this approach is that the procedure does not require the QR process to converge to a fixed Gaussian process, nor does it require the studentized supremum statistic to have a fixed limiting distribution. To compute the critical values for the UCBs, we  adopt the Gaussian approximation framework of \citet{chernozhukov2013gaussian,chernozhukov2014anti,chernozhukov2014gaussian}  to approximate the distribution of the supremum of the $t$-statistic process under network dependence. Through the Bahadur linearization of the QR estimator, this problem reduces to approximating the distribution of the maximum of a high-dimensional vector formed by sums of network-dependent random variables \citep{fang2021high,chang2024central}. We then simulate the corresponding Gaussian maximum using  the estimated network-conditional covariance structure and use its empirical quantiles as critical values.
We assess the finite-sample performance of the proposed UCBs through Monte Carlo simulations. The results show that the bands achieve coverage probabilities close to the nominal levels, with coverage improving as the sample size increases.


\subsection*{Related Literature}

First, this paper contributes to the literature on causal inference under interference. One common framework is partial interference, which assumes that the sample consists of many independent clusters  and rules out interference across clusters; see \cite{hudgens2008toward,liu2016inverse,liu2019doubly} and \cite{qu2026semiparametric}. Within this framework, \cite{cheng2026nonparametric} study network quantile causal effects and develop efficient inference. However, any empirical settingsinvolve a single connected network rather than independent clusters. In such settings, researchers often impose a low-dimensional exposure mapping, under which potential outcomes depend on the high-dimensional treatment assignment vector only through exposure states such as own treatment status and the fraction of treated neighbors; see \cite{manski2013identification,aronow2017estimating,leung2020treatment,li2022random} and \cite{vazquez2023identification}. Existing work in this setting has primarily focused on average treatment and spillover effects. To  our knowledge, this is the first study of quantile effects in a single large network.

Among the existing literature, our paper is closely related to \cite{leung2020treatment}. This paper starts from an average structural function (ASF), which describes the mean potential outcome under a counterfactual treatment status and  number of treated neighbors. The ASF is identified using network experiments, and average exposure effects are defined as differences between ASFs evaluated at distinct exposure states. Our paper follows this structural perspective, which differs from design-based approaches to exposure effects \citep{leung2022causal,gao2025causal}.

Second, our paper contributes and builds on the large literature on  quantile treatment effects (QTE) and quantile regression (QR), where QR serves as the main  estimation and inference tool for QTE \citep{koenker2002inference,angrist2006quantile,belloni2019conditional}. QTEs have been studied extensively under the SUTVA framework in both experimental and observational settings \citep{chernozhukov2005iv,chernozhukov2006instrumental,chernozhukov2008instrumental,firpo2007efficient}. Recent work develops uniform inference for QTEs under a variety of (quasi-)experimental designs without interference \citep{chiang2019robust,qu2019uniform,zhang2020quantile,hou2026nonparametric}. We extend this line of work to network experiments.

Third, our paper builds on a seminal line of work on Gaussian approximation \citep{chernozhukov2013gaussian,chernozhukov2014anti,chernozhukov2014gaussian,chernozhukov2017central}. This literature provides Gaussian approximations to the laws of maxima of sums of  high-dimensional vectors, without the existence of a fixed limit distribution. While \Cref{subsec:nonparametric-QTE} establishes weak convergence of the quantile-effect estimators by embedding the observed network in a growing network sequence, that route requires the covariance kernel, which depends on the underlying graph sequence, to converge to a fixed limit uniformly over all quantile levels. This motivates the  uniform inference based on Gaussian approximation in \Cref{sec:dependent bootstrap}, which avoids characterizing a limit covariance kernel or an explicit limit distribution for the supremum statistic.  To accommodate network dependence, we draw on recent Gaussian approximation results for dependent high-dimensional data \citep{zhang2017gaussian,zhang2018gaussian,fang2021high,chang2024central}.

\paragraph{Organization of the paper.}
The remainder of the paper is organized as follows. \Cref{sec:framework} introduces the econometric model, defines the structural quantile effects, and establishes their identification in randomized network experiments. \Cref{subsec:nonparametric-QTE} proposes the QR estimator and establishes the weak convergence of the estimated quantile-effect process under conventional large-network asymptotics, which require the stabilization of the degree distribution and the network-dependent covariance structure. \Cref{sec:dependent bootstrap} develops UCBs based on Gaussian approximations conditional on the realized network and establishes their asymptotic validity without requiring either of these stabilization conditions. Finally, \Cref{section: Numerical_Examples} reports simulation results and presents an empirical application to the randomized savings-account experiment in Nepal \citep{prina2015banking}.

\paragraph{Notation.} 
For any $n\in\mathbb N^+$, let $[n]\equiv\{1,\ldots,n\}$.
We use
$\mathds{1}\{\cdot\}$ to denote the indicator function.
For $a,b\in\mathbb R$, let
$a\vee b=\max\{a,b\}$ and $a\wedge b=\min\{a,b\}$.
We use $\operatorname{Bern}(p)$ denote the Bernoulli distribution with $p \in (0,1)$. 
For a matrix $\boldsymbol{Q}$, define
$\|\boldsymbol{Q}\|_{\max}=\max_{j,k}|\boldsymbol{Q}_{jk}|$, and let
$\|\boldsymbol{Q}\|_{\mathrm{op}}$ denote its operator norm.
We write $\boldsymbol{Q}\succeq0$ when $\boldsymbol{Q}$ is symmetric and positive semidefinite.  For a network with adjacency matrix $\boldsymbol A_n \in \mathbb{R}^{n \times n}$, let
$\ell_{\boldsymbol A_n}(i,j)$ denote the shortest path 
between $i$ and $j$, with
$\ell_{\boldsymbol A_n}(i,i)=0$.
We write
$\mathbb P_{\boldsymbol A_n}(\cdot)
=\mathbb P(\cdot|\boldsymbol A_n)$,
and use $\mathbb E_{\boldsymbol A_n}$, $\operatorname{Var}_{\boldsymbol A_n}$ and
$\operatorname{Cov}_{\boldsymbol A_n}$ analogously. The abbreviations i.i.d.\ and PSD stand for ``independent and
identically distributed'' and ``positive semidefinite,'' respectively.

\section{Econometric Model}\label{sec:framework}

We consider a randomized controlled trial (RCT) with $n$ units. The population of $n$ units is embedded in a network, characterized by a symmetric adjacency matrix $ \boldsymbol{A}_n \in \{0,1\}^{n \times n}$, where $A_{ij}=1$ indicates a link between units $i$ and $j$. For each unit $i \in [n]$, let $Y_i \in \mathbb{R}$ denote the observed post-treatment outcome, and $D_i \in \{0,1\}$ denote the treatment assignment, where $D_i \sim \mathrm{Bern}(p)$ for some $p \in (0,1)$. The vector of treatment assignments is denoted by $\boldsymbol{D}_n \equiv (D_i)_{i=1}^n \in \{0,1\}^n$. Under the Neyman-Rubin causal model \citep{imbens2015causal}, the potential outcome for unit $i$, denoted by $Y_i(\boldsymbol{d})$, depends on the entire treatment assignment vector $\boldsymbol{d} = (d_i)_{i=1}^n \in \{0,1\}^n$. 

\subsection{Main Specification}\label{section: structural_quantile_3odel}

To restrict the dimensionality of arbitrary interference, we assume local interference via exposure mappings \citep{manski2013identification, aronow2017estimating}. Specifically, we assume that a unit's outcome is determined  by its own treatment status and those of its neighbors. Let $|N_i| = \sum_{j=1}^n A_{ij}$ denote the degree of unit $i$, and $T_i(\boldsymbol{d}, \boldsymbol{A}_n) = \sum_{j=1}^n A_{ij} d_j$ represent the number of its treated neighbors under the treatment assignment  $\boldsymbol{d}$. Moreover, let $W_i(\boldsymbol{d}, \boldsymbol{A}_n) = \left( d_i, T_i(\boldsymbol{d}, \boldsymbol{A}_n), |N_i| \right)$. When the context is clear, we may omit the explicit dependence on $\boldsymbol{A}_n$ and simply write $T_i(\boldsymbol{d})$ and $W_i(\boldsymbol{d})$.

\begin{assumption}\label{assumption:Interference}
For each unit $i \in [n]$, the potential outcome $Y_i(\boldsymbol{d})$ is determined by
\begin{equation}\label{eq: Interference}
Y_i(\boldsymbol{d}) = q \left( W_i(\boldsymbol{d},\boldsymbol{A}_n), U_i \right) = q \left(d_i, T_i(\boldsymbol{d}), |N_i|, U_i \right),
\end{equation}
where $q\left(\cdot\right)$ is an unknown function, and $U_i \in (0,1)$ captures unobserved heterogeneity.
\end{assumption}

\cref{assumption:Interference} imposes a structural restriction.  We implicitly rule out treatment-induced changes in the network. For the fixed pre-treatment network, the potential outcome 
\(Y_i(\boldsymbol d)\) depends on \(\boldsymbol d\)
only through $d_i$ and $T_i(\boldsymbol{d})$. 
The degree \(|N_i|\) is not itself a causal treatment. It is instead viewed as a 
conditioning/control variable, as the intervention does not change the interference network. By consistency, the observed outcome is $Y_i = q(D_i, T_i, |N_i|, U_i)$, where $T_i = \sum_{j=1}^n {A}_{ij} D_j$ is the number of treated neighbors. Following \cite{manski1993identification,manski2013identification} and  \cite{leung2020treatment}, we refer to  $(D_i, T_i)$ as the effective treatment. To interpret $q(\cdot)$ as a
structural quantile function (SQF), we  impose conditions under which \(U_i\) can be normalized as a latent rank in the degree-conditional outcome distribution.
{assumption:Interference}{assumption:normal}

\begin{assumption}\label{assumption:normal}
The  function $q: \{0,1\} \times \mathbb{N}^2 \times (0,1) \rightarrow \mathbb{R}$ and the unobservables \(\{U_i\}_{i=1}^n\) satisfy the following conditions:
\begin{assumpenum}
    \item\label{assumption:nor} Conditional on $|N_i|$, the unobservable $U_{i}$ are uniformly distributed, i.e., 
    $ U_i\mid |N_i| \sim \mathrm{Unif}(0,1)$.
    \item\label{assumption:mono} For any $w \equiv (d, t, l) \in \{0,1\} \times \mathbb{N}^2$ with $l \geq t$, the function
\(\tau \mapsto q(w,\tau)\) is  continuous and strictly increasing.
\end{assumpenum}
\end{assumption}

\cref{assumption:nor} is a natural normalization in view of the Skorohod
representation \citep{chernozhukov2005iv,chernozhukov2006instrumental}.
In randomized network experiments, the treatment assignment changes a unit's own treatment
status and the number of treated neighbors, but not the interference network.
Thus, \(U_i\) can be viewed as a latent rank determined by unobserved
factors that do not vary across counterfactual treatment assignments. Meanwhile, these underlying unobserved factors may be related to network
position, which motivates normalizing the latent rank conditional on degree
\(|N_i|\). \cref{assumption:mono} imposes monotonicity and
continuity of the function $q(w, \tau)$ in $\tau$.  Hence, \cref{assumption:normal} implies
that, conditional on degree \(l\), the potential outcome distribution under the effective treatment $(d, t) $ is characterized by \(q(w,\tau)\), where $w = (d,t,l)$. In particular, \(q(w,\tau)\)
is the \(\tau\)-quantile of the potential outcome for units with degree \(l\)
when the effective treatment is $(d,t)$.

\begin{remark}
Assumptions \ref{assumption:Interference} and \ref{assumption:normal} essentially assume a correctly specified exposure mapping and rank invariance \citep{doksum1974empirical,heckman1997making}. The latter can be relaxed to  rank similarity  without affecting the results of this paper \citep{chernozhukov2005iv,torgovitsky2015identification}.\footnote{Under rank similarity, the analogous versions of Assumptions \ref{assumption:RCT}, \ref{assumption:CE}, and \ref{assumption:local dependency} can be imposed on the family of exposure-specific ranks \(\left\{U_i(w): w\in\mathcal{W}\right\}\), where \(\mathcal{W}\) denotes the set of feasible exposure states.}  To formalize this relaxation, retain the assumption that   $Y_i(\boldsymbol d)=Y_i(\boldsymbol d')$ whenever $W_i(\boldsymbol d)=W_i(\boldsymbol d')$ for all $i$ and $\boldsymbol d,\boldsymbol d'\in\{0,1\}^n$. Rank similarity allows the latent rank  \(U_i(w)\) to vary across exposure states while preserving its marginal distribution. Accordingly, conditional on $W_i(\boldsymbol{d}) = w$, the potential outcome $Y_i(\boldsymbol{d})$  admits the following representation:
\[
Y_i(\boldsymbol d) =q\left(w,U_i(w)\right), \quad \text{where} \quad U_i(w)\sim\mathrm{Unif}\left(0,1\right).
\]

\end{remark}

\subsection{Parameters of Interest}

Building on the structural quantile model under network interference introduced in \cref{section: structural_quantile_3odel}, our main object of interest is the structural quantile effect. Specifically, conditional on the degree, we aim to measure how the quantiles of potential outcomes respond to counterfactual manipulations of the effective treatment.

For two exposure states $w=(d,t,l)$ and $w^\prime=(d^\prime,t^\prime,l)$ with $t,t^\prime \leq l$, we define the structural quantile effect at level $\tau \in (0,1)$ as:
\begin{equation}\label{eq: def_SQE}
 q(w,\tau) - q(w^\prime, \tau),
\end{equation}
where the SQF $q(\cdot)$ is defined in \cref{assumption:Interference}. Throughout this paper, the quantile effects compare only exposure states with the same degree, which avoids conflating exposure effects with degree-related heterogeneity.

Interest often centers on quantile treatment and spillover effects. First, the quantile direct
effect measures the effect of changing a unit's own treatment status while
holding fixed the number of treated neighbors and the degree: 
\[
q \left(1, t, l, \tau \right) - q \left(0, t, l, \tau \right).
\]
Second, the quantile spillover effect measures the effect of changing the
number of treated neighbors while holding fixed the unit's own treatment
status and degree:
\[
q\left(d, t, l, \tau\right) - q\left(d, t^\prime, l, \tau \right) .
\]
By varying \(\tau\), these estimands describe how direct treatment and
spillover effects differ across the outcome distribution.

\subsection{Identification}\label{subsection:identification}

When treatments are completely randomized,
structural quantile effects can be easily  
identified from network experiment data. Since identifying these effects is a direct implication of identifying the SQF, we focus on the latter in this subsection. We first require  that the treatment vector $\boldsymbol{D}_n$ is independent of both the interference network $\boldsymbol{A}_{n}$ and the latent ranks $\boldsymbol{U}_{n} \equiv (U_i)_{i=1}^n$.

\begin{assumption}\label{assumption:RCT} 
    The treatment assignments $D_i \overset{\mathrm{i.i.d.}}{\sim} \mathrm{Bern}(p)$  for a fixed $p \in \left(0,1\right)$. Moreover, $\boldsymbol{D}_n$ is independent of both  $\boldsymbol{A}_{n}$ and $\boldsymbol{U}_{n}$, that is, $\boldsymbol{D}_n \ \indep \left(\boldsymbol{A}_{n}, \boldsymbol{U}_{n}\right)$.
\end{assumption}

\cref{assumption:RCT} formalizes the Bernoulli randomized network
experiment \citep{cai2015social,prina2015banking,carter2021subsidies}.
The network is treated as a pre-treatment object, and the treatment assignment
vector is assumed to be independent of both the interference network and the
latent individual ranks. As a result, given the pre-treatment network,
variation in a unit's own treatment status and in the number of treated neighbors
is generated solely by the randomized assignment. In particular, the probability
that a unit realizes a specific effective treatment is known from the experimental
design.

\cref{assumption:RCT} rules out selection into treatment, but still allows
latent ranks to be correlated with the pre-treatment network. To allow for such
network-related unobserved heterogeneity while preserving identification, we
impose the following conditional exogeneity condition.

\begin{assumption}\label{assumption:CE}
    For any unit $i\in [n]$, $U_{i} \ \indep \ \boldsymbol{A}_{n}\mid \left|N_i\right|$.
\end{assumption}

\cref{assumption:CE} allows latent ranks to be related to network position
through degree, but requires that, conditional on degree, the  network carries no further information about \(U_i\). For example, if
\(U_i\) captures unobserved sociability or communication skill, more sociable
individuals may have more friends, so \(U_i\) may be correlated with
\(|N_i|\). The assumption permits this dependence through the number of friends,
while ruling out additional dependence on who those friends are, conditional on
the number of friends. Together, \cref{assumption:RCT,assumption:CE} imply that 
 $U_i \ \indep \left(D_i,T_i\right) \mid |N_i|$,
which is the key condition for identifying the SQFs
from observed conditional outcome distributions.

For a prespecified exposure state $w=(d,t,l)$, the SQF $q(w,\tau)$ can be
identified only if the observed network contains units with degree $l$. We therefore impose the following degree support condition.

\begin{assumption}\label{assumption:overlap} For the prespecified common degree $l \in \mathbb{N}$, 
there is a constant $c_l >0$ such that 
\[
\liminf_{n\rightarrow \infty} \ \frac{1}{n} \sum_{i=1}^n \mathds{1}\left\{ |N_i| = l \right \} > c_{l} , \quad \text{a.s.}
\]
\end{assumption}

Under \cref{assumption:overlap}, once the network contains units with
degree $l$, the Bernoulli design assigns positive probability to every feasible
effective treatment $(d,t)$ with $d\in\{0,1\}$ and $0\leq t\leq l$.

\begin{proposition}\label{prop:id}
Suppose \cref{assumption:Interference,assumption:normal,assumption:RCT,assumption:CE}
hold. For all $\tau \in (0,1)$ and any feasible exposure state $w=(d,t,l)$ whose degree $l$ satisfies
\cref{assumption:overlap},  the SQF $q(w,\tau)$ is identified by
\[
\mathbb P\left[Y_i\le q(w,\tau)\mid W_i=w\right]=\tau.
\]
\end{proposition}

\section{Estimation and Asymptotic Theory}\label{subsec:nonparametric-QTE}

This section studies estimation of the structural quantile effects
over a range of quantile levels. Let $\mathcal T$ be a compact subset of $(0,1)$,
and fix two feasible exposure states \(w\) and \(w'\) with the same degree.
Recall from \cref{eq: def_SQE} that the quantile-effect function is defined as
\begin{equation}\label{eq: QTE_Def}
q(\tau) \equiv q\left(w,\tau\right)-q\left(w',\tau\right), \quad \tau\in\mathcal T .
\end{equation}

We propose to estimate the SQF by a conditional quantile regression.  Let  $\mathds{1}_i(w^\dagger )=\mathds{1}\{W_i=w^\dagger \}$ where $W_i \equiv (D_i, T_i, |N_i|)$, and let the check function be given by
$\rho_\tau(u)=u\left(\tau-\mathds{1}\{u<0\}\right)$.  Following \cref{prop:id}, we estimate $q_w(\tau) \equiv  q(w,\tau)$ by
\begin{equation}\label{equation: Normalized_estimator}
\widehat{q}_w(\tau)  = \mathop{\argmin}_{q \in \mathbb{R} }  \sum_{i=1}^n   \mathds{1}_i (w) \rho_\tau( Y_i - q ) .
\end{equation}
Similarly, $q_{w^\prime}(\tau) \equiv q \left(w^\prime,\tau\right)$ can be estimated by $\widehat{q}_{w^\prime}(\tau)$. Recall $q(\tau) =  q_w( \tau ) - q_{w^\prime}( \tau )$, and we then estimate  $q(\tau)$ by 
\[
\widehat{q}\left( \tau \right) = \widehat{q}_w( \tau ) - \widehat{q}_{w^\prime}( \tau ).
\]

The asymptotic analysis for $\widehat{q}\left( \tau \right)$ is nonstandard because effective treatments are
network-induced and latent ranks may be locally dependent. We therefore derive the
limit distribution of the estimated quantile-effect process under local
dependence of latent ranks, stabilization of the empirical degree distribution,
and high-level smoothness and
limit conditions stated below.

\subsection{Bahadur Representation}

Throughout the remainder of this paper, we maintain the assumptions in
\cref{sec:framework}, including \cref{assumption:overlap} for the
exposure states \(w\) and \(w'\). We impose the following regularity
conditions.

\begin{assumption}\label{ass:degree}
There is a universal integer $K_{\max} > 0$ such that $ \max_{1\le i\le n} |N_i| \leq   K_{\max}$, almost surely. 
\end{assumption}

Under the assumptions in \cref{sec:framework}, the conditional distribution
of \(Y_i\) given \(W_i=w^\dagger\) is common across units \(i\). For
\(w^\dagger\in \left\{w,w' \right\}\), let \(F_{w^\dagger}(\cdot)\) and
\(f_{w^\dagger}(\cdot)\) denote this common conditional distribution function
and density, respectively.

\begin{assumption} \label{assumption: density} 
For each $w^\dagger \in \{w, w^\prime\}$, the conditional density $f_{w^\dagger}(\cdot)$ is bounded, bounded away from zero, and Lipschitz continuous on an open neighborhood of $\{ q_{w^\dagger}(\tau) : \tau \in \mathcal{T} \}$.
\end{assumption}

To accommodate network-correlated unobservables while keeping the dependence tractable, we allow latent ranks to exhibit local network dependence,
as in \citet{leung2020treatment} and \citet{viviano2025policy}.

\begin{assumption}\label{assumption:local dependency}
Conditional on the network $\boldsymbol{A}_n$, the sets of random variables
$\left\{U_i : i \in B_1 \right\}$ and $\left\{U_j : j \in B_2 \right\}$ are independent for any two disjoint
sets $B_1, B_2 \subset [n]$ such that
$\ell_{\boldsymbol{A}_n}(i,j) > 2$ for any $(i,j) \in B_1 \times B_2$.
\end{assumption}

\begin{remark}
\cref{assumption:local dependency} captures the idea that nearby units may
share unobserved local shocks or information. 
Conditional on the network, latent ranks may be dependent for
directly linked units and for units sharing common neighbors, but are
independent for units whose network distance exceeds two. Our results continue to hold for any fixed local dependence radius $r \geq 3$.\footnote{The estimation procedure is unchanged. For the construction of the UCBs, it suffices to redefine $\boldsymbol{\Omega}_n(i,j)=\mathds{1}\{\ell_{\boldsymbol{A}_n}(i,j)\leq r\}$, while leaving the remainder of the inference procedure unchanged.}
\end{remark}

We next establish the uniform Bahadur representation of the  estimated quantile-effect function $\widehat{q} \left(\tau\right)$ over $\tau \in \mathcal{T}$. Define
\begin{equation}\label{eq:influence_function}
\widetilde{\psi}_i(\tau)
=
\frac{
\mathds{1}_i(w)
\bigl(\tau-\mathds{1}\{Y_i\le q_w(\tau)\}\bigr)
}{
\pi_n(w)\,
f_w\bigl(q_w(\tau)\bigr)
}
-
\frac{
\mathds{1}_i(w')
\bigl(\tau-\mathds{1}\{Y_i\le q_{w'}(\tau)\}\bigr)
}{
\pi_n (w^\prime)\,
f_{w'}\bigl(q_{w'}(\tau)\bigr)
},
\end{equation}
where $\pi_n( w^\dagger) = n^{-1} \sum_{i=1}^n \mathbb{P}\big(W_i = w^\dagger | \boldsymbol{A}_n  \big)$ for $w^\dagger \in \left \{ w, w^\prime  \right\}$. For simplicity, the dependence of $\widetilde{\psi}_i(\tau)$ on $n$ is suppressed, but it depends on $\boldsymbol{A}_n$ through
$\pi_n(w)$ and $\pi_n(w')$.

\begin{lemma}\label{lemma:Bahadur representation}
Suppose \cref{ass:degree,assumption: density,assumption:local dependency}   hold. Conditionally on  $\{ \boldsymbol{A}_n \}_{n\geq 1}$, it follows that 
\[
\sup_{\tau \in \mathcal{T}}\left|\sqrt{n} \left( \widehat{q}(\tau) -  q(\tau) \right) - \frac{1}{\sqrt{n} } \sum_{i=1}^n \widetilde{\psi}_i(\tau) \right| = O_P\left(r_{\mathrm{B},n}\right),
\]
where $r_{\mathrm{B},n} = n^{-1/8}(\log n)^{1/4}$.
\end{lemma}

\subsection{Asymptotic Normality}\label{section: nonpara-estimation}

This subsection studies the weak convergence of  $\{ \widehat q(\tau):\tau \in \mathcal{T} \}$.
The Bahadur representation established in
\cref{lemma:Bahadur representation} alone does not guarantee convergence to a fixed Gaussian process, because the scores $\widetilde{\psi}_i(\tau)$ and their covariance structure may vary along the network sequence $\{ \boldsymbol{A}_n \}_{n\geq1}$.
We therefore impose additional stabilization conditions on the empirical degree distribution and the network-conditional covariance kernel.

\begin{assumption}\label{assumption: DD_limit}
There exists a probability measure $\mu$ on $ \{0,1,\ldots, K_{\max}  \}$ such that, for each
$0\leq k \leq K_{\max}$,
\[
\frac{1}{n}\sum_{i=1}^n \mathds{1}\{|N_i|=k \}
\overset{ \mathrm{a.s.} }{\longrightarrow}
\mu(k).
\]
Moreover, for the common degree \(l\) of the exposure states \(w\) and \(w'\),
we assume \(\mu(l)>0\).
\end{assumption}

Under \Cref{assumption:RCT,assumption: DD_limit},   for $w^\dagger = \big(d^\dagger, t^\dagger, l \big)$,  both $\frac{1}{n}\sum_{i=1}^n \mathds{1}_i(w^\dagger)$ and $\pi_n \big(w^\dagger\big)$ converge to $\pi(w^\dagger)$ in probability, where
\begin{equation}\label{eq: pi_w}
\pi\big(w^\dagger\big) = \mu(l) \times \binom{l}{t^\dagger} p^{t^\dagger+d^\dagger}(1-p)^{l-t^\dagger+1-d^\dagger} > 0.    
\end{equation}
We further define the score $\psi_i(\tau)$ by substituting $\pi(w^\dagger)$ for $\pi_n(w^\dagger)$ in \eqref{eq:influence_function}, for $w^\dagger \in \{w, w^\prime\}$. Combining the result $\left|\pi_n(w^\dagger) - \pi(w^\dagger)\right| = o_P(1)$ with
\cref{lemma:Bahadur representation} yields the following corollary.

\begin{corollary}\label{corollary: fixed_Representation}
Suppose \cref{ass:degree,assumption: density,assumption:local dependency,assumption: DD_limit}   hold. Conditionally on  $\{ \boldsymbol{A}_n \}_{n\geq 1}$, it follows that 
\[
\sup_{\tau \in \mathcal{T}}\left|\sqrt{n} \left( \widehat{q}(\tau) -  q(\tau) \right) - \frac{1}{\sqrt{n} } \sum_{i=1}^n \psi_i(\tau) \right| = o_P(1).
\]
\end{corollary}

To derive a fixed Gaussian limit for $\sqrt{n}\left(\widehat{q}\left(\tau\right) - q(\tau) \right)$,
we require the network-dependent covariance kernel to stabilize along the
graph sequence, analogously to Assumption 6 of \cite{leung2020treatment}. The  network-conditional  covariance kernel  $ \mathbb V_{\boldsymbol A_n}(\tau,\tau^\prime)$ is defined as
\begin{equation}\label{eq: con_variance_kernel}
 \mathbb V_{\boldsymbol A_n}(\tau,\tau^\prime)
\equiv \frac{1}{n}
\mathrm{Cov}\left(\sum_{i=1}^n\psi_i(\tau),\sum_{i=1}^n\psi_i(\tau^\prime)\Big|\boldsymbol A_n\right), \quad \tau, \tau^\prime \in \mathcal{T}.   
\end{equation}

\begin{assumption}\label{ass:final_variance}
There exists a nonrandom, symmetric, and PSD kernel $\Sigma :\mathcal T\times \mathcal T \to \mathbb R$ such that
\[
\sup_{\tau,\tau' \in \mathcal T}
\left|
\mathbb V_{\boldsymbol A_n}(\tau,\tau')
-
\Sigma (\tau,\tau')
\right|
=o_P(1).
\]
\end{assumption}

\begin{proposition}\label{proposition:NP_quantile_process_convergence}
Suppose the conditions of \cref{corollary: fixed_Representation} hold. Under the additional \cref{ass:final_variance}, we have 
\begin{equation}\label{eq:limiting_distribution}
      \sqrt{n} \left( \widehat{q}\left(\tau\right) -  q\left(\tau\right) \right) \rightsquigarrow \mathbb{G}(\tau) \quad \mathrm{in} \ \ell^\infty(\mathcal{T}),
  \end{equation}
  where $\mathbb{G}(\cdot)$ is a tight, mean-zero Gaussian process with covariance kernel $\Sigma(\tau, \tau')$. 
\end{proposition}

\begin{remark}
\cref{ass:final_variance} is used mainly to derive \eqref{eq:limiting_distribution}, which establishes the limiting distribution of the quantile process $\sqrt{n} \left( \widehat{q}(\tau) - q(\tau) \right)$.
\end{remark}

\cref{proposition:NP_quantile_process_convergence} derives a Gaussian limit under the large-network asymptotics of \cite{leung2020treatment}, where the observed network is embedded into a growing sequence of graphs. This fixed  Gaussian limit result requires
the network sequence to stabilize in two respects. In particular, \cref{assumption: DD_limit} requires
the empirical degree distribution to converge along the graph sequence,  while \cref{ass:final_variance} assumes uniform convergence of the
network-conditional covariance kernel of the linearized quantile process.  These conditions are restrictions on the underlying graph sequence and
cannot be directly verified from a single observed network.

\section{Uniform Inference}\label{sec:dependent bootstrap}

This section  develops the uniform inference for the quantile-effect function. In particular, for the two fixed exposure states $w$ and $w'$, we construct uniform confidence bands (UCBs) for $q(\tau) \equiv q_w(\tau) - q_{w'}(\tau)$ for $\tau \in \mathcal{T}$.  We do not rely on the fixed Gaussian limit shown in
\cref{proposition:NP_quantile_process_convergence}, we  approximate the distribution of the studentized
supremum statistic by the maximum of a Gaussian vector
with an estimated network-dependent covariance \citep{chernozhukov2013gaussian,chernozhukov2014anti,chernozhukov2014gaussian}. Consequently, our procedure does not require the stabilization of either the
empirical degree distribution or the network-conditional covariance kernel
imposed in \cref{assumption: DD_limit,ass:final_variance}.

A standard i.i.d. multiplier bootstrap is inappropriate in our setting, since it perturbs the estimated score functions independently across units and therefore does not reproduce the local dependence.\footnote{A standard i.i.d. multiplier bootstrap would use 
$\widehat{\mathbb G}^{\prime}_{n}(\tau)
=
n^{-1/2}\sum_{i=1}^n \xi_i\widehat\psi_i(\tau)$, 
where $\xi_i \overset{\mathrm{i.i.d.}}{\sim} \mathrm{N}(0,1)$ and the multipliers are generated independently of the data.} As a result, it ignores the off-diagonal covariance terms between units that are directly connected or share common neighbors, which may  distort the resulting standard errors and critical values.

\subsection{Construction of UCBs}\label{subsec:UCB}

Let $\boldsymbol{\Omega}_n$ be an $n \times n$ matrix with entries $\boldsymbol{\Omega}_n(i,j) = \mathds{1}\{\ell_{\boldsymbol{A}_n}(i,j) \le 2\}$. For each $i\in[n]$, we estimate the score $\widetilde{\psi}_i(\tau)$ defined in \cref{eq:influence_function} by
	\begin{equation}\label{eq: feasible_score_UCB}
	\widehat{\psi}_i(\tau)
	=
	\frac{\mathds{1}_i(w)\bigl(\tau - \mathds{1}\{Y_i \le \widehat{q}_w(\tau)\}\bigr)}
	{\widehat{\pi}(w)\,\widehat{f}_w\bigl(\widehat{q}_w(\tau)\bigr)}
    -
	\frac{\mathds{1}_i(w')\bigl(\tau - \mathds{1}\{Y_i \le \widehat{q}_{w'}(\tau)\}\bigr)}
	{\widehat{\pi}(w')\,\widehat{f}_{w'}\bigl(\widehat{q}_{w'}(\tau)\bigr)},
	\end{equation}
    where $\widehat{\pi} \big(w^\dagger\big) = \frac{1}{n} \sum_{i=1}^n \mathds{1}_i( w^\dagger)$ and $\widehat{f}_{w\dagger}(\cdot)$ denote the estimator of conditional density of $Y_i$ given $W_i =w^\dagger$. We define the network-dependent covariance kernel estimator
$\widehat{\mathbb V}_n:\mathcal T\times\mathcal T\to\mathbb R$ as 
\begin{equation}\label{eq:variance_estimator}
\widehat{\mathbb{V}}_n(\tau, \tau^\prime) = \frac{1}{n}  \widehat{\boldsymbol{\psi}}(\tau)^\top  \boldsymbol{\Omega}_n \widehat{\boldsymbol{\psi}}(\tau^\prime),
\end{equation}
where $\widehat{\boldsymbol{\psi}}(\tau) = \bigl( \widehat{\psi}_1(\tau), \dots, \widehat{\psi}_n(\tau) \bigr)^{\top}$.  A 100$(1-\alpha)$\% Gaussian bootstrap UCB for $\{ q(\tau): \tau \in \mathcal{T} \}$ is constructed as
\[
\tau \mapsto \left[ \widehat{q}(\tau ) -  \frac{z_{1-\alpha}^\star \widehat{\sigma}_n(\tau) }{\sqrt{n}}, \ \widehat{q}(\tau ) +  \frac{ z_{1-\alpha}^\star \widehat{\sigma}_n(\tau) }{\sqrt{n}}   \right],
\]
where $\widehat{\sigma}_n^2(\tau) = \widehat{\mathbb{V}}_n(\tau, \tau)$, and $z_{1-\alpha}^\star$ is a bootstrap-based critical value to be defined below.

To compute the critical value, we approximate the supremum over
$\mathcal T$ by a finite grid. Let
$\mathcal{T}_n\equiv \{\tau_1,\ldots,\tau_{p_n}\}\subseteq\mathcal{T}$ 
be a pre-specified grid. We write $\boldsymbol\tau_n=(\tau_1,\ldots,\tau_{p_n})$  and define its mesh size $\delta_n$ of $\mathcal{T}_n$ by
\[
\delta_n
\equiv
\sup_{\tau\in\mathcal T}
\min_{1\le j\le p_n}
|\tau-\tau_j|.
\]
The number of grid points $p_n$ is required to increase with
$n$ so that $\delta_n\to0$. Let $\widehat{\boldsymbol V}_n\in\mathbb R^{p_n\times p_n}$ denote the
estimated covariance matrix on the grid, with $(i,j)$-th entry $\widehat{\mathbb{V}}_n(\tau_i,\tau_j)$. Let $\widehat{\boldsymbol \Gamma}_n =  \widehat{\boldsymbol{\Lambda}}_n^{-1} \widehat{\boldsymbol V}_n \widehat{\boldsymbol{\Lambda}}_n^{-1} $, where $\widehat{\boldsymbol{\Lambda}}_n = \mathrm{diag} \left( \widehat{\sigma}_n(\tau_1),\ldots, \widehat{\sigma}_n (\tau_{p_n})  \right) \in \mathbb{R}^{p_n\times p_n}$. Since $\widehat{\boldsymbol\Gamma}_n$ may not be positive semidefinite (PSD) in
finite samples, we replace it by a nearest correlation matrix. Specifically, let
\begin{equation}\label{eq: PS_programming}
 \widehat{\boldsymbol{ \Gamma}}^+_n \in \mathop{\argmin}_{  \boldsymbol{\Gamma} \in \mathbb{R}^{p_n \times p_n} } \left\{   \big \| \widehat{\boldsymbol \Gamma}_n - \boldsymbol{\Gamma}  \big \|_{\max}  : \boldsymbol{\Gamma} \succeq  0,  \mathrm{diag}( \boldsymbol{\Gamma} ) = 1  \right\}.    
\end{equation}
Conditional on the observed data and the realized network, draw a mean-zero
Gaussian vector
\[
\widehat{\boldsymbol{Z}}_n \equiv \big( \widehat{Z}_n(\tau_i): 1\leq i \leq p_n  \big) \sim \mathrm{N}\big(0, \widehat{\boldsymbol{ \Gamma}}^+_n\big).
\]
The Gaussian bootstrap maximal statistic over the grid $\mathcal{T}_n$ is defined as
\[
\widehat{T}_n^\star
=
\sup_{1\leq i \leq p_n}
\left|
\widehat{Z}_n(\tau_i)
\right| .
\]
Finally, the critical value is given by the conditional $(1-\alpha)$-quantile of $\widehat{T}_n^\star$:
\[
z_{1-\alpha}^\star
=
\inf
\left\{
s\in\mathbb R:
\mathbb P_\star\big(
\widehat{T}_n^\star\le s
\big)
\ge 1-\alpha
\right\},
\]
where $\mathbb P_\star$ denotes the conditional probability, given
$\left(Y_i,W_i \right)_{i=1}^n$ and $\boldsymbol A_n$.

\subsection{Validity of UCBs}

We now establish the validity of the  UCBs via Gaussian bootstrap.  Recall the  score $\widetilde{\psi}_i(\tau)$
defined in \cref{eq:influence_function}, and  define its network-conditional
covariance kernel by 
\begin{equation}\label{eq: network_conditional_covariance_kernel}
\widetilde{\mathbb V}_{\boldsymbol A_n}(\tau,\tau')
\equiv
\frac{1}{n}
\mathrm{Cov}\left(
\sum_{i=1}^n \widetilde\psi_i(\tau),
\sum_{i=1}^n \widetilde\psi_i(\tau')
\,\Big|\,\boldsymbol A_n
\right).
\end{equation}

By Assumptions in \cref{sec:framework} and \cref{assumption:local dependency}, it is easy to see that
$\mathbb{E} \big[\widetilde\psi_i(\tau)|\boldsymbol A_n \big]=0$, and 
\[
\widetilde{\mathbb V}_{\boldsymbol A_n}(\tau,\tau')
=
\frac{1}{n}
\sum_{i=1}^n
\sum_{\ell_{\boldsymbol A_n}(i,j)\le2}
\mathbb E\left[
\widetilde\psi_i(\tau)\widetilde\psi_j(\tau')
|\boldsymbol A_n
\right].
\]

\begin{assumption}
\label{assumption:density_est_convergence} There is a  interval $[a,b]$ such that  $\{q_{w^\dagger}(\tau) : \tau \in \mathcal{T}\} \subseteq [a, b]$ for $w^\dagger \in \{w, w^\prime\}$. Moreover,  there exist density estimators $\widehat{f}_{w^\dagger}$ for $w^\dagger \in \{w, w^\prime\}$ such that 
\[
\sup_{a \leq y \leq b} \big|\widehat{f}_{w^\dagger}(y) - f_{w^\dagger}(y)\big| = O_P \big( r_{f,n}  \big).
\]
\end{assumption}

\begin{remark}
We take an agnostic view on how the conditional density estimator $\widehat{f}_w$ is obtained, with \cref{assumption:density_est_convergence} imposing only high-level conditions on its convergence rate. Under sufficient regularity and an appropriate choice of bandwidth, estimators satisfying these rate conditions can be constructed using kernel density estimation.
\end{remark}

\begin{lemma}\label{lemma: feasible_variance_convergence}
Suppose \cref{ass:degree,assumption: density,assumption:local dependency,assumption:density_est_convergence} hold. Then, conditional on the \(\{ \boldsymbol{A}_n \}_{n\ge1}\),
\[
\sup_{\tau, \tau^\prime \in \mathcal{T}}\left|\widehat{\mathbb{V}}_n(\tau, \tau^\prime) - \widetilde{\mathbb V}_{\boldsymbol A_n}(\tau,\tau')\right| = O_P\left(\bar{r}_{\mathbb{V}, n} \right),
\]
where $\bar{r}_{\mathbb{V}, n} = r_{f,n} \vee n^{-1/4}$. 
\end{lemma}

\begin{assumption}\label{assumption: smallest_variance_diagonal}
Let $\sigma_n^2(\tau) = \widetilde{\mathbb V}_{\boldsymbol A_n} (\tau,\tau)$.
There are constants $c_\sigma, c_\sigma^\prime >0$ such that
 $c_\sigma \leq \sigma_n(\tau) \leq  c_\sigma^\prime$ for all $\tau \in \mathcal{T}$.
\end{assumption}

Recall that \(\bar{r}_{\mathbb{V}, n} \) is defined in
\cref{lemma: feasible_variance_convergence}, and define
\[
\eta_n
=
n^{-1/8}(\log n)^{1/4}
+
\sqrt{\delta_n\log(e/\delta_n)}
+
n^{-1/2}\log(e/\delta_n)
+
\delta_n ,
\]
where $\delta_n$ is the mesh size of the grid $\mathcal T_n$. The following theorem shows that the conditional distribution of the
Gaussian maximum $\widehat{T}_n^\star$ approximates that of the studentized
supremum statistic.

\begin{theorem}
\label{theorem: feasible_gaussian_approximation}
Suppose  \cref{ass:degree,assumption: density,assumption:local dependency,assumption:density_est_convergence,assumption: smallest_variance_diagonal} hold. Further assume that $\delta_n = o(1)$,  $\left(\log p_n\right)^7/n\to0$, and $
\bar{r}_{\mathbb{V}, n} \log^2 p_n + \eta_n \sqrt{\log p_n} = o(1)$. Then, conditional on \(\{ \boldsymbol{A}_n \}_{n\ge1}\),
\[
\sup_{s\in\mathbb R}
\left|
\mathbb P_{\boldsymbol A_n}
\left[
\sup_{\tau\in\mathcal T}
\left|
\frac{\sqrt n\left(\widehat q(\tau)-q(\tau)\right) }
{\widehat\sigma_n(\tau)}
\right|
\le s
\right] - \mathbb P_\ast
\big(
\widehat{T}_n^\star 
\le s
\big)
\right|
=o_P(1),
\]
where \(\mathbb P_\ast\) denotes probability conditional on $\left(Y_i, W_i\right)_{i=1}^n$ and $\boldsymbol A_n$.
\end{theorem}

\cref{theorem: feasible_gaussian_approximation} immediately yields the
asymptotic validity of the UCB constructed in
\cref{subsec:UCB}.

\section{Numerical Examples}\label{section: Numerical_Examples}

\subsection{Monte Carlo}

In this section, we examine the finite sample performance of the proposed estimators.  Treatment is independently assigned according to $D_{i} \sim \text{Bern}\left(0.5\right)$. 
The network $\boldsymbol{A}_n$ is generated following \cite{leung2020treatment}, with the parameter $\kappa$ set to $1.55$ to yield an expected degree of 3 in the simulated network.

Let $\left\{\nu_i \right\}_{i=1}^n$ be \text{i.i.d.} random variables drawn from $\mathrm{N}\left(0,1/4\right)$, and $V_i = \nu_i + \vert{}N_i\vert{}^{-1/2} \sum_{j=1}^{n} A_{ij} \nu_j$. We then set $U_i = \Phi\big(\sqrt{2}\, V_i \big)$, where $\Phi(\cdot)$ denotes the standard normal cdf, so that $U_i$ is marginally uniform on $(0, 1)$. The observed outcome $Y_i$ is generated by
\[
Y_i = \beta_0\left(U_i\right) + \beta_{\mathrm{D}}\left(U_i\right) D_i + \beta_{\mathrm{M}}\left(U_i\right) M_i,
\]
where $M_i = T_{i}/|N_i|$ is the fraction of treated neighbors, and $\beta(\tau) \equiv \left(\beta_0(\tau), \beta_{\mathrm{D}}(\tau), \beta_{\mathrm{M}}(\tau) \right)$ is given by
\[
\beta_0(\tau) = \Phi^{-1}(\tau)/\sqrt{2}  , \quad
\beta_{\mathrm{D}}(\tau) = \Phi\left(  \beta_0(\tau) \right),\quad
\beta_{\mathrm{M}}(\tau) = \frac{1.5 \exp\left( \beta_0(\tau)  \right)}{1 + \exp\left( \beta_0(\tau)   \right)}.
\]

\paragraph{Implementation.} We estimate and perform uniform inference for the quantile spillover effect $q(\tau) = q_w(\tau) -  q_{w^{\prime}}(\tau)$ over $\mathcal{T} = [0.1,0.9]$. Specifically, as described in \cref{subsec:nonparametric-QTE}, we compute the estimators $\widehat{q}\left(\tau\right)$ over a grid of $p_n = 30$ equally spaced quantile levels, i.e., $0.1 = \tau_0 < \tau_1 < \ldots < \tau_{p_n} = 0.9$. Following Section 2.2 of \cite{belloni2019valid}, the estimates $\widehat{f}_{w^{\dagger}}(\widehat{q}_{w^{\dagger}}(\tau_k))$ are computed as
\begin{equation*}
    \widehat{f}_{w^{\dagger}}(\widehat{q}_{w^{\dagger}}(\tau_k)) = \frac{\tau_k^+ - \tau_k^-}{\widehat{q}_{w^{\dagger}}(\tau_k^+) - \widehat{q}_{w^{\dagger}}(\tau_k^-)},
\end{equation*}
where $\tau_k^+ = \min\{\tau_k + h_n, 0.9\}$, $\tau_k^- = \max\{\tau_k - h_n, 0.1\}$, and  $h_n \propto n^{-1/5}$. We construct our UCBs for $\left\{q(\tau) : \tau \in \mathcal{T} \right\}$ following the procedure outlined in Section \ref{subsec:UCB}, with critical values computed via $B = 200$ replications.

In order to provide more details on the quantile effect estimates, we report 95\% pointwise confidence intervals (CIs), the UCB developed in \cref{subsec:UCB}, and a UCB based on the network-dependent wild bootstrap. First, given  the estimated covariance kernel $\widehat{\mathbb{V}}_n\left(\cdot, \cdot\right)$, the pointwise CI for $q(\tau)$ is constructed by 
\[
\left[\widehat{q}\left(\tau\right) - \frac{1.96\cdot \widehat{\sigma}_n(\tau)}{\sqrt{n}} , \ \widehat{q}\left(\tau\right) + \frac{1.96\cdot \widehat{\sigma}_n(\tau)}{\sqrt{n}} \right],
\]
where $\widehat{\sigma}_n^2(\tau) = \widehat{\mathbb{V}}_n(\tau, \tau)$.

Alternatively, we implement a network-dependent wild 
bootstrap following \cite{shao2010dependent}, \cite{conley2023bootstrap} and \cite{kojevnikov2021bootstrap}, referred to as ``DWB'' hereafter. This procedure introduces auxiliary multipliers whose covariance structure is designed to mimic the network dependence of the original data. The resulting UCB follows the similar construction as that in \cref{subsec:UCB}, using the same point estimates and pointwise standard errors, but differs in how the critical value is computed. Recall that the kernel matrix $\boldsymbol{\Omega}_n $  with 
	\(
	\boldsymbol{\Omega}_n(i,j) = \mathds{1}\{\ell_{\boldsymbol{A}_n}(i,j) \le 2 \}
	\), which may not be PSD. Following \cite{gao2025causal}, we adjust $\boldsymbol{\Omega}_n$  by truncating the negative eigenvalues. Specifically,    let $\boldsymbol{Q}_n \boldsymbol{\Phi}_n \boldsymbol{Q}_n^{\top}$ be the eigen-decomposition of $\boldsymbol{\Omega}_n$, and $\boldsymbol{\Omega}^+_n = \boldsymbol{Q}_n \max\{\boldsymbol{\Phi}_n,0\} \boldsymbol{Q}_n^{\top}$ be the eigenvalue-truncated version of $\boldsymbol{\Omega}_n$. Let $\boldsymbol{\xi}_n = \left(\xi_1,\ldots, \xi_n \right) \sim \mathrm{N}( 0, \boldsymbol{\Omega}_n^+ )$, and the bootstrap score process is defined as $ \widehat{ \mathbb{G}}_n^\star(\tau)
    =
    \frac{1}{\sqrt n}
    \sum_{i=1}^n
    \xi_i\widehat\psi_i(\tau)$,
where $\widehat\psi_i(\tau)$ is defined in \cref{eq: feasible_score_UCB}. Define \[
\widehat T_n^{\mathrm{DWB}}
    =
    \sup_{\tau\in\mathcal T_n}
    \left|
        \frac{\widehat{ \mathbb{G}}_n^\star(\tau) }
        {\widehat\sigma_n(\tau)}
    \right|,
\]
where $\mathcal{T}_n = \left\{  \tau_k: 1\leq k \leq p_n \right\}$.    Letting  $c_{1-\alpha}^{\star}$ denote the empirical $(1-\alpha)$-quantile of $\widehat{T}_n^{\mathrm{DWB}}$, the resulting uniform confidence band is given by
\[
    \tau
    \mapsto
    \left[
        \widehat q(\tau)
        -
        \frac{ c_{1-\alpha}^{\star} \widehat\sigma_n(\tau)}{\sqrt n},
        \ 
        \widehat q(\tau)
        +
        \frac{        c_{1-\alpha}^{\star}\widehat\sigma_n(\tau)}{\sqrt n}
    \right].
\]

\paragraph{Results.} We conduct 3,000 Monte Carlo (MC) replications and implement the proposed estimation and uniform inference procedures for
\( \left\{q(\tau): \tau \in [0.1,0.9]\right\}\).
\cref{fig:ucb_plot} plots the estimated quantile spillover-effect functions together with the true effect curves, 95\% pointwise CIs, and the 95\% UCBs proposed in \cref{subsec:UCB}. The three panels correspond to the quantile-effect functions $q(\tau) = q_w( \tau) - q_{w^\prime}(\tau)$  for $w^\prime=(0,0,3)$ with $w=(0,1,3)$, $(0,2,3)$, and $(0,3,3)$, respectively. As illustrated in \cref{fig:ucb_plot}, the true quantile-effect functions  are seen to lie inside our UCBs.\footnote{To avoid confirmation bias in the visual presentation, the empirical realization plotted in each sub-panel is independently selected by minimizing the absolute deviation of its mean integrated squared error (MISE) from the median MISE across 300 independent replications with $n=3,000$.} The results of this MC experiment are presented in \cref{tab:simultaneous_coverage_results}. The pointwise  CIs substantially under-cover the entire quantile-effect functions. While the UCBs constructed using the network-dependent wild bootstrap (DWB) tend to be conservative, our proposed UCBs achieve coverage rates that are much closer to the nominal levels, especially in larger samples. 

\begin{figure}[htbp]
    \centering
    \includegraphics[width=\textwidth]{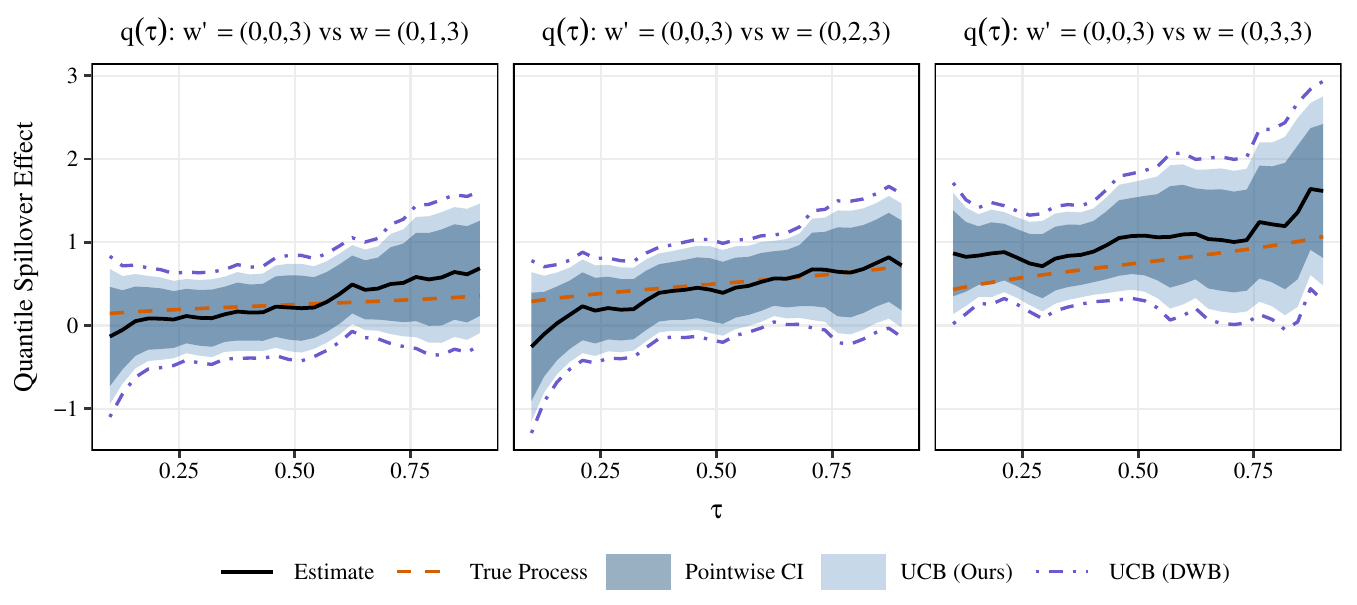}
    \caption{Estimated quantile spillover-effect functions for $n=3{,}000$, together with 95\% pointwise CIs and 95\% UCBs constructed using our  method and DWB. The three panels correspond to $w=(0,1,3)$, $(0,2,3)$, and $(0,3,3)$, respectively, with $w'=(0,0,3)$ serving as the baseline exposure state.}  
    \label{fig:ucb_plot}
\end{figure}

\begin{table}[htbp]
    \centering
    \begin{threeparttable}
        \caption{Coverage probabilities of pointwise CIs and UCBs for quantile effects.}
        \label{tab:simultaneous_coverage_results}
        \small
        \setlength{\tabcolsep}{5.2pt}     
        \renewcommand{\arraystretch}{1.15} 
        \begin{tabular}{cc ccc ccc ccc}
            \toprule
            & & \multicolumn{3}{c}{90\% Cover} & \multicolumn{3}{c}{95\% Cover} & \multicolumn{3}{c}{99\% Cover} \\
            \cmidrule(lr){3-5} \cmidrule(lr){6-8} \cmidrule(lr){9-11}
            Comparison & $n$ & PW & DWB & Ours & PW & DWB & Ours & PW & DWB & Ours \\
            \midrule
            
            \multirow{3}{*}{\makecell[c]{$(0,1,3)$ \\ vs \\ $(0,0,3)$}} 
            & 1,000 & 0.574 & 0.935 & 0.878 & 0.718 & 0.959 & 0.921 & 0.897 & 0.983 & 0.961 \\
            & 2,000 & 0.600 & 0.966 & 0.911 & 0.750 & 0.983 & 0.948 & 0.923 & 0.996 & 0.983 \\
            & 3,000 & 0.596 & 0.966 & 0.914 & 0.740 & 0.987 & 0.948 & 0.923 & 0.996 & 0.988 \\
            \midrule
            
            \multirow{3}{*}{\makecell[c]{$(0,2,3)$ \\ vs \\ $  (0,0,3)$}} 
            & 1,000 & 0.594 & 0.937 & 0.881 & 0.730 & 0.965 & 0.926 & 0.902 & 0.987 & 0.970 \\
            & 2,000 & 0.607 & 0.963 & 0.914 & 0.762 & 0.980 & 0.948 & 0.928 & 0.993 & 0.982 \\
            & 3,000 & 0.579 & 0.970 & 0.918 & 0.747 & 0.985 & 0.955 & 0.928 & 0.994 & 0.986 \\
            \midrule
            
            \multirow{3}{*}{\makecell[c]{$(0,3,3)$ \\ vs \\ $(0,0,3)$}} 
            & 1,000 & 0.592 & 0.906 & 0.849 & 0.727 & 0.937 & 0.890 & 0.873 & 0.970 & 0.946 \\
            & 2,000 & 0.643 & 0.965 & 0.921 & 0.795 & 0.979 & 0.952 & 0.931 & 0.994 & 0.983 \\
            & 3,000 & 0.630 & 0.973 & 0.929 & 0.785 & 0.989 & 0.961 & 0.937 & 0.997 & 0.990 \\
            \bottomrule
        \end{tabular}
        \begin{tablenotes}[flushleft]
            \item \textit{Note.} PW, DWB, and Ours denote pointwise CIs, UCBs using dependent wild bootstrap, and our UCBs, respectively. Coverage probabilities are calculated over $3,000$ replications. The average effective sample sizes, $\sum_{i=1}^{n}\mathds{1}_i( w)$, for $w=(0,0,3)$, $(0,1,3)$, $(0,2,3)$ and $(0,3,3)$ are $11.4$, $33.9$, $34.0$, and $11.3$ when $n=1,000$. These effective sample sizes increase to $25.3$, $75.6$, $75.4$, and $25.1$ for $n=2,000$, and to $36.7$, $109.4$, $109.3$, and $36.3$ for $n=3,000$. 
        \end{tablenotes}
    \end{threeparttable}
\end{table}

\paragraph{Linear QR.} Given the simplicity and interpretability of linear QR,
we next examine the finite-sample performance of the estimators and UCBs for the linear QR coefficient. Details of the linear
QR and the uniform inference procedure under
network interference are provided in \cref{appendix: linear}. Specifically, we
consider a correctly specified linear model of the form

\[
q(d,t,l,\tau)
=
\beta_0(\tau)
+
\beta_{\mathrm{D}}(\tau)d
+
\beta_{\mathrm{M}}(\tau) (t/l),
\]
where $\beta_{\mathrm{D}}(\tau)$ and $\beta_{\mathrm{M}}(\tau)$ represent the direct and spillover quantile effect coefficients, respectively.

Our goal is to estimate and construct UCBs for $\beta_{\mathrm{D}}(\tau)$ and $\beta_{\mathrm{M}}(\tau)$ over $\mathcal{T} = [0.1, 0.9]$. Further implementation details are provided in \cref{appendix: linear}. Following \cite{powell1986censored} and \cite{angrist2006quantile}, we estimate the Jacobian matrix $J(\tau)$ as $\widehat{J}_n(\tau) = (2 n h_n)^{-1} \sum_{i=1}^n \mathds{1}\{ \vert{}Y_i - X_i^\top \widehat{\beta}(\tau)\vert{} \leq h_n \} X_i X_i^\top$, where $X_i = (1, D_i, M_i)^\top$ and the bandwidth $h_n(\tau) \propto n^{-1/5}$ as suggested by  \cite{bofinger1975estimation}.

\cref{fig:ucb_coef_plots} plots the estimated linear QR coefficients $\widehat{\beta}_{\mathrm{D}}(\tau)$ and $\widehat{\beta}_{\mathrm{M}}(\tau)$ alongside their true counterparts $\beta_{\mathrm{D}}(\tau)$ and $\beta_{\mathrm{M}}(\tau)$, along with the 95\% pointwise CIs and 95\% UCBs for sample size $n=1,000$. \cref{tab:linear_coverage_rates} summarizes the MC simulation results evaluated over 3,000 replications. While the pointwise CIs exhibit severe under-coverage and the DWB UCBs tend to be conservative, our proposed UCBs deliver coverage rates that are much closer to the nominal levels across all sample sizes and significance levels.

\begin{figure}[htbp]
	\centering
    \includegraphics[width=0.84\textwidth]{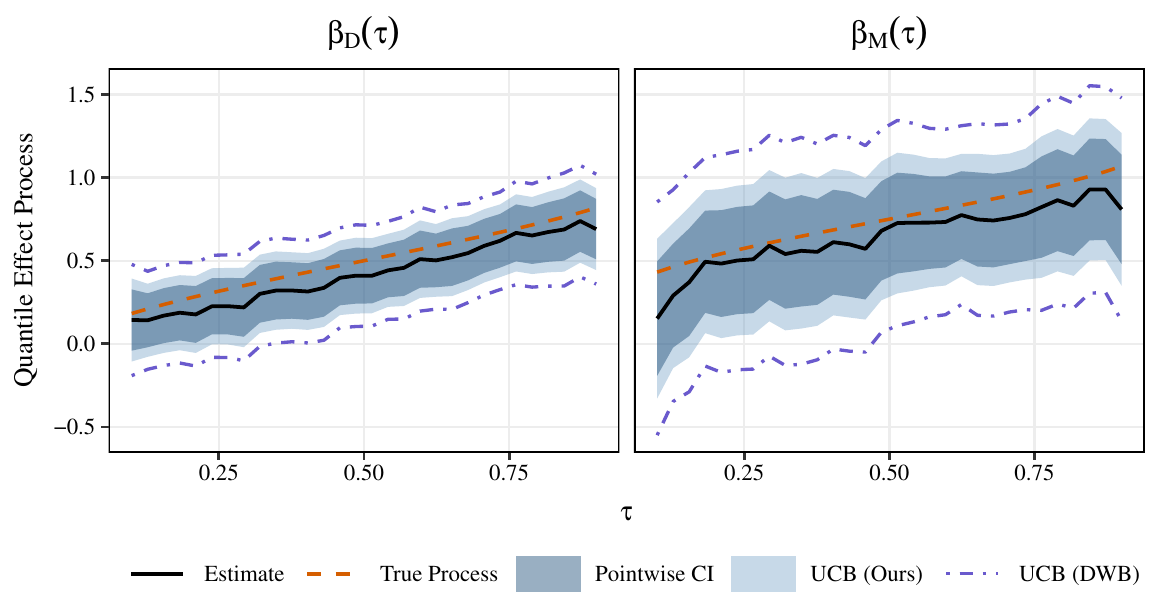}
	\caption{Estimated linear QR coefficients for $n=1,000$, together with 95\% pointwise CIs and 95\% UCBs constructed using our proposed method and DWB.
    }
	\label{fig:ucb_coef_plots}
\end{figure}

\begin{table}[htbp]
	\centering
    \begin{threeparttable}
        \caption{Coverage probabilities of pointwise CIs and UCBs for linear QR.}
        \label{tab:linear_coverage_rates}
        \small
        \setlength{\tabcolsep}{5.2pt}    
        \renewcommand{\arraystretch}{1.15}
        \begin{tabular}{cc ccc ccc ccc}
            \toprule
            & & \multicolumn{3}{c}{90\% Cover} & \multicolumn{3}{c}{95\% Cover} & \multicolumn{3}{c}{99\% Cover} \\
            \cmidrule(lr){3-5} \cmidrule(lr){6-8} \cmidrule(lr){9-11}
            Coef. & $n$ & PW & DWB & Ours & PW & DWB & Ours & PW & DWB & Ours \\
            \midrule

            \multirow{3}{*}{$\beta_{\mathrm{D}}(\cdot)$} 
            & 1,000 & 0.437 & 0.953 & 0.883 & 0.633 & 0.980 & 0.934 & 0.888 & 0.995 & 0.984 \\
            & 2,000 & 0.453 & 0.959 & 0.890 & 0.649 & 0.983 & 0.942 & 0.898 & 0.997 & 0.985 \\
            & 3,000 & 0.449 & 0.965 & 0.898 & 0.661 & 0.985 & 0.947 & 0.904 & 0.998 & 0.988 \\
            \midrule
            
            \multirow{3}{*}{$\beta_{\mathrm{M}}(\cdot)$} 
            & 1,000 & 0.521 & 0.925 & 0.883 & 0.693 & 0.964 & 0.932 & 0.904 & 0.993 & 0.980 \\
            & 2,000 & 0.523 & 0.941 & 0.897 & 0.708 & 0.970 & 0.945 & 0.922 & 0.993 & 0.985 \\
            & 3,000 & 0.526 & 0.936 & 0.895 & 0.710 & 0.965 & 0.941 & 0.916 & 0.993 & 0.982 \\
            \bottomrule
        \end{tabular}
\begin{tablenotes}[flushleft]
    \item \textit{Note.} PW, DWB, and Ours denote pointwise CIs, UCBs based on the dependent wild bootstrap, and our proposed UCBs, respectively. Coverage probabilities are estimated using $3{,}000$ Monte Carlo replications.
\end{tablenotes} 
    \end{threeparttable}
\end{table}

\subsection{Empirical Application}

To illustrate our methods, we revisit the randomized savings-account experiment studied by \cite{prina2015banking} and \cite{comola2021treatment}. In this experiment, female household heads in 19 slums surrounding Pokhara, Nepal's second largest city, were randomly offered fee-free savings accounts through public lotteries conducted within each community. The dataset consists of 915 households together with the network of regular financial support among these households. The network is constructed from survey responses on repeated financial exchanges: two households are linked if either household identifies a member of the other household as a regular financial-support partner. The resulting network is highly sparse, containing a total of $113$ undirected links and exhibiting a network density of $2.71 \times 10^{-4}$. Node degrees range from $0$ to $4$, with an average degree of $0.25$. 

We focus on three household outcomes: monetary assets, total assets, and education expenditures. For each original outcome $Y_i^{o}$, we use  $Y_i = \log(1+Y_i^{o})$  in our analysis. We estimate the following linear QR model:
\begin{equation}\label{empirical model2}
Q_\tau (Y_i\vert{}X_i) = \beta_0(\tau) + \beta_{\mathrm{D}}(\tau)D_i + \beta_{\mathrm{M}}(\tau) M_i + C_i^{\top}\gamma(\tau),    
\end{equation}
where  $X_i = (D_i, M_i, C_i)$, $D_i$ denotes the binary treatment whether household $i$ was offered a savings account, $M_i = T_i/|N_i|$ is the fraction of treated neighbors, and $C_i$ contains baseline household characteristics, village dummies, and a zero-neighbor indicator $\mathds{1}\{|N_i| = 0\}$ to properly account for isolated households. The parameter processes are evaluated over quantile domains $\mathcal{T}$: we set $\mathcal{T} = [0.10, 0.75]$ for monetary and total assets, and $\mathcal{T} = [0.40, 0.90]$ for education expenditures to avoid the point mass at zero concentrated at lower quantiles.

Figure \ref{fig:expenditure} reports the estimated coefficients $\beta_{\mathrm{D}}(\tau)$ and $\beta_{\mathrm{M}}(\tau)$, along with their 90\% pointwise CIs and UCBs constructed using the method in \cref{subsec:UCB}.
For monetary assets, the estimated  coefficient $\widehat{\beta}_{\mathrm{D}}(\tau)$ is positive  and generally declines with $\tau$. The 90\% UCB remains entirely above zero over
$\tau\in[0.1,0.75]$, suggesting that the offer of a savings account generated positive direct gains in monetary assets for households across the lower and middle quantiles. The estimated coefficients $\widehat{\beta}_{\mathrm{D}}(\tau)$ for total assets and educational expenditures exhibit a similar but less pronounced pattern, with positive estimates concentrated in the lower quantiles. Moreover, the 90\% UCBs for $\beta_{\mathrm{M}}(\tau)$ include zero
throughout the quantile ranges considered for all three outcomes, providing
no uniform evidence of spillover effects. The  wide bands reflect
substantial sampling uncertainty of
$\widehat{\beta}_{\mathrm{M}}(\tau)$, which may partly arise because many
households have no recorded financial-support links.

\begin{figure}
    \centering
    \includegraphics[width=\linewidth]{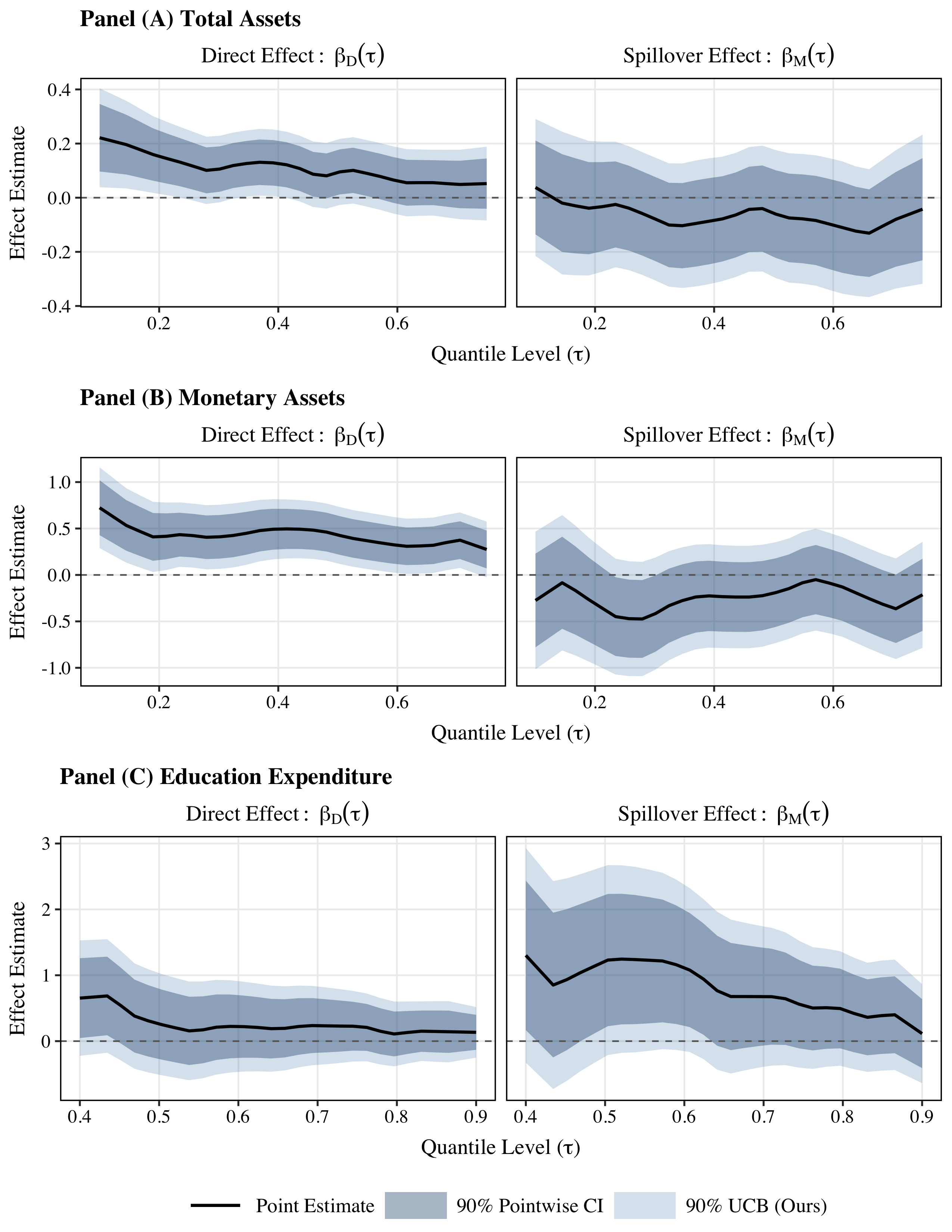}
    \caption{Estimated quantile regression coefficient processes with 90\% pointwise CI and UCBs.}
    \label{fig:expenditure}
\end{figure}


		

\section{Concluding Remarks}
This paper develops a framework for identifying and estimating quantile treatment and spillover effects in randomized experiments with network interference. Based on a structural quantile model, we construct QR estimators for these effects. Our main contribution is a novel uniform inference procedure using Gaussian approximations conditional on the realized network. By directly approximating the distribution of the studentized supremum statistic, we construct the UCBs without requiring the stringent stabilization conditions on the degree distribution and the network-dependent covariance structure typically needed for a fixed Gaussian limit.

Our analysis relies on Assumptions \ref{assumption:Interference} and \ref{assumption:local dependency}, that is, a correctly specified exposure mapping and conditional independence of the latent ranks for units whose network distance exceeds two. \cref{assumption:Interference} is needed for the causal interpretation. When the exposure mapping is misspecified, the QR estimator constructed in \cref{subsec:nonparametric-QTE} still compares conditional outcome quantiles between the specified exposure states, but these contrasts may no longer represent structural quantile effects. \cref{assumption:local dependency} is needed for the validity of our uniform inference in \cref{sec:dependent bootstrap}. If the outcome dependence extends beyond network distance two, the proposed UCBs  may be invalid.

However, \cref{assumption:local dependency} can be relaxed to allow for finite-range network dependence with any fixed radius, and our UCBs remain valid once the covariance-kernel estimator is adjusted to account for the expanded dependency graph. Another extension is to replace finite-range dependence with $\psi$-dependence on the QR process \citep{kojevnikov2021limit,leung2022causal}. Recent advances in high-dimensional Gaussian approximations for network-dependent data \citep{zhenggaussian2026} make this extension feasible. Finally, misspecification of the exposure mapping is a more fundamental issue, as it invalidates the structural interpretation of the estimand in our setting. Defining estimands that capture distributional heterogeneity robust to misspecified exposure mappings, alongside valid inference procedures, is left for future research.

\newpage

\appendix

\section{Linear QR under Network Interference}\label{appendix: linear}

In \cref{subsec:nonparametric-QTE}, we proposed a nonparametric estimator of the structural quantile function $q(w,\tau)$. This approach is flexible, but
the resulting estimates may have large finite-sample variance when some
exposure state cells contain relatively few observations.  To complement the
nonparametric approach, we therefore consider a parsimonious linear
specification, which may yield gains in finite-sample precision when correctly
specified.

Let $M_i=T_i/|N_i|$ denote the fraction of treated neighbors, with the
convention that $M_i=0$ when $|N_i|=0$. For an exposure state $w=(d,t,l)$, let
$m=t/l$ when $l>0$ and $m=0$ when $l=0$. For each $\tau\in\mathcal T$, suppose that the structural quantile function satisfies
\[
q(w,\tau)
=
\beta_0(\tau)
+
\beta_{\mathrm{D}}(\tau)d
+
\beta_{\mathrm{M}}(\tau)m.
\] 
Under this linear specification, $\beta_{\mathrm{D}}(\tau)$ is the quantile direct effect,
holding the fraction of treated neighbors fixed. For a unit with degree
$l>0$, one additional treated neighbor changes $m$ by $1/l$, so the
corresponding quantile spillover effect is $\beta_{\mathrm{M}}(\tau)/l$.

Let $X_i\equiv (1,D_i,M_i)^\top$ and
\({\beta}(\tau) = \bigl(\beta_0(\tau),\beta_{\mathrm{D}}(\tau),\beta_{\mathrm{M}}(\tau)\bigr)^\top \). The Koenker-Bassett regression estimator of $\beta(\tau)$ is 
\[
\widehat{\beta}(\tau)
\in
\mathop{\argmin}_{ b\in\mathbb R^3}
\sum_{i=1}^n
\rho_\tau\bigl(Y_i- X_i^\top b\bigr).
\]
We next construct a uniform confidence band for $\beta_{\mathrm{M}}(\tau)$. Let \(e_3=(0,0,1)\) selects the spillover-effect coefficient. The  estimation for $\beta_{\mathrm{D}}(\tau)$ is obtained by replacing $e_3$ with $e_2=(0,1,0)$ and is therefore omitted. By arguments analogous to those used in \cref{lemma:Bahadur representation}, the influence score for $\beta_{\mathrm{M}}(\tau)$ is
\begin{equation}\label{eq:infeasible score}
    \psi_{M,i}(\tau) = e_3 J(\tau)^{-1} X_i\left(\tau-\mathds 1\{Y_i\le X_i^\top \beta(\tau)\}\right),
\end{equation}
where $J(\tau)$ is the Jacobian matrix defined by
\[
J(\tau) = \lim_{n\to\infty} \frac{1}{n} \sum_{i=1}^n \mathbb{E} \left[  X_i X_i^\top f_{Y|X}\left(X_i^\top \beta(\tau) | X_i\right) \right]   .
\]
For $i\in[n]$, we estimate $\psi_{M,i}(\tau)$ by
\begin{equation}\label{eq:feasible score}
    \widehat{\psi}_{M,i}(\tau) = e_3 \widehat{J}_{n}(\tau)^{-1} X_i\left(\tau-\mathds 1\{Y_i\le X_i^\top \widehat{\beta}(\tau)\}\right).
\end{equation}
Following \citet{powell1986censored} and \citet{angrist2006quantile}, we estimate the Jacobian matrix $J(\tau)$ by
\[
\widehat{J}_n(\tau) = \frac{1}{2 n h_n} \sum_{i=1}^n \mathds{1}\{ \vert{}Y_i - X_i^\top \widehat{\beta}(\tau)\vert{} \leq h_n \} X_i X_i^\top,
\]
where $h_n$ is the corresponding bandwidth.

We then estimate the network-dependent covariance kernel for $\widehat{\beta}_{M}(\tau)$ by
\begin{equation}\label{eq:linear_variance_estimator_M}
\widehat{\mathbb{V}}_{M,n}(\tau, \tau^\prime) = \frac{1}{n}  \widehat{\boldsymbol{\psi}}_{M}(\tau)^\top  \boldsymbol{\Omega}_n \widehat{\boldsymbol{\psi}}_{M}(\tau^\prime),
\end{equation}
where $\widehat{\boldsymbol{\psi}}_{M}(\tau) = \bigl( \widehat{\psi}_{M,1}(\tau), \dots, \widehat{\psi}_{M,n}(\tau) \bigr)^\top$, and $\boldsymbol{\Omega}_n$ is a $n\times n$ matrix with 
	\(
	\boldsymbol{\Omega}_n(i,j) = \mathds{1}\{\ell_{\boldsymbol{A}_n}(i,j) \le 2 \}
	\).  A 100$(1-\alpha)$\% Gaussian bootstrap UCB for $\{ \beta_{M}(\tau): \tau \in \mathcal{T} \}$ is constructed as
\[
\tau \mapsto \left[ \widehat{\beta}_{M}(\tau ) -  \frac{z_{M,1-\alpha}^\star \widehat{\sigma}_{M,n}(\tau) }{\sqrt{n}}, \ \widehat{\beta}_{M}(\tau ) +  \frac{ z_{M,1-\alpha}^\star \widehat{\sigma}_{M,n}(\tau) }{\sqrt{n}}   \right],
\]
where $\widehat{\sigma}_{M,n}^2(\tau) = \widehat{\mathbb{V}}_{M,n}(\tau, \tau)$, and $z_{M,1-\alpha}^\star$ is a bootstrap-based critical value. It can be
computed by following the procedure in \cref{subsec:UCB}, replacing the
estimated covariance kernel $\widehat{\mathbb{V}}_n(\cdot,\cdot)$ with $\widehat{\mathbb{V}}_{M,n}(\cdot,\cdot)$ throughout.

\section{Proofs for Section \ref{sec:framework}}

\begin{proof}[Proof of \cref{prop:id}]

Fix $\tau\in(0,1)$ and a feasible exposure state $w=(d,t,l)$. Consider any
unit $i$ for which $w$ lies in the support of $W_i$. Suppose first that $U_i\mid W_i=w\sim \mathrm{Unif}(0,1)$, as established below. Then
\[
\begin{aligned}
\mathbb{P}\left[ Y_i \leq q(w, \tau) \mid W_i = w \right] & = \mathbb{P}\left[ q(W_i, U_i) \leq q(w, \tau) \mid W_i = w \right] \\
 & = \mathbb{P}\left[ q(w, U_i) \leq q(w, \tau) \mid W_i = w \right] \\ 
  & = \mathbb{P}[ U_i \leq \tau \mid W_i = w ] \\
  & = \tau.
\end{aligned}
\]
Since $\tau \mapsto q(w, \tau)$ is strictly increasing and continuous by \cref{assumption:mono}, the conditional distribution function of $Y_i \mid W_i = w$ is continuous and strictly increasing on its support. Hence, $q(w,\tau)$ is identified as the unique $\tau$-quantile of this conditional distribution.

Next, we show that $U_i \mid W_i=w \sim \mathrm{Unif}(0,1)$. For any $u \in (0,1)$, $l \in \mathbb{N}$, $\boldsymbol{d} \in \{0,1\}^n$, and $\boldsymbol{a} \in \{0,1\}^{n\times n}$, it follows that
\[
\begin{aligned}
& \mathbb{P}[U_i \leq u, \boldsymbol{A}_{n} =\boldsymbol{a}, \boldsymbol{D}_n = 
\boldsymbol{d} \mid |N_i| =l ]\\
= \ &  \mathbb{P}\left(\boldsymbol{D}_n = \boldsymbol{d} \right) \cdot \mathbb{P}[U_i \leq u, \boldsymbol{A}_{n} =\boldsymbol{a} \mid |N_i| =l ]  \\
= \ &  \mathbb{P}\left(\boldsymbol{D}_n = \boldsymbol{d} \right) \cdot \mathbb{P}[U_i \leq u \mid |N_i| =l ]  \cdot \mathbb{P}[\boldsymbol{A}_{n} =\boldsymbol{a} \mid |N_i| =l ]   \\
= \ &    \mathbb{P}[U_i \leq u \mid |N_i| =l ]  \cdot \mathbb{P}\left[\boldsymbol{D}_n = \boldsymbol{d},\boldsymbol{A}_{n} =\boldsymbol{a} \mid |N_i| =l \right].  \\
\end{aligned}
\]
The first and third equalities hold by \cref{assumption:RCT}, and the second equality follows from \cref{assumption:CE}. Thus, $U_i \indep \left(\boldsymbol{A}_{n}, \boldsymbol{D}_n \right) \mid |N_i|$. Since $(D_i,T_i)\in\sigma(\boldsymbol A_n,\boldsymbol D_n)$, for every $w=(d,t,l)$ in the support of $W_i$,
\[
\begin{aligned}
\mathds{P}\left[ U_i \leq u \mid  W_i = w \right] & =  \mathds{P}\left[ U_i \leq u \mid  D_i = d, T_i = t, |N_i| = l \right]  \\
&  = \mathds{P}\left[ U_i \leq u\mid  |N_i| = l \right] \\
& = u,
\end{aligned}
\]
where the last equality follows from \cref{assumption:nor}.
\end{proof}

\section{Proofs for Section \ref{subsec:nonparametric-QTE}} \label{appendix:proof for section 3}

\subsection{Preliminaries}\label{sec:preliminary}

In this subsection, we introduce notation that is essential for the proofs in Section 3. Let $\mathcal{A}_n$ denote the collection of $n \times n$ undirected adjacency matrices with zeros on the main diagonal. Since an undirected network is fully characterized by its adjacency matrix, we use the terms \textit{network, graph}, and \textit{adjacency matrix} interchangeably throughout.

\begin{definition}\label{def:proper_cover}
Given $\boldsymbol{A}\in\mathcal{A}_n$, a collection $\mathcal{C}_n=\{ \mathcal{C}_n(k): 1 \leq k \leq K\}$ of disjoint subsets of $[n] = \{1,\ldots,n\}$ is called a \textit{proper cover} of $\boldsymbol{A}$ if $\bigcup_{k=1}^K \mathcal{C}_n(k)=[n]$, and, for each $k$, $\boldsymbol{A}_{ij}=0$ for all distinct $i,j\in \mathcal{C}_n(k)$.
\end{definition}

\begin{definition}\label{def:2_hop}
For a given adjacency matrix $\boldsymbol{A}\in\mathcal{A}_n$, let $\boldsymbol{A}^{(2)}$ be defined such that $\boldsymbol{A}^{(2)}(i,i)=0$ for all $i$, and for $i \neq j$, $\boldsymbol{A}^{(2)}(i,j)=1$ if $i$ and $j$ are adjacent or share a common neighbor, and $0$ otherwise.
\end{definition}

\begin{definition}[Dependency graph]
A graph $\boldsymbol{G}$ with vertex set $[n]$ is called a dependency graph for the collection of random elements
$\{X_i\}_{i=1}^n$ if, for any two disjoint sets $S,T\subset[n]$ such that
no edge of $\boldsymbol{G}$ connects a vertex in $S$ to a vertex in $T$,
the $\sigma$-fields $\sigma(X_i:i\in S)$ and $\sigma(X_j:j\in T)$ are independent.
\end{definition}

Recall the interference network $\boldsymbol{A}_n \in \mathcal {A}_n$ introduced in \cref{sec:framework}. By \cref{assumption:Interference,assumption:RCT,assumption:CE,assumption:local dependency}, the graph $\boldsymbol{A}_n^{(2)}$ is a dependency graph for the random vectors $\left( Y_i, W_i, U_i  \right)_{i=1}^n$ conditional on $\boldsymbol{A}_{n}$.  Let $\Delta_n$ denote the
maximum degree of $\boldsymbol{A}_n^{(2)}$, i.e., $\Delta_n \equiv \max_{i \in [n]}\sum_{j=1}^n \boldsymbol{A}_n^{(2)}(i,j)$. For any proper cover $\{\mathcal{C}_k\}$ of $\boldsymbol{A}_n^{(2)}$,   $\left\{ \left(Y_i, W_i, U_i\right): i\in \mathcal{C}_k \right\}$ are mutually independent. Finally, if $\Delta_n = 0$, the dependency graph has no edges and $(Y_i, W_i, U_i)_{i=1}^n$ are conditionally independent. Hence, without loss of generality, we assume $\Delta_n \geq 1$. By the Hajnal-Szemer\'edi theorem \citep{kierstead2008short}, there exists a proper cover
$\{\mathcal{C}_1,\dots,\mathcal{C}_{\Delta_n + 1}\}$ of $\boldsymbol{A}_n^{(2)}$ such that
\begin{equation}\label{eq:HS_thm}
\sup_{1\le k\le \Delta_n + 1 }\left|\,|\mathcal{C}_k|-\frac{n}{\Delta_n + 1}\right|\le 1.
\end{equation}

Our asymptotic analysis is conducted conditional on the realized network sequence $\{\boldsymbol{A}_n\}_{n\ge 1}$, treated as fixed.  Randomness arises from the randomized treatment assignment and from the latent variables $\{U_i\}_{i=1}^n$, where $\{U_i\}_{i=1}^n$ may exhibit network-induced dependence and may be correlated with network characteristics. All stochastic orders and probability statements, including $O_P(\cdot)$, $o_P(\cdot)$, $\rightsquigarrow$,  are interpreted under the conditional law given $\{ \boldsymbol{A}_n\}_{n\ge 1}$.

\subsection{Proof of Lemma \ref{lemma:Bahadur representation}}

We first establish the uniform rate for the estimator of $q_{w^\dagger}(\tau)$. For any $w^\dagger = \left(d^\dagger, t^\dagger, l \right) \in \{0,1\} \times  \mathbb{N}^2$ with $t^\dagger \leq l$, Recall that $\pi_{n}(w^\dagger) = \frac{1}{n} \sum_{i=1}^n \mathbb{P}_{\boldsymbol{A}_n}[ W_i = w^\dagger  ]$. Since $D_i \overset{\mathrm{i.i.d.}}{\sim} \mathrm{Bern} \left(p \right)$ with $p \in (0,1)$, it follows that
\begin{equation}\label{eq:big_Pi}
\pi_{n}(w^\dagger)  = \underbrace{\binom{l}{t^\dagger} p^{t^\dagger + d^\dagger} (1-p)^{l - t^\dagger+ 1-d^\dagger}}_{ \equiv \omega\left(w^\dagger\right) } \left[ \frac{1}{n} \sum_{i=1}^n \mathds{1} \left( |N_i| = l \right) \right]. 
\end{equation}
By \cref{assumption:RCT,assumption:overlap}, we can conclude that there is a constant $c_{\pi}>0$ such that
\begin{equation}\label{eq:bounded_big_Pi}
    \liminf_{n\to\infty}\pi_{n}(w^\dagger) >  c_{\pi}>0.
\end{equation}

\begin{lemma}\label{lemma: Bahadur Representation-rate}
Suppose \cref{ass:degree,assumption: density,assumption:local dependency} hold. 
Then, conditional on $\left\{ \boldsymbol{A}_n \right\}_{n\geq1}$, for each $w^\dagger \in \{w, w^\prime \}$,  
\begin{equation}\label{eq: Bahadur Representation}
\sup_{\tau \in \mathcal{T}} \left|\sqrt{n}\left(\widehat q_{w^\dagger}(\tau)-q_{w^\dagger}(\tau)\right)  - \frac{1}{\sqrt{n}} \sum_{i=1}^n \widetilde{\phi}_{i}\big( w^\dagger, \tau \big) \right| =  O_P\left( r_{\mathrm{B},n} \right) ,    
\end{equation}
where $r_{\mathrm{B},n} =n^{-1/8} \left(\log n\right)^{1/4}$ and
\[
\widetilde{\phi}_{i}\big( w^\dagger, \tau \big) =  \frac{ \mathds{1}_i(w^\dagger) \left(\tau - \mathds{1}\{ Y_i \leq q_{w^\dagger}(\tau) \} \right) }{ f_{w^\dagger}\big(  q_{w^\dagger} (\tau)  \big) \pi_{n}(w^\dagger)  }.
\]
\end{lemma}

\begin{proof}[Proof of \cref{lemma: Bahadur Representation-rate}]

We prove \eqref{eq: Bahadur Representation} for $w^\dagger=w$; the argument for $w^\dagger=w^\prime$ is identical. The identity in \citet{knight1998limiting} gives
\begin{equation}\label{eq: u_n_definition}
\widehat u_n(\tau) \equiv \sqrt{n} \left(  \widehat{q}_w(\tau) - {q}_w(\tau) \right)  = \mathop{\argmin}_{u\in \mathbb{R}} L_n(u, \tau),
\end{equation}
where 
\[
\begin{aligned}
L_n(u, \tau) & = \sum_{i=1}^n \mathds{1}_i(w)  \left[ \rho_\tau(Y_i - q_w(\tau) - u/\sqrt{n} )  -  \rho_\tau( Y_i - q_w(\tau) ) \right].
\end{aligned}
\]
The objective admits the decomposition
\[
L_n(u, \tau) = - L_{1,n}(\tau) u + L_{2,n}(u, \tau),
\]
where $L_{1,n}(\tau)$ and  $L_{2,n}(u, \tau)$ are defined as 
\[
\begin{aligned}
L_{1,n}(\tau)  & = \frac{1}{\sqrt{n}}  \sum_{i=1}^n \mathds{1}_i(w) \left[\tau - \mathds{1}\{ Y_i \leq q_w(\tau) \} \right],\\
L_{2,n}(u, \tau) & = \sum_{i=1}^n \mathds{1}_i(w) \int_0^{u/\sqrt{n}}  \left[\mathds{1}\{ Y_i \leq q_w(\tau) + t \} - \mathds{1}\{ Y_i \leq q_w(\tau) \} \right] \mathrm{d}t.
\end{aligned}
\]


Let $\zeta_n(w,\tau)=f_w(q_w(\tau))\pi_{n}(w)$. \cref{eq:bounded_big_Pi} and \cref{assumption: density} together imply that there exist universal constants $C_\zeta>c_\zeta>0$ such that
\begin{equation}\label{eq: bounds_for_zeta_n}
c_\zeta \leq \inf_{n, \tau\in \mathcal{T}} \zeta_n(w, \tau) \leq \sup_{n, \tau\in \mathcal{T}} \zeta_n(w, \tau) \leq C_\zeta.    
\end{equation}
We first establish that, for every fixed $M>0$,
\begin{equation}\label{equation: second_term_convergence}
\sup_{\tau \in \mathcal{T}, |u|\leq M}\left|L_{2,n}(u, \tau) -  \frac{1}{2} \zeta_n(w,\tau) u^2  \right| = O_P\left(  n^{-1/4} \sqrt{\log n}  \right).    
\end{equation}
For each $u$ and $\tau$, let 
\begin{equation}\label{eq:R_i_u_tau}
R_i(u, \tau) = \int_0^{u/\sqrt{n}} \left[\mathds{1}\{ Y_i \leq q_w(\tau) + t \} - \mathds{1}\{ Y_i \leq q_w(\tau) \} \right] \mathrm{d}t.
\end{equation}
Write $L_{2,n}(u,\tau)$ as its conditional expectation given $\boldsymbol{A}_n$ plus a centered remainder:
\[
L_{2,n}(u, \tau) = \mathbb{E}_{\boldsymbol{A}_n}\left[L_{2,n}(u, \tau)   \right] + \underbrace{ L_{2,n}(u, \tau) - \mathbb{E}_{\boldsymbol{A}_n}\left[L_{2,n}(u, \tau)   \right]  }_{ \equiv \nu_n(u, \tau) }.
\]

\noindent \underline{\bf Step 1.}  First, we analyze $\mathbb{E}_{\boldsymbol{A}_n}\left[L_{2,n}(u, \tau)   \right]$. By Fubini's theorem and the law of iterated expectations, it follows that 
\begin{equation}\label{eq:L_2n}
\begin{aligned}
\mathbb{E}_{\boldsymbol{A}_n}[L_{2,n}(u,\tau)] = \sum_{i=1}^n \mathbb{P}_{\boldsymbol{A}_n}\left[W_i=w\right] \mathbb{E}\left[R_i(u, \tau) | \boldsymbol{A}_n, W_i = w \right].
\end{aligned}
\end{equation}
Moreover, \cref{assumption:RCT,assumption:CE} imply that the conditional density of $Y_i$ given $(W_i,\boldsymbol{A}_n)$ coincides with that given $W_i$. Consequently,
\begin{equation}\label{eq:Taylor_exp}
\begin{aligned}
 \mathbb{E}\left[  R_i(u, \tau) | \boldsymbol{A}_n, W_i =w \right]  & =  \mathbb{E}\left[R_i(u, \tau) | W_i = w \right] =   \int_0^{u/\sqrt{n}} \left[ F_w\left(q_w(\tau) + t\right) - F_w\left(q_w(\tau) \right) \right] \mathrm{d}t  \\
 & =   \int_0^{u/\sqrt{n}}  f_w\left(  q_w(\tau) \right) t  \mathrm{d}t + r_n(u, \tau) , \\
  & = \frac{1}{2}   f_w\left(  q_w(\tau) \right) n^{-1} u^2  + r_n(u, \tau), \\
\end{aligned}
\end{equation}
where the remainder term  $r_n(u, \tau)$ is given by 
\[
r_n(u, \tau) = \int_0^{u/\sqrt{n}} \int_0^t \left[ f_w(q_w(\tau) + s) - f_w(q_w(\tau)) \right] \mathrm{d}s \, \mathrm{d}t.
\]
Taylor's theorem and the Lipschitz continuity of $f_w(\cdot)$ implies that, for every fixed $M>0$,
\[
\sup_{\tau \in \mathcal{T}, |u| \leq M}\left| r_n(u, \tau) \right| = O\big( n^{-3/2} \big).
\]
Substituting this expansion into \cref{eq:L_2n} gives
\[
\begin{aligned}
\mathbb{E}_{\boldsymbol{A}_n}[L_{2,n}(u,\tau)] & = \frac{1}{n} \sum_{i=1}^n \mathbb{P}_{\boldsymbol{A}_n}\left[W_i=w\right]  \left[ \frac{u^2}{2}  f_w(q_w(\tau)) \right] + O\big(n^{-1/2}\big)  \\
& =  \pi_{n}(w)   \left[ \frac{u^2}{2}  f_w(q_w(\tau)) \right] +  O\big(n^{-1/2}\big),
\end{aligned}
\]
where the remainder is uniform over $(\tau,u)\in\mathcal T\times[-M,M]$. Equivalently,
\begin{equation}\label{eq:leading_rate}
    \sup_{\tau \in \mathcal{T}, |u|\leq M}\left|\mathbb{E}_{\boldsymbol{A}_n}\left[L_{2,n}(u, \tau) \right] - \frac{1}{2} \zeta_n (w,\tau) u^2 \right| =  O\big(n^{-1/2}\big).
\end{equation}

\noindent \underline{\bf Step 2.} Second, we show that
\begin{equation}\label{eq:st_2_convergence}
\sup_{\tau \in \mathcal{T}, |u| \leq M} | \nu_n(u, \tau) | = \sup_{\tau \in \mathcal{T}, |u| \leq M} \left| \sum_{i=1}^n Z_i(u, \tau) \right| = O_P\left(  \left(\Delta_n + 1 \right) n^{-1/4} \sqrt{\log n} \right),
\end{equation}
where $Z_i(u, \tau) = \mathds{1}_i(w) R_i(u, \tau) - \mathbb{E}_{\boldsymbol{A}_n}\left[\mathds{1}_i(w) R_i(u, \tau)  \right]$, and $R_i(u, \tau)$ is defined in \cref{eq:R_i_u_tau}. The argument parallels the proof of Theorem D.1 in \citet{viviano2025policy}. The Cauchy-Schwarz inequality gives
\begin{equation}\label{eq:squared_R}
\begin{aligned}
\left|R_i(u, \tau)\right|^2  & \leq \left|\frac{u}{\sqrt{n}} \int_0^{u/\sqrt{n}} \left| \mathds{1}\{ Y_i \leq q_w(\tau) + t \} - \mathds{1}\{ Y_i \leq q_w(\tau) \} \right| \mathrm{d}t \right| \leq u^2/n.
\end{aligned}
\end{equation}

Recall that $\boldsymbol{A}_n^{(2)}$ is a dependency graph of $ \left(Y_i, W_i, U_i \right)_{i=1}^n$, and  $\Delta_n$ is the maximum degree of $\boldsymbol{A}_n^{(2)}$. By the Hajnal-Szemer\'edi theorem and \cref{eq:HS_thm}, there exists a proper cover 
$\left\{ \mathcal{C}_1,\ldots,\mathcal{C}_{\Delta_n +1} \right\}$. Conditional on $\boldsymbol A_n$, the processes $\{\tau\mapsto Z_i(u,\tau):i\in\mathcal C_k\}$ are mutually independent within each class $\mathcal C_k$. Hence,
\begin{equation}\label{eq:target_nu_n}
\mathbb{E}_{\boldsymbol{A}_n}\left[ \sup_{\tau \in \mathcal{T}} | \nu_n(u, \tau) |  \right] \leq \sum_{k=1}^{\Delta_n + 1} \mathbb{E}_{\boldsymbol{A}_n}\left[ \sup_{\tau \in \mathcal{T}} \left| \sum_{i \in \mathcal{C}_k} Z_i(u, \tau) \right|  \right].
\end{equation}

For $(q,\epsilon)\in\mathbb R^2$, define $h_{q,\epsilon}(\bar y)=\int_0^\epsilon\left(\mathds{1}\{\bar y\le q+t\}-\mathds{1}\{\bar y\le q\}\right)\mathrm{d}t$. For fixed $M>0$, let
\[
\mathcal{H}_n = \left\{ (\bar{w},\bar{y} ) \mapsto \mathds{1} \{\bar{w} = w\}  h_{q, \epsilon}(\bar{y}) :  q \in \mathbb{R}, \left|\epsilon \right| \le M/\sqrt{n} \right\}.
\]
By Lemma 2.6.18 in \citet{vaart2023empirical}, $\mathcal H_n$ is a VC-subgraph class with $\mathrm{VC}(\mathcal H_n)\le3$. Since $R_i(u,\tau)=h_{q_w(\tau),u/\sqrt n}(Y_i)$, the class $\mathcal H_n^\prime\equiv\{Y_i\mapsto R_i(u,\tau):\tau\in\mathcal T,\ |u|\le M\}$ is also VC-subgraph.

Fix a color class $\mathcal C_k$. Conditional on $\boldsymbol A_n$, the variables $\{Z_i(u,\tau)\}_{i\in\mathcal C_k}$ are independent and zero-mean. Lemma 2.3.1 in \citet{vaart2023empirical} therefore yields
\begin{equation}\label{eq:Rad_process}
\begin{aligned}
\mathbb{E}_{\boldsymbol{A}_n}\left[ \sup_{\tau \in \mathcal{T}} \left| \sum_{i \in \mathcal{C}_k} Z_i(u, \tau) \right|  \right] & \leq  2 \mathbb{E}_{\boldsymbol{A}_n} \left[\mathbb{E}_\varepsilon \left[ \sup_{\tau \in \mathcal{T}} \left| \sum_{i \in \mathcal{C}_k} \varepsilon_i \mathds{1}_i(w) R_i(u, \tau)  \right| \right]  \right],
\end{aligned}
\end{equation}
where $\{\varepsilon_i\}_{i=1}^n$ are i.i.d. Rademacher variables
independent of the data, and $\mathbb{E}_{\varepsilon}$ denotes expectation
with respect to the  $\{\varepsilon_i\}_{i=1}^n$, conditional on the data and
$\boldsymbol A_n$. In order to bound the RHS of \cref{eq:Rad_process}, define
\begin{equation}\label{eq:sigma_nk}
\widetilde{\sigma}_{n,k} = \sup_{\tau \in \mathcal{T}, |u| \leq M } \sqrt{ \frac{1}{|\mathcal{C}_k|} \sum_{i \in \mathcal{C}_k} \left| R_i(u, \tau) \right|^2 }.
\end{equation}
The VC maximal inequality gives
\[
\begin{aligned}
\mathbb{E}_\varepsilon \left[ \sup_{\tau \in \mathcal{T}, |u|\leq M} \left| \sum_{i \in \mathcal{C}_k} \varepsilon_i \mathds{1}_i(w) R_i(u, \tau)  \right|\right]  & \le  K \sqrt{|\mathcal{C}_k|} \int_0^{\widetilde{\sigma}_{n,k}} \sqrt{ 2 \log (A/\epsilon) } \mathrm{d}\epsilon,
\end{aligned}
\]
where $A$ and $K$ are universal constants. Consequently, for all sufficiently large $n$,
\begin{equation}\label{eq:Chi_n_sum}
\begin{aligned}
\sum_{k=1}^{\Delta_n + 1}  \mathbb{E}_\varepsilon \left[ \sup_{\tau \in \mathcal{T}, |u|\leq M} \left| \sum_{i \in \mathcal{C}_k} \varepsilon_i \mathds{1}_i(w) R_i(u, \tau)  \right| \right]  & \leq  K  \sum_{k=1}^{\Delta_n + 1} \sqrt{|\mathcal{C}_k|}  \int_{0}^{\widetilde{\sigma}_{n,k}} \sqrt{2\log (A/\epsilon) 
} \mathrm{d}\epsilon\\
& \leq  2 K  \sum_{k=1}^{\Delta_n + 1} \sqrt{|\mathcal{C}_k|}  \widetilde{\sigma}_{n,k} \sqrt{2 \log \left(A/\widetilde{\sigma}_{n,k} \right)},
\end{aligned}
\end{equation}
where the last inequality follows from $\widetilde{\sigma}_{n,k}^{2}\le M^2/n$, \cref{eq:squared_R,lem:entropy-integral-bound}.

Let $\varphi(x)=x\sqrt{2\log(A/x)}$. Combining \cref{eq:target_nu_n,eq:Rad_process} and taking expectations in \cref{eq:Chi_n_sum} conditional on $\boldsymbol{A}_{n}$ gives, for all sufficiently large $n$,
\[
\begin{aligned}
\mathbb{E}_{\boldsymbol{A}_n}\left[ \sup_{\tau \in \mathcal{T}, |u|\leq M } | \nu_n(u, \tau) |   \right] & \leq  2\sum_{k=1}^{\Delta_n + 1}  \mathbb{E}_{\boldsymbol{A}_n} \left[ \mathbb{E}_\varepsilon \left[ \sup_{\tau \in \mathcal{T}, |u|\leq M } \left| \sum_{i \in \mathcal{C}_k} \varepsilon_i \mathds{1}_i(w) R_i(u, \tau)  \right| \right]  \right] \\
& \leq K \sum_{k=1}^{\Delta_n + 1} \sqrt{|\mathcal{C}_k|} \, \mathbb{E}_{\boldsymbol{A}_n} \left[ \widetilde{\sigma}_{n,k} \sqrt{2 \log (A/\widetilde{\sigma}_{n,k})}   \right] \\
& = K \sum_{k=1}^{\Delta_n + 1} \sqrt{|\mathcal{C}_k|} \ \mathbb{E}_{\boldsymbol{A}_n}\left[ \varphi( \widetilde{\sigma}_{n,k}  ) \right] \leq_{(1)} K \sum_{k=1}^{\Delta_n + 1} \sqrt{|\mathcal{C}_k|} \ \varphi\left( \mathbb{E}_{\boldsymbol{A}_n}\left[  \widetilde{\sigma}_{n,k}   \right] \right)\\
& \leq_{(2)} K \sum_{k=1}^{\Delta_n + 1} \sqrt{|\mathcal{C}_k|} \ \varphi\left( \sqrt{\mathbb{E}_{\boldsymbol{A}_n}[  \widetilde{\sigma}_{n,k}^2  ] } \right) \\
& = K  \sum_{k=1}^{\Delta_n + 1} \sqrt{|\mathcal{C}_k|}  \   \sqrt{  \mathbb{E}_{\boldsymbol{A}_n}[  \widetilde{\sigma}_{n,k}^2   ]  \log \left(A \big / \sqrt{  \mathbb{E}_{\boldsymbol{A}_n}[  \widetilde{\sigma}_{n,k}^2 ] } \right) }   \\
& \lesssim_{(3)}  \left(\Delta_n + 1 \right) n^{-1/4} \sqrt{\log n}.
\end{aligned}
\]
Inequality (1) follows from the concavity of $\varphi$ near zero. Inequality (2) uses the concavity of $x\mapsto\sqrt{x}$ and the monotonicity of $\varphi$ near zero, and inequality (3) follows from \cref{lemma: Variance_Color_K}. This proves \cref{eq:st_2_convergence}.

\vspace{0.1cm}

\noindent \underline{\bf Step 3.} Finally, combining \cref{eq:leading_rate} with \cref{eq:st_2_convergence} gives, for every fixed $M>0$,
\[
\sup_{\tau \in \mathcal{T}, |u|\leq M} \left| L_{2,n}(u, \tau) - \tfrac{1}{2} \zeta_n(w,\tau) u^2 \right| = O_P\left( (\Delta_n+ 1) n^{-1/4} \sqrt{\log n}  \right)=  O_P\left( r_n  \right),
\]
where $r_n=n^{-1/4}\sqrt{\log n}$, and the last equality follows from \cref{ass:degree}.\footnote{By definition, $\Delta_n\le K_{\max}^2$, so $\Delta_n$ is uniformly bounded almost surely.} Hence, for every fixed $M>0$,
\begin{equation}\label{eq:compact-uniform approximation}
\sup_{\tau \in \mathcal{T}, |u| \leq M} \left|  L_n(u, \tau) - Q_n(u,\tau) \right| = O_P\left( r_n  \right),   
\end{equation}
where  
\begin{equation}\label{eq:Q_n_def}
Q_n(u,\tau)
=
\frac{1}{2}\zeta_n(w,\tau)u^2-L_{1,n}(\tau)u.    
\end{equation}

Let $u_n(\tau)=\mathop{\arg\min}_{u\in\mathbb R}Q_n(u,\tau)$. Since $Q_n(\cdot,\tau)$ is quadratic,
$u_n(\tau)=\zeta_n(w,\tau)^{-1}L_{1,n}(\tau)$. By \cref{eq: bounds_for_zeta_n,lemma:L1n_uniform_bounded}, we have
$\sup_{\tau\in\mathcal T}|u_n(\tau)|=O_P(1)$. Moreover, \cref{lemma:uniform_argmin_localization} gives $\sup_{\tau\in\mathcal T}|\widehat u_n(\tau)|=O_P(1)$, where $\widehat u_n(\tau)=\sqrt n\{\widehat q_w(\tau)-q_w(\tau)\}$ is defined in \cref{eq: u_n_definition}.

Therefore, for any $\epsilon>0$, there exists a fixed $M<\infty$ such that, for all sufficiently large $n$, both $\sup_{\tau\in\mathcal T}|\widehat u_n(\tau)|\le M$ and $\sup_{\tau\in\mathcal T}|u_n(\tau)|\le M$ hold with probability at least $1-\epsilon$. On this event, the minimizing property of $\widehat u_n(\tau)$ gives
\[
L_n\left(\widehat u_n(\tau),\tau\right),
\le
L_n(u_{n}(\tau),\tau)
\]
uniformly over $\tau\in\mathcal T$. Applying \cref{eq:compact-uniform approximation}, we obtain
\[
Q_n\left(\widehat u_n(\tau),\tau\right)
\le
Q_n(u_{n}(\tau),\tau)+O_P(r_n),
\]
uniformly over \(\tau\in\mathcal T\). Since $Q_n(\cdot,\tau)$ is quadratic with unique minimizer $u_n(\tau)$,
\[
Q_n(u,\tau)-Q_n(u_{n}(\tau),\tau)
=
\frac{1}{2}\zeta_n(w,\tau)\left|u-u_{n}(\tau)\right|^2.
\]
Using \cref{eq: bounds_for_zeta_n} and the definitions of $u_n(\tau)$ and $\widehat u_n(\tau)$ gives
\[
\sup_{\tau\in\mathcal T}
\left|
\sqrt n\left(\widehat q_w(\tau)-q_w(\tau)\right)
-
\zeta_n(w,\tau)^{-1}L_{1,n}(\tau)
\right|
=
O_P\left(r_n^{1/2}\right).
\]
Since $r_n^{1/2}=n^{-1/8}(\log n)^{1/4}=r_{\mathrm B,n}$, this is the desired result.
\end{proof}

Next, we prove \cref{lemma:Bahadur representation}, which follows as a direct consequence of \cref{lemma: Bahadur Representation-rate}.

\begin{proof}[Proof of \cref{lemma:Bahadur representation}]
Recall $\widehat q(\tau)=\widehat q_w(\tau)-\widehat q_{w^\prime}(\tau)$, $q(\tau)=q_w(\tau)-q_{w^\prime}(\tau)$, and $\widetilde\psi_i(\tau)
=
\widetilde\phi_i(w,\tau)
-
\widetilde\phi_i(w^\prime,\tau)$.
Therefore, applying the triangle inequality and
\cref{lemma: Bahadur Representation-rate} give
\[
\begin{aligned}
&\sup_{\tau\in\mathcal T}
\left|
\sqrt n\left\{\widehat q(\tau)-q(\tau)\right\}
-
\frac{1}{\sqrt n}\sum_{i=1}^n\widetilde\psi_i(\tau)
\right|
\\
\le \ &
\sum_{w^\dagger\in\{w,w^\prime\}}
\sup_{\tau\in\mathcal T}
\left|
\sqrt n\left\{
\widehat q_{w^\dagger}(\tau)-q_{w^\dagger}(\tau)
\right\}
-
\frac{1}{\sqrt n}
\sum_{i=1}^n
\widetilde\phi_i(w^\dagger,\tau)
\right|
\\
= \ &
O_P(r_{\mathrm B,n})
=
o_P(1).
\end{aligned}
\]
\end{proof}

\subsection{Proof of Corollary \ref{corollary: fixed_Representation}}

\begin{proof}[Proof of \cref{corollary: fixed_Representation}]
We omit the proof of \cref{corollary: fixed_Representation}, as it is an immediate implication of \cref{lemma:Bahadur representation} given that $\left|\pi_n(w^\dagger) - \pi(w^\dagger)\right| = o_P(1)$.
\end{proof}

\subsection{Proof of Proposition \ref{proposition:NP_quantile_process_convergence}}
\label{appendix:proof of proposition:NP_quantile_process_convergence}

\begin{proof}[Proof of \cref{proposition:NP_quantile_process_convergence}]
Recall that the score function $\psi_i(\tau)$ is defined as
\[
\psi_i(\tau)
=
\frac{
\mathds{1}_i(w)
\bigl(\tau-\mathds{1}\{Y_i\le q_w(\tau)\}\bigr)
}{
\pi(w)\,
f_w\bigl(q_w(\tau)\bigr)
}
-
\frac{
\mathds{1}_i(w')
\bigl(\tau-\mathds{1}\{Y_i\le q_{w'}(\tau)\}\bigr)
}{
\pi(w^\prime)\,
f_{w'}\bigl(q_{w'}(\tau)\bigr)
},
\]
where the probabilities $\pi(w^\dagger)$ for $w^\dagger \in \{ w,w^\prime \}$ are defined in \cref{eq: pi_w}. By \cref{corollary: fixed_Representation}, to show the weak convergence of $  \sqrt{n} \left( \widehat{q}\left(\tau\right) -  q\left(\tau\right) \right)$, it suffices to show 
\[
\mathbb{G}_n(\tau) \equiv \frac{1}{\sqrt{n}} \sum_{i=1}^n \psi_i(\tau)  \rightsquigarrow \mathbb{G}(\tau)
\quad\text{in } \  \ell^\infty(\mathcal T),
\]
where $\mathbb{G}(\cdot)$ is a tight, mean-zero Gaussian process with covariance kernel $\Sigma(\tau, \tau')$.  Following  Example 1.5.10 in \cite{vaart2023empirical}, this requires establishing the marginal convergence and stochastic equicontinuity.

\vspace{0.1cm}

\noindent \underline{\textbf{Step 1.} Marginal Convergence of $\mathbb{G}_n(\tau) $.} For any $\bar{K} \in \mathbb{N}^+$, fix a finite set $\{\tau_1, \ldots, \tau_{\bar{K}}\} \subseteq \mathcal{T}$ and $\boldsymbol{\lambda} \equiv \left(\lambda_1,\ldots, \lambda_{\bar{K}} \right)^\top \in \mathbb{R}^{\bar{K}}$. Let  
\[
S_n = \frac{1}{\sqrt n} \sum_{i=1}^n  \eta_i, \quad\text{where}\quad \eta_i = \sum_{k=1}^{\bar{K}} \lambda_k \psi_i(\tau_k).
\]
By definition of $\eta_i$,  the dependency graph $\boldsymbol{A}^{(2)}_n$ of $\{ \eta_i \}_{i=1}^n$ has  the maximal degree $\Delta_n$ uniformly bounded by a universal constant, not depending on $n, {\bar{K}},  \boldsymbol{\lambda}$ and $\{ \tau_j \}_{k=1}^{\bar{K}}$. Let $\boldsymbol{G}_{n}(\boldsymbol{\tau} ) = \left( \mathbb{G}_{n}(\tau_k) \right)_{k=1}^{\bar{K}} $, and it follows that
\begin{equation}\label{eq:co_variance_limit}
\begin{aligned}
\mathrm{Var}(S_n) & = \boldsymbol{\lambda}^\top \mathrm{Cov}\left(  \boldsymbol{G}_{n}(\boldsymbol{\tau} ),  \boldsymbol{G}_{n}(\boldsymbol{\tau} ) \right)  \boldsymbol{\lambda}  \rightarrow  \boldsymbol{\lambda}^\top \boldsymbol{\Sigma}_{\bar{K}}  \boldsymbol{\lambda},  \quad \text{as } n \rightarrow \infty,
\end{aligned}
\end{equation}
where $\boldsymbol{\Sigma}_{\bar{K}}$ denotes the ${\bar{K}} \times {\bar{K}}$ matrix with entries ${\Sigma}\left(\tau_k, \tau_\ell\right)$ given in \cref{ass:final_variance}. Moreover, it is easy to see that $\eta_i$  are mean-zero, and uniformly bounded. Let $\bar{\sigma}_n=\sqrt{\mathrm{Var}(S_n)}$, Theorem 3.6 in \cite{ross2011fundamentals} implies that 
\[
\mathbb{W}_1\left(   S_n/\bar{\sigma}_n, Z \right) \leq \frac{\Delta_n^2}{\bar{\sigma}_n^3} \sum_{i=1}^n \mathbb{E} \left|\eta_i/\sqrt{n}\right|^3 + \frac{\sqrt{28 \Delta_n^3 }  }{\sqrt{\pi} \bar{\sigma}_n^2 }  \sqrt{\sum_{i=1}^n \mathbb{E} \left|\eta_i/\sqrt{n}\right|^4 },
\]
where $\mathbb{W}_1(\cdot, \cdot)$ denotes the 1-Wasserstein distance, and $Z \sim N(0,1)$ is the standard normal variable.

Notice that  $\Delta_n$ is uniformly bounded in $n$ under \cref{ass:degree}. Together with $\sum_{i=1}^n \mathbb{E} |\eta_i|^3 = O(n)$, $\sum_{i=1}^n \mathbb{E} |\eta_i|^4 = O(n)$, $\bar{\sigma}_n \rightarrow |\boldsymbol{\lambda}^\top \boldsymbol{\Sigma}_{\bar{K}}\boldsymbol{\lambda}|^{1/2}$, it follows that
\[
\lim_{n\rightarrow \infty}\mathbb{W}_1\left(   S_n/\bar{\sigma}_n, Z \right) = 0.
\]
This establishes that $S_n / \bar{\sigma}_n \rightsquigarrow \mathrm{N}(0,1)$ by Theorem 7.12 in \cite{villani2021topics}. Then, applying \cref{eq:co_variance_limit} and Slutsky’s Theorem, we conclude that $S_n \rightsquigarrow N \big(0, \boldsymbol{\lambda}^\top \boldsymbol{\Sigma}_{\bar{K}}  \boldsymbol{\lambda} \big)$. This completes the proof of the marginal convergence.

\noindent \underline{\textbf{Step 2.} Stochastic Equicontinuity of $\mathbb{G}_n(\tau)$.} Recall that the score function $\psi_i(\tau)$ is defined as
\[
\psi_i(\tau)
=
\frac{
\mathds{1}_i(w)
\bigl(\tau-\mathds{1}\{Y_i\le q_w(\tau)\}\bigr)
}{
\pi(w)\,
f_w\bigl(q_w(\tau)\bigr)
}
-
\frac{
\mathds{1}_i(w')
\bigl(\tau-\mathds{1}\{Y_i\le q_{w'}(\tau)\}\bigr)
}{
\pi(w^\prime)\,
f_{w'}\bigl(q_{w'}(\tau)\bigr)
}.
\]
Consider the following decomposition:
\[
\begin{aligned}
\mathbb{G}_{n}(\tau) = \frac{1}{\sqrt{n}} \sum_{i=1}^n \psi_i(\tau)  =  \frac{1}{\sqrt{n}} \sum_{i=1}^n    \phi_i(\tau, w)  - \frac{1}{\sqrt{n}} \sum_{i=1}^n    \phi_i(\tau, w^\prime),
\end{aligned}
\]
where 
\[
\phi_i(\tau, w^\dagger) \equiv \frac{
\mathds{1}_i(w^\dagger)
\bigl(\tau-\mathds{1}\{Y_i\le q_{w^\dagger}(\tau)\}\bigr)
}{
\pi(w^\dagger)\,
f_{w^\dagger}\bigl(q_{w^\dagger}(\tau)\bigr)
}  = \frac{
\mathds{1}_i(w^\dagger)
\bigl(\tau-\mathds{1}\{U_i \leq \tau\}\bigr)
}{
\pi(w^\dagger)\,
f_{w^\dagger}\bigl(q_{w^\dagger}(\tau)\bigr)
}  .
\]
Consequently, to establish the stochastic equicontinuity of $\mathbb{G}_n(\tau)$, it suffices to verify this property for the components $n^{-1/2}\sum_{i=1}^n \phi_i(\tau, w^\dagger)$ for each $w^\dagger \in \{ w, w^\prime \}$. We focus on the case $w^\dagger = w$, as the argument for $w^\prime$ is identical. Let $\zeta(w,\tau) = f_w\bigl(q_w(\tau)\bigr) \pi(w)$. By \cref{assumption: density,assumption: DD_limit}, the function $\zeta(w,\tau)$ is uniformly bounded and bounded away from zero over $\tau \in \mathcal{T}$. Moreover, the mapping $\tau \mapsto \zeta(w,\tau)$ is Lipschitz continuous. Hence, proving the stochastic equicontinuity of $\tau \mapsto n^{-1/2} \sum_{i=1}^n \phi_i(\tau, w)$ in $\ell^\infty(\mathcal{T})$ reduces to establishing that of the simplified process
\[
L_{1,n}(\tau)
=
\frac{1}{\sqrt n}
\sum_{i=1}^n
\mathds{1}_i(w) \left[\tau-\mathds{1}\{U_i \leq \tau \} \right].
\]

Now, let us establish that for any $\epsilon > 0$,
\begin{equation}\label{eq: SEQC}
\lim_{\delta \downarrow 0} \limsup_{n \to \infty} \mathbb{P}_{\boldsymbol{A}_n}\left[\sup_{ \tau_1 < \tau_2 < \tau_1 + \delta}  |L_{1,n}(\tau_1) - L_{1,n}(\tau_2)| > \epsilon  \right] = 0.
\end{equation}
For any $\tau_1, \tau_2 \in \mathcal{T}$, define $f_{\tau_1, \tau_2}: \{0,1\} \times \mathbb{N}^2 \times  [0,1] \rightarrow \mathbb{R}$ as
\[
f_{\tau_1, \tau_2}(\bar{w}, \bar{u}) = \mathds{1}\{ \bar{w} =w \} \left[ (\tau_2- \tau_1) - \mathds{1}\{ \tau_1 < \bar{u} \leq \tau_2 \} \right].
\]
Consequently, we have for any $\tau_2 > \tau_1$,
\[
L_{1,n}(\tau_2) - L_{1,n}(\tau_1) = \frac{1}{\sqrt{n}}\sum_{i=1}^n  f_{\tau_1, \tau_2}(W_i,U_i).
\]

For any $\delta >0$, define $\mathcal{F}_\delta \equiv \left\{ f_{\tau_1, \tau_2}:  \tau_1 < \tau_2 < \tau_1 + \delta \right \}$. Example 4.17 in \cite{wainwright2019high} implies that $\mathrm{VC}(\mathcal{F}_\delta) = 3$ for any $\delta >0$.  In fact, it suffices to show that 
\begin{equation}\label{eq:SEC}
\lim_{\delta \downarrow 0} \limsup_{n \to \infty} \mathbb{P}_{\boldsymbol{A}_n}\left[ \sup_{ f\in \mathcal{F}_\delta} \left| \frac{1}{\sqrt{n}} \sum_{i=1}^n   f (W_i, U_i) \right| > \epsilon \right] = 0
\end{equation}
with probability approaching one. Let $\delta \downarrow 0$ be arbitrary. By Markov's inequality and the symmetrization Lemma 2.3.1 in \cite{vaart2023empirical},  
\begin{equation}\label{eq:Markov}
\mathbb{P}_{\boldsymbol{A}_n}\left[  \sup_{ f\in \mathcal{F}_\delta} \left| \frac{1}{\sqrt{n}} \sum_{i=1}^n   f(W_i, U_i) \right| > \epsilon   \right]  \leq \frac{2}{\epsilon} \mathbb{E}_{\boldsymbol{A}_n} \left[ \mathbb{E}_{\varepsilon} \left[\sup_{f \in \mathcal{F}_\delta}  \left|\frac{1}{\sqrt{n}} \sum_{i=1}^n \varepsilon_i f(W_i, U_i)  \right| \right]  \right], 
\end{equation}
where $\varepsilon_i$ are i.i.d. Rademacher
variables. Recall the partition  $[n] = \mathcal{C}_1 \cup \mathcal{C}_2 \cup \dots \cup \mathcal{C}_{\Delta_{n}+1}$ introduced in \cref{sec:preliminary}, such that the variables $\{ (W_i, U_i):i \in \mathcal{C}_k \}$ are 
independent conditional on $\boldsymbol{A}_n$.  For each color  $k \in \{1, \ldots, \Delta_n+1 \}$, define 
\[
\sigma_{n,k} = \sup_{f \in \mathcal{F}_{\delta} } \sqrt{ |\mathcal{C}_k|^{-1} \sum_{i \in \mathcal{C}_k} \left| f(W_i, U_i)  \right|^2 }.
\]
Given that $\|f\|_\infty \leq 1$ for all  $f \in \mathcal{F}_\delta$, and the fact that $\mathrm{VC}(\mathcal{F}_\delta) = 3$, we obtain that
\[
\mathbb{E}_{\varepsilon} \left[\sup_{f \in \mathcal{F}_\delta}  \left|\frac{1}{\sqrt{|\mathcal{C}_k|}} \sum_{i\in\mathcal C_k} \varepsilon_i f(W_i, U_i)  \right|  \right] \leq K \int_0^{\sigma_{n,k}} \sqrt{ 2 \log (A/\epsilon)  } \mathrm{d} \epsilon,
\]
for some universal constant $A,K > 0$.  Let $\sigma_n = \sup_{f \in \mathcal{F}_{\delta}} \sqrt{ n^{-1} \sum_{i=1}^n |f(W_i, U_i)|^2  }$.  So, we have 
\[
\sigma_{n,k} \leq \sqrt{n/|\mathcal{C}_k|}  \sigma_n \leq \sqrt{\Delta_n +1}  \sigma_n = O(\sigma_n),
\]
by \cref{eq:HS_thm}. Thus,
\[
\begin{aligned}
\mathbb{E}_{\varepsilon} \left[\sup_{f \in \mathcal{F}_\delta}  \left|\frac{1}{\sqrt{n}} \sum_{i\in\mathcal C_k} \varepsilon_i f(W_i, U_i)  \right|  \right] \leq K  \sqrt{ |\mathcal{C}_k|/ n } \int_0^{ \sqrt{\Delta_n +1}  \sigma_n  } \sqrt{ 2 \log (A/\epsilon)  } \mathrm{d} \epsilon.
\end{aligned}
\]
Summing over $k \in \{1,\ldots, \Delta_n+1\}$, there is a universal constant $K > 0$ large enough such that
\begin{equation}\label{eq:integral_1}
\begin{aligned}
\mathbb{E}_{\boldsymbol{A}_n} \left[  \mathbb{E}_\varepsilon \left[ \sup_{f\in \mathcal{F}_{\delta} }  \left|  \frac{1}{\sqrt{n}} \sum_{i=1}^n \varepsilon_i f(W_i, U_i) \right|  \right]  \right] & \leq \sum_{k=1}^{\Delta_n+1} \mathbb{E}_{\boldsymbol{A}_n} \left[ \mathbb{E}_\varepsilon\left[\sup_{f\in \mathcal{F}_{\delta} }\left|\frac{1}{\sqrt{n}}\sum_{i\in\mathcal C_k}\varepsilon_i f(W_i, U_i)\right|\right]  \right] \\
 & \leq K  \sum_{k=1}^{\Delta_n+1} \sqrt{ |\mathcal{C}_k|/ n } \int_0^{ \sqrt{\Delta_n +1}  \sigma_n  } \sqrt{ 2 \log (A/\epsilon)  } \mathrm{d} \epsilon \\
  & \leq K   \int_0^{ \sqrt{\Delta_n +1}  \sigma_n  } \sqrt{ 2 \log (A/\epsilon)  } \mathrm{d} \epsilon.
\end{aligned}
\end{equation}
By \cref{lemma: ULLN_U}, we can derive
 \[
 \begin{aligned}
\sup_{f \in \mathcal{F}_{\delta } }\frac{1}{n} \sum_{i=1}^n |f(W_i,U_i)|^2 &   \leq  \sup_{ \tau_1 < \tau_2 < \tau_1 + \delta }   \frac{1}{n} \sum_{i=1}^n \left| (\tau_2 - \tau_1) -  \mathds{1}\{ \tau_1 < U_i \leq \tau_2 \} \right|^2 \\
 & \leq  \delta^2 + \delta + \sup_{ \tau_1 \in\mathcal{T}} \left|   \frac{1}{n} \sum_{i=1}^{n} \mathds{1}\{  U_i \leq \tau_1 \} -  \tau_1 \right|
 + \sup_{ \tau_2 \in\mathcal{T} } \left|   \frac{1}{n} \sum_{i=1}^{n} \mathds{1}\{  U_i \leq \tau_2 \} -  \tau_2 \right|\\
 & \leq \delta^2 + \delta + O_P(n^{-1/2}) = o_P(1),
 \end{aligned}
 \]
which implies that $\sigma_n = o_P(1)$. Thus, by using the Cauchy-Schwarz inequality and the dominated convergence theorem to see that the expectation  of the  integral on the RHS of \cref{eq:integral_1} converges to zero. Therefore, \cref{eq:Markov} converges to zero in probability, and the desired result follows.
\end{proof}

\section{Proofs for Section \ref{sec:dependent bootstrap} }\label{appendix:proof for section 4}

\subsection{Proof of Lemma \ref{lemma: feasible_variance_convergence}}

Recall the score function $\widetilde{\psi}_i(\tau)$ defined in \cref{eq:influence_function} and the network-conditional covariance kernel $\widetilde{\mathbb V}_{\boldsymbol A_n}(\tau,\tau^\prime)$ defined in \cref{eq: network_conditional_covariance_kernel}. Let $\widetilde{\boldsymbol{\psi}}(\tau) = \bigl( \widetilde{\psi}_1(\tau), \dots, \widetilde{\psi}_n(\tau) \bigr) \in \mathbb{R}^n$, and define
\[
\widetilde{\mathbb{V}}_n(\tau, \tau^\prime) = \frac{1}{n} \widetilde{ \boldsymbol{\psi}}(\tau)^\top  \boldsymbol{\Omega}_n  \widetilde{ \boldsymbol{\psi}} (\tau^\prime).
\]

\begin{proof}[Proof of \cref{lemma: feasible_variance_convergence}]
We decompose
\begin{equation}\label{eq:variance_difference}
\begin{aligned}
\widehat{\mathbb{V}}_n(\tau,\tau') - \widetilde{\mathbb V}_{\boldsymbol{A}_n}(\tau,\tau') &= \left[\widehat{\mathbb{V}}_n(\tau,\tau')-\widetilde{\mathbb{V}}_n(\tau,\tau')  \right]  + \left[ \widetilde{\mathbb{V}}_n(\tau,\tau')- \widetilde{\mathbb V}_{\boldsymbol{A}_n}(\tau,\tau')\right].
\end{aligned}
\end{equation}
By \cref{lemma:Oracle_variance}, the second term on the RHS of \cref{eq:variance_difference} is $O_P\big(n^{-1/2}\big)$. It therefore remains to control the first term. Define $\widehat{\Delta}_i(\tau)=\widehat\psi_i(\tau)-\widetilde{\psi}_i(\tau)$ and $\widehat{\boldsymbol{\Delta}}(\tau) = \big(\widehat{\Delta}_1(\tau), \ldots, \widehat{\Delta}_n(\tau) \big) $. Then
\[
\widehat{\mathbb{V}}_n(\tau,\tau') - \widetilde{\mathbb{V}}_n(\tau,\tau') = \frac1n \widehat{\boldsymbol\Delta}(\tau)^\top \boldsymbol{\Omega}_n\widehat{\boldsymbol\psi}(\tau') + \frac1n \widetilde{ \boldsymbol\psi}(\tau)^\top \boldsymbol{\Omega}_n\widehat{\boldsymbol\Delta}(\tau').
\]
The Cauchy-Schwarz inequality gives
\[
\left| \frac1n \widehat{\boldsymbol\Delta}(\tau)^\top \boldsymbol{\Omega}_n\widehat{\boldsymbol\psi}(\tau') \right| \le \|\boldsymbol{\Omega}_n\|_{\mathrm{op}} \left[ \frac1n\sum_{i=1}^n|\widehat{\Delta}_i(\tau)|^2 \right]^{1/2} \left[ \frac1n\sum_{i=1}^n|\widehat\psi_i(\tau')|^2 \right]^{1/2},
\]
where $\| \boldsymbol{\Omega}_n\|_{\mathrm{op}}$ denotes the spectral norm of $\boldsymbol{\Omega}_n$. Moreover,
\[
\|\boldsymbol\Omega_n\|_{\mathrm{op}}
\le
\max_i\sum_{j=1}^n|\boldsymbol\Omega_n(i,j)|
\le 1+K_{\max}^2 = O(1),
\]
where the last inequality follows from \cref{ass:degree}. By \cref{eq: bounds_for_zeta_n}, $\sum_{i=1}^n| \widetilde{\psi}_i(\tau)|^2/n$ is uniformly bounded over $\tau \in \mathcal{T}$. Together with \cref{lemma: score_LLN}, this yields
\[
\frac{1}{n}\sum_{i=1}^n|\widehat\psi_i(\tau)|^2 \leq 2 \frac{1}{n}\sum_{i=1}^n|\widetilde\psi_i(\tau)|^2+2 \frac{1}{n}\sum_{i=1}^n|\widehat\Delta_i(\tau)|^2=O_P(1),
\] 
uniformly over $\tau\in\mathcal T$. Applying \cref{lemma: score_LLN} once more gives
\[
\sup_{\tau,\tau'\in\mathcal T} \left| \frac1n \widehat{\boldsymbol\Delta}(\tau)^\top \boldsymbol{\Omega}_n \widehat{\boldsymbol\psi}(\tau') \right| = O_P\left( \sup_{\tau\in\mathcal T}\sqrt{\frac1n\sum_{i=1}^n|\widehat{\Delta}_i(\tau)|^2} \right)  = O_P \big(  r_{f,n} \vee n^{-1/4} \big).
\]
The same argument controls $\frac1n \widetilde{\boldsymbol\psi}(\tau)^\top \boldsymbol{\Omega}_n \widehat{ \boldsymbol\Delta}(\tau')$. Hence,
\[
\sup_{\tau,\tau'\in\mathcal T} \left| \widehat{\mathbb{V}}_n(\tau,\tau') - \widetilde{\mathbb{V}}_n(\tau,\tau') \right|  =  O_P \big(  r_{f,n} \vee n^{-1/4} \big),
\]
and the desired result follows.
\end{proof}

\subsection{Proof of Theorem \ref{theorem: feasible_gaussian_approximation}}

Let $\mathcal{T}_n \equiv \{\tau_{1},\ldots,\tau_{p_n}\}\subseteq\mathcal T$ be a prespecified grid, and define its mesh size by
\begin{equation}\label{eq: def_delta_}
\delta_n\equiv \sup_{\tau\in\mathcal T}\inf_{1\le i\le p_n}|\tau-\tau_{i}|.
\end{equation}

\begin{lemma}\label{lemma: grid approximation}
Suppose \cref{ass:degree,assumption: density,assumption:local dependency,assumption: smallest_variance_diagonal} hold. Then, conditional on \(\{ \boldsymbol{A}_n \}_{n\ge1}\),
\[
\left| \sup_{\tau \in \mathcal{T}}  \left|  \frac{\sqrt{n} \left( \widehat{q}(\tau) - q(\tau) \right) }{ \sigma_n(\tau) } \right|  -   \sup_{\tau \in \mathcal{T}_n} \left|   \frac{1}{\sqrt{n}}\sum_{i=1}^n  \frac{\widetilde{\psi}_i(\tau) }{ \sigma_n(\tau)} \right|  \right| =O_P\left( \eta_n \right),
\]
where $\eta_n = r_{\mathrm{B},n} + r_{\mathrm{Loc},n} +  \delta_n$, and
\[
r_{\mathrm{B},n} = n^{-1/8} (\log n)^{1/4} \quad \text{and} \quad r_{\mathrm{Loc},n} = \sqrt{ \delta_n \log(e/\delta_n)} + n^{-1/2} \log (e/\delta_n).
\]
\end{lemma}

\begin{proof}[Proof of \cref{lemma: grid approximation}]

Define the standardized score process $\widetilde{\mathbb{H}}_{n}(\tau) = n^{-1/2}\sum_{i=1}^n  {\widetilde{\psi}_i(\tau) }/{ \sigma_n(\tau)} $. By \cref{lemma:Bahadur representation},
\begin{equation}\label{eq: Rate_B_n}
\sup_{\tau \in \mathcal{T}} \left| \sqrt{n}\left( \widehat{q}(\tau) - q(\tau) \right) - \frac{1}{\sqrt{n}} \sum_{ i=1}^n \widetilde{\psi}_i(\tau)  \right| = O_P\left( r_{\mathrm{B},n}  \right).
\end{equation}
Since \cref{assumption: smallest_variance_diagonal} gives $\inf_{\tau\in\mathcal T}\sigma_n(\tau)\ge c_\sigma$, the preceding bound implies
\[
\begin{aligned}
\sup_{\tau \in \mathcal{T}} \left|  \frac{\sqrt{n}\left( \widehat{q}(\tau) - q(\tau) \right) }{\sigma_n(\tau)} -   \widetilde{\mathbb{H}}_{n}(\tau)\right| & \leq   \sup_{\tau \in \mathcal{T}} \left|  \frac{\sqrt{n}\left( \widehat{q}(\tau) - q(\tau) \right)  -  \frac{1}{\sqrt{n}}\sum_{i=1}^n  \widetilde{\psi}_i(\tau) }{c_\sigma }   \right| \\
& = O_P\left( r_{\mathrm{B},n} \right).
\end{aligned}
\]

It remains to establish the grid approximation
\begin{equation}\label{eq: final}
\sup_{\tau \in \mathcal{T}} \left|    \widetilde{\mathbb{H}}_{n}(\tau)  \right|  -     \sup_{\tau \in \mathcal{T}_n} \left|   \widetilde{\mathbb{H}}_{n}(\tau) \right|  = O_P\left( r_{\mathrm{Loc},n} + \delta_n \right).    
\end{equation}
For any $\tau,\tau'$, decompose
\[
\widetilde{\mathbb{H}}_{n}(\tau) - \widetilde{\mathbb{H}}_{n}(\tau^\prime) = \frac{1}{\sigma_n(\tau)} \frac{   1 }{ \sqrt{n}  } \sum_{i=1}^n \Big( \widetilde{\psi}_i(\tau) - \widetilde{\psi}_i(\tau^\prime) \Big) + \frac{1}{\sqrt{n}} \sum_{i=1}^n \widetilde{\psi}_i(\tau^\prime) \left[ \frac{1}{\sigma_n(\tau)} - \frac{1}{\sigma_n(\tau^\prime)} \right].
\]
Taking the supremum over $|\tau-\tau'|\le\delta_n$ and applying the triangle inequality together with \cref{assumption: smallest_variance_diagonal} yields
\begin{equation}\label{eq: Z_n_eq_continuity}
\begin{aligned}
\sup_{|\tau - \tau^\prime| \leq \delta_n} \left| \widetilde{\mathbb{H}}_{n}(\tau) - \widetilde{\mathbb{H}}_{n}(\tau^\prime)   \right| & \leq \frac{1}{c_\sigma} \underbrace{\sup_{|\tau - \tau^\prime| \leq \delta_n}  \left| \frac{1}{\sqrt{n}} \sum_{i=1}^n \Big( \widetilde{\psi}_i(\tau) - \widetilde{\psi}_i(\tau^\prime) \Big)  \right | }_{= \mathrm{Term}_{\widetilde{\psi}, n} }  \\
& + \frac{1}{c_\sigma^2}  \sup_{|\tau - \tau^\prime| \leq \delta_n}  \left| \frac{1}{\sqrt{n}} \sum_{i=1}^n  \widetilde{\psi}_i(\tau^\prime)  \right |   \underbrace{\sup_{|\tau - \tau^\prime| \leq \delta_n}   \left| \sigma_n(\tau) -\sigma_n(\tau^\prime)   \right| }_{= \mathrm{Term}_{\sigma,n}}
\end{aligned}
\end{equation}

\noindent \underline{\bf Step 1.} 
Recall that $\widetilde\psi_i(\tau)
=
\widetilde\phi_i(w,\tau)-\widetilde\phi_i(w^\prime,\tau)$. For $w^\dagger \in \{w,w^\prime \}$,
\[
\widetilde\phi_i(w^\dagger,\tau)
=\frac{\mathbb{1}_i(w^\dagger)\left( \tau- \mathds{1} \left\{ Y_i\le q_{w^\dagger}(\tau) \right\}  \right) }{f_{w^\dagger}(q_{w^\dagger}(\tau))
\pi_{n}(w^\dagger) } = \frac{\mathbb{1}_i(w^\dagger)\left( \tau- \mathds{1} \left\{U_i \leq\tau \right\}  \right) }{f_{w^\dagger}(q_{w^\dagger}(\tau))
\pi_{n}(w^\dagger) } .
\]
To bound $\mathrm{Term}_{\widetilde\psi,n}$, it suffices to establish, for each $w^\dagger\in\{w,w'\}$,
\begin{equation}\label{eq: eq_continuity}
\sup_{|\tau - \tau^\prime| \leq \delta_n}  \left| \frac{1}{\sqrt{n}} \sum_{i=1}^n \left[\widetilde\phi_i(w^\dagger,\tau) -   \widetilde\phi_i(w^\dagger,\tau^\prime)  \right] \right | = O_P\left(r_{\mathrm{Loc},n}  + \delta_n \right).
\end{equation}
We prove \cref{eq: eq_continuity} for $w^\dagger=w$; the argument for $w^\dagger=w'$ is identical.

Recall that  $\zeta_n(w,\tau) = f_{w}(q_{w}(\tau))\pi_{n}(w)$ and $\gamma_i(\tau)=\mathds{1}_i(w)\{\tau-1(U_i\le \tau)\}.$ Then
\begin{equation}\label{eq: phi_i_decomposition}
\begin{aligned}
\frac1{\sqrt n}\sum_{i=1}^n \left[ \widetilde\phi_i(w,\tau)
-\widetilde\phi_i(w,\tau')  \right] &=\frac1{\sqrt n}\sum_{i=1}^n\frac{\gamma_i(\tau)-\gamma_i(\tau')}{\zeta_n(w,\tau) }\\
& + \left[\zeta_n(w,\tau)^{-1}-\zeta_n(w,\tau')^{-1}\right]\frac1{\sqrt n}\sum_{i=1}^n\gamma_i(\tau').
\end{aligned}
\end{equation}
By \cref{eq: bounds_for_zeta_n,assumption: density}, $\zeta_n(w,\tau)$ is uniformly bounded away from zero and Lipschitz continuous in $\tau$. Hence $\tau\mapsto\zeta_n(w,\tau)^{-1}$ is also Lipschitz continuous, so
\begin{equation}\label{eq: Lipschiz_gamma}
\sup_{ |\tau - \tau^\prime| \leq \delta_n }\left| \zeta_n(w,\tau)^{-1}
-\zeta_n(w,\tau')^{-1}  \right| \lesssim \delta_n.    
\end{equation}
By \cref{lemma:L1n_uniform_bounded}, the second term on the right-hand side of \cref{eq: phi_i_decomposition} is therefore $O_P(\delta_n)$ uniformly over $|\tau-\tau'|\le\delta_n$:
\[
\sup_{ \left|\tau - \tau^\prime \right| \leq \delta}\left| \left(\zeta_n(w,\tau)^{-1}
-\zeta_n(w,\tau')^{-1}\right)
\frac1{\sqrt n}
\sum_{i=1}^n
\gamma_i(\tau')\right| = O_P\left(\delta_n\right).
\]
By \cref{lemma:local_modulus_gamma}, we have 
\[
\sup_{ |\tau - \tau^\prime| \leq \delta_n } \left| \frac{1}{\sqrt{n}} \sum_{i=1}^n \left[  \gamma_i(\tau) - \gamma_i(\tau^\prime) \right] \right| = O_P\left(r_{\mathrm{Loc},n} \right).
\]
The above results yield
\[
\sup_{|\tau-\tau'|\le\delta_n}\left|\frac1{\sqrt n}\sum_{i=1}^n \left[  \widetilde\phi_i(w,\tau)-\widetilde\phi_i(w,\tau') \right]  \right|  = O_P\left(r_{\mathrm{Loc},n} + \delta_n \right). 
\]
Combining these two bounds proves \cref{eq: eq_continuity}, and hence 
\[
\mathrm{Term}_{\widetilde{\psi}, n}= O_P\left(r_{\mathrm{Loc},n}  + \delta_n \right).
\]

\noindent \underline{\bf Step 2.} We next show that
\[
\sup_{|\tau - \tau^\prime| \leq \delta_n}   \left| \sigma_n(\tau) -\sigma_n(\tau^\prime)   \right| = O\left( \delta_n\right).
\]
Recall that 
\[
\sigma_n^2(\tau)=\frac1n\sum_{i=1}^n\sum_{{\ell}_{\boldsymbol A_n}(i,j)\le 2}\mathbb E_{\boldsymbol A_n}\left[\widetilde\psi_i(\tau)\widetilde\psi_j(\tau)\right].
\]
Then, uniformly over \(|\tau-\tau'|\le\delta_n\),
\[
\begin{aligned}
\left|\sigma_n^2(\tau)-\sigma_n^2(\tau')\right|
&\le\frac1n\sum_{i=1}^n\sum_{{\ell}_{\boldsymbol A_n}(i,j)\le 2}\mathbb E_{\boldsymbol A_n}\left[\left|\widetilde\psi_i(\tau)\widetilde\psi_j(\tau)-\widetilde\psi_i(\tau')\widetilde\psi_j(\tau')\right|\right]\\
&\lesssim\frac1n\sum_{i=1}^n\sum_{{\ell}_{\boldsymbol A_n}(i,j)\le 2}|\tau-\tau'| \lesssim(\Delta_n+1)\delta_n.
\end{aligned}
\]
By \cref{ass:degree,assumption: smallest_variance_diagonal}, we can conclude that
\[
\mathrm{Term}_{\sigma,n} =  \sup_{|\tau-\tau'|\le\delta_n} |\sigma_n(\tau)-\sigma_n(\tau')|\le\frac{1}{2c_\sigma}\sup_{|\tau - \tau^\prime| \leq \delta_n}|\sigma_n^2(\tau)-\sigma_n^2(\tau')| = O \left( \delta_n \right).
\]

\noindent \underline{\bf Step 3.}  Substituting the bounds from Steps 1 and 2 into \cref{eq: Z_n_eq_continuity} gives
\begin{equation}\label{eq: Z_n_sup_diff}
\sup_{|\tau - \tau^\prime| \leq \delta_n} \left| \widetilde{\mathbb{H}}_{n}(\tau) - \widetilde{\mathbb{H}}_{n}(\tau^\prime)   \right|  = O_P\left(  r_{\mathrm{Loc}, n} + \delta_n \right). 
\end{equation}
Since \(\mathcal T_n\subseteq \mathcal T\), we have $\sup_{\tau\in\mathcal T}|\widetilde{\mathbb{H}}_{n}(\tau)|
\geq 
\sup_{\tau\in\mathcal T_n}\left|\widetilde{\mathbb{H}}_{n}(\tau)\right| $. For each \(\tau\in\mathcal T\), let \( \mathfrak{s}_n(\tau)\in\mathcal T_n\) be such that $|\tau-\mathfrak{s}_n(\tau)|\le \delta_n$. Then, for every
\(\tau\in\mathcal T\),
\[
\begin{aligned}
|\widetilde{\mathbb{H}}_{n}(\tau)|
& \leq
|\widetilde{\mathbb{H}}_{n}\left(\mathfrak{s}_n(\tau) \right)|
+
|\widetilde{\mathbb{H}}_{n}(\tau)-\widetilde{\mathbb{H}}_{n}\left(\mathfrak{s}_n(\tau) \right)|\\
& \leq \sup_{\tau\in\mathcal T_n}|\widetilde{\mathbb{H}}_{n}(\tau)| + \sup_{\tau^{\prime}\in\mathcal{T}_{n}:|\tau-\tau'|\le\delta_n}
\left|\widetilde{\mathbb{H}}_{n}(\tau)-\widetilde{\mathbb{H}}_{n}(\tau')\right|
\end{aligned}
\]
Taking the supremum over \(\tau\in\mathcal T\), we obtain
\[
\sup_{\tau\in\mathcal T}|\widetilde{\mathbb{H}}_{n}(\tau)|
\le
\sup_{\tau\in\mathcal T_n}|\widetilde{\mathbb{H}}_{n}(\tau)|
+
\sup_{|\tau-\tau'|\le\delta_n}
|\widetilde{\mathbb{H}}_{n}(\tau)-\widetilde{\mathbb{H}}_{n}(\tau')|.
\]
Hence,
\[
0\le
\sup_{\tau\in\mathcal T}|\widetilde{\mathbb{H}}_{n}(\tau)|
-
\sup_{\tau\in\mathcal T_n} \left|\widetilde{\mathbb{H}}_{n}(\tau)\right|
\le
\sup_{|\tau-\tau'|\le\delta_n}
\left|\widetilde{\mathbb{H}}_{n}(\tau)-\widetilde{\mathbb{H}}_{n}(\tau')\right|.
\]
By \cref{eq: Z_n_sup_diff}, we conclude that
\[
\left| \sup_{\tau\in\mathcal T}|\widetilde{\mathbb{H}}_{n}(\tau)|
-
\sup_{\tau\in\mathcal T_n}|\widetilde{\mathbb{H}}_{n}(\tau)|\right|
= O_P\left( r_{\mathrm{Loc},n} + \delta_n \right).
\]
Combining \cref{eq: Rate_B_n} with the derived upper bound proves the lemma.
\end{proof}

\begin{lemma}
\label{lemma: grid_oracle_gaussian_approx}
Let $\boldsymbol{\Lambda}_n=\operatorname{diag} \left\{\sigma_n(\tau_1),\ldots,\sigma_n(\tau_{p_n}) \right\}$, let \( \boldsymbol V_n\in\mathbb R^{p_n\times p_n}\) have entries $ V_n(j,k)=\widetilde{\mathbb V}_{\boldsymbol A_n}(\tau_j,\tau_k)$, and set $\boldsymbol\Gamma_n = \boldsymbol{\Lambda}_n^{-1}\boldsymbol V_n\boldsymbol{\Lambda}_n^{-1}$.  Let
\[
\boldsymbol{Z}_n \equiv  \left( Z_n(\tau_i): 1 \leq i \leq p_n \right) \sim \mathrm{N} \left(0, \boldsymbol\Gamma_n \right).
\]
Suppose \cref{ass:degree,assumption: density,assumption:local dependency,assumption: smallest_variance_diagonal} hold and $(\log p_n)^7/n\to0$. Then, conditional on \(\{ \boldsymbol{A}_n \}_{n\ge1}\),
\begin{equation}\label{eq:S_n_Z_n_comparision}
\sup_{s\in\mathbb R}\left|\mathbb P_{\boldsymbol A_n}\left[\sup_{\tau\in\mathcal T_n}\left|\frac{1}{\sqrt n}\sum_{i=1}^n\frac{\widetilde\psi_i(\tau)}{\sigma_n(\tau)}\right|\le s\right]-\mathbb P_{\boldsymbol A_n }\left[\sup_{1\le i\le p_n}\left|Z_n(\tau_i)\right|\le s\right]\right| =o(1).
\end{equation}

\end{lemma}

\begin{proof}
Define the \(p_n\)-dimensional vector
\begin{equation}\label{eq: tilde_S_n}
\boldsymbol{\widetilde{H}}_n=\left(\widetilde{\mathbb{H}}_{n}(\tau_{1}),\ldots,\widetilde{\mathbb{H}}_{n}(\tau_{p_{n}})\right).
\end{equation}
By construction,
\[
\sup_{\tau\in\mathcal T_n}
\left|
\frac{1}{\sqrt n}
\sum_{i=1}^n
\frac{\widetilde\psi_i(\tau)}{\sigma_n(\tau)}
\right|
=
\sup_{1\leq i\leq p_n}\left| \widetilde{\mathbb{H}}_{n}(\tau_i) \right|.
\]
It is easy to see that $\mathbb E_{\boldsymbol A_n}\big[ \boldsymbol{\widetilde{H}}_n \big]=0$ and $\mathrm{Cov}_{\boldsymbol A_n}(\boldsymbol{\widetilde{H}}_n ) = \boldsymbol\Gamma_n$.  Thus, $\boldsymbol Z_n$ is the centered Gaussian vector with the same conditional covariance matrix as $\boldsymbol{\widetilde H}_n$.

We establish \cref{eq:S_n_Z_n_comparision} by verifying the conditions of Theorem 2 in \cite{chang2024central}. Recall that $\boldsymbol A_n^{(2)}$ is as a dependency graph for $(Y_i,W_i,U_i)_{i=1}^n$, and its maximum degree is $\Delta_n$. In the notation of \cite{chang2024central}, let $D_n$ and $D_n^\ast$ denote the maximum sizes of the first- and second-order neighborhoods, respectively. Because their convention includes each node in its own neighborhood, $D_n\leq \Delta_n+1$, and $D_n^\ast\leq D_n^2\leq(\Delta_n+1)^2$. We now verify Conditions 1 and 3 of \cite{chang2024central}. By \cref{eq: bounds_for_zeta_n}, the scores $\widetilde\psi_i(\tau)$ are uniformly bounded. Together with \cref{assumption: smallest_variance_diagonal}, this gives, for some constant $C<\infty$,
\[
\sup_{1\le i\le n}\sup_{1\le j\le p_n}
\left|
\frac{\widetilde\psi_i(\tau_j)}{\sigma_n(\tau_j)}
\right|
\le C,
\]
which verifies Condition 1 in \cite{chang2024central}. Moreover, for each \(1\leq i\leq p_n\), $\operatorname{Var}_{\boldsymbol A_n}(\widetilde{\mathbb{H}}_{n}(\tau_i) )
=
\widetilde{\mathbb V}_{\boldsymbol A_n}(\tau_i,\tau_i)
/
\sigma_n^2(\tau_i)
=
1$,
so  Condition 3 is satisfied. Applying Theorem 2 of \cite{chang2024central} yields
\[
\sup_{A\in \mathcal{R} }
\left|
\mathbb P_{\boldsymbol A_n}\big[\boldsymbol{\widetilde{H}}_n\in A\big]
-
\mathbb P_{\boldsymbol A_n}\left[\boldsymbol Z_n\in A\right]
\right|
\lesssim
(\Delta_n+1)
\frac{(\log p_n)^{7/6}}{n^{1/6}},
\]
where \(\mathcal{R}\) denotes the class of all hyper-rectangles in
\(\mathbb R^{p_n}\). Finally, for each \(s\in\mathbb R\),
\[
\begin{aligned}
\boldsymbol{\widetilde{H}}_n \in [-s,s]^{p_n}  &   \Leftrightarrow   \sup_{1\leq i\leq p_n}\big|  \widetilde{\mathbb{H}}_{n}(\tau_i) \big|  \leq s ,\\
\boldsymbol Z_n\in[-s,s]^{p_n}  &\Leftrightarrow \sup_{1\leq i\leq p_n}|Z_n(\tau_i)| \leq s .
\end{aligned}
\]
Since \([-s,s]^{p_n}\in\mathcal R\), taking the supremum over \(s\in\mathbb R\) gives
\[
\sup_{s\in\mathbb R}\left|\mathbb P_{\boldsymbol A_n}\left[\sup_{1\leq i\leq p_n}\left| \widetilde{\mathbb{H}}_{n}(\tau_i) \right| \leq s\right]-\mathbb P_{\boldsymbol A_n}\left[\sup_{1\leq i\leq p_n}|Z_n(\tau_i)|\leq s\right]\right|=O\left( \frac{(\Delta_n+1) \left(\log p_n\right)^{7/6}}{n^{1/6}} \right).
\]
Under \cref{ass:degree}, $\Delta_n=O(1)$. Hence $(\log p_n)^7/n\to0$ makes the preceding bound $o(1)$, which completes the proof.
\end{proof}

\begin{lemma}
\label{lemma:psd_corrected_gamma_consistency}
Suppose \cref{ass:degree,assumption: density,assumption:local dependency,assumption:density_est_convergence,assumption: smallest_variance_diagonal} hold. Then, conditional on the realized network sequence \(\{ \boldsymbol{A}_n \}_{n\ge1}\),
\[
\big\|\widehat{\boldsymbol\Gamma}_{n}^+-\boldsymbol\Gamma_n\big\|_{\max}=O_P(\bar{r}_{\mathbb{V}, n} ),
\]
where $\bar{r}_{\mathbb{V}, n} = r_{f,n} \vee n^{-1/4}$.
\end{lemma}

\begin{proof}
First, by \cref{lemma: feasible_variance_convergence}, we have $\big\|
\widehat{\boldsymbol V}_n-\boldsymbol V_n
\big\|_{\max}
=
O_P\left(\bar{r}_{\mathbb{V}, n} \right)$. In particular,
\[
\sup_{1\le i\le p_n}
\left|
\widehat\sigma_n^2(\tau_i)-\sigma_n^2(\tau_i)
\right|
=
O_P\left(\bar{r}_{\mathbb{V}, n} \right).
\]
By \cref{assumption: smallest_variance_diagonal}, there exists
\(c_\sigma>0\) such that $\inf_{\tau\in\mathcal T}\sigma_n(\tau)\ge c_\sigma$. So, with probability approaching one,
\(\widehat\sigma_n(\tau_i)\ge c_\sigma/2\) for all  $i \in \left[p_n\right]$. Thus,
\[
\sup_{1\le i\le p_n}
\left|
\widehat\sigma_n^{-1}(\tau_i)-\sigma_n^{-1}(\tau_i)
\right|
=
O_P(\bar{r}_{\mathbb{V}, n} ).
\]
For any $1\le j,k\le p_n$,
\begin{align}
\big|\widehat{\Gamma}_n(j,k)-\Gamma_n(j,k)\big|
&=\left|\frac{\widehat{V}_n(j,k)}{\widehat\sigma_n(\tau_j)\widehat\sigma_n(\tau_k)}-\frac{V_n(j,k)}{\sigma_n(\tau_j)\sigma_n(\tau_k)}\right| \notag\\
&\le\left|\widehat{V}_n(j,k)- V_n(j,k)\right|\left|\widehat\sigma_n^{-1}(\tau_j)\widehat\sigma_n^{-1}(\tau_k)\right| \label{eq:gamma_bound_1}\\
& +\left| V_n(j,k)\right|\left|\widehat\sigma_n^{-1}(\tau_j)\widehat\sigma_n^{-1}(\tau_k)-\sigma_n^{-1}(\tau_j)\sigma_n^{-1}(\tau_k)\right|. \label{eq:gamma_bound_2}
\end{align}
The term in \eqref{eq:gamma_bound_1} is \(O_P(\bar{r}_{\mathbb{V}, n} )\) uniformly over \(j,k\) by \cref{lemma: feasible_variance_convergence}. For  \eqref{eq:gamma_bound_2},  the Cauchy-Schwarz inequality gives $|V_n(j,k)|\le\sigma_n(\tau_j)\sigma_n(\tau_k)$. Moreover,  $\widehat{\sigma}_n^{-1}(\tau)$ are uniformly bounded over $\tau$ with probability approaching one. Hence \eqref{eq:gamma_bound_2} is also $O_P(\bar r_{\mathbb V,n})$ uniformly over $j,k$. Therefore,
\[
\big\|
\widehat{\boldsymbol\Gamma}_n-\boldsymbol\Gamma_n
\big\|_{\max}
=
O_P(\bar{r}_{\mathbb{V}, n} ).
\]

Because \(\boldsymbol\Gamma_n = \boldsymbol{\Lambda}_n^{-1}   \boldsymbol{V}_n \boldsymbol{\Lambda}_n^{-1} \) is PSD and satisfies $\mathrm{diag}(\boldsymbol\Gamma_n)=\boldsymbol 1_{p_n}$. Therefore, $\boldsymbol\Gamma_n$ is feasible for \eqref{eq: PS_programming}. The optimality of $\widehat{\boldsymbol\Gamma}_n^+$ implies
\[
\big\|\widehat{\boldsymbol\Gamma}_{n}^+-\widehat{\boldsymbol\Gamma}_n\big \|_{\max}\le\big\|\boldsymbol\Gamma_n-\widehat{\boldsymbol\Gamma}_n\big \|_{\max}.
\]
By the triangle inequality,
\[
\big\|\widehat{\boldsymbol\Gamma}_{n}^+-\boldsymbol\Gamma_n\big\|_{\max}\le\big\|\widehat{\boldsymbol\Gamma}_{n}^+-\widehat{\boldsymbol\Gamma}_n\big\|_{\max}+\big\|\widehat{\boldsymbol\Gamma}_n-\boldsymbol\Gamma_n\big \|_{\max}\le2\big\|\widehat{\boldsymbol\Gamma}_n-\boldsymbol\Gamma_n\big\|_{\max} = O_P\left(\bar{r}_{\mathbb{V}, n} \right).
\]
This completes the proof.
\end{proof}

\begin{proof}[Proof of \cref{theorem: feasible_gaussian_approximation}]
Recall that $\boldsymbol Z_n=\big(Z_n(\tau_i):1\le i\le p_n\big)\sim\mathrm{N}\left(0,\boldsymbol\Gamma_n\right)$, and $\widehat{\boldsymbol{Z}}_n  = \big( \widehat{Z}_n(\tau_i): 1\leq i \leq p_n  \big) \sim \mathrm{N}\big(0, \widehat{\boldsymbol{ \Gamma}}^+_n\big)$. Define
\[
\widehat{T}_n
=
\sup_{\tau\in\mathcal T}
\left|
\frac{\sqrt n\left(\widehat q(\tau)-q(\tau)\right)}
{\widehat{\sigma}_n(\tau)}
\right| \quad \text{and} \quad
\widetilde{T}_n
=
\sup_{\tau\in\mathcal T_n}
\left|
\frac{1}{\sqrt n}
\sum_{i=1}^n
\frac{\widetilde\psi_i(\tau)}{\sigma_n(\tau)}
\right|.
\]
We prove
\[
\sup_{s\in\mathbb R}\left| \mathbb P_{\boldsymbol A_n}\big[\widehat{T}_n \le s\big] - \mathbb P_\ast\left[\sup_{1\leq i \leq p_n}\left|\widehat{Z}_n(\tau_i)\right|\le s\right] \right|=o_P(1),
\]
in three steps.

\noindent \underline{\bf Step 1.} We first compare the distributions of $\sup_{1\leq i \leq p_n}|  \widehat{Z}_n(\tau_i)|$ and $\sup_{1\leq i \leq p_n} \left|  Z_n(\tau_i) \right| $. By \cref{lemma:psd_corrected_gamma_consistency} and  \(\bar{r}_{\mathbb{V}, n} \left(\log p_n\right)^2 =o(1) \), it follows that
\[
\big\|
\widehat{\boldsymbol\Gamma}_n^+
-
\boldsymbol\Gamma_n
\big\|_{\max}
\left(\log p_n\right)^2
=o_P(1).
\]
The Gaussian comparison inequality in Lemma C.1 in \cite{chen2018gaussian} or  Theorem 4.1 of \cite{chernozhukov2017central} gives
\begin{equation}\label{eq: comparison_1}
\begin{aligned}
\sup_{s\in\mathbb R}
\left|
\mathbb P_\ast\left[
\sup_{1\le i\le p_n}|\widehat Z_n(\tau_i)|\le s
\right]
-
\mathbb P_{\boldsymbol A_n}\left[
\sup_{1\le i\le p_n}|Z_n(\tau_i)|\le s
\right]
\right| 
& \lesssim 
\big\|
\widehat{\boldsymbol\Gamma}_n^+
-
\boldsymbol\Gamma_n
\big\|_{\max}^{1/3}
\left(\log p_n\right)^{2/3}\\
& =  O_P\left(\bar{r}_{\mathbb{V}, n}^{1/3}  \log^{2/3} p_n\right)\\
& = o_P(1).
\end{aligned}
\end{equation}

\noindent \underline{\bf Step 2.} We next compare $\widetilde T_n$ with  $\sup_{1\le i\le p_n}|Z_n(\tau_i)|$. By \cref{lemma: grid_oracle_gaussian_approx},
\begin{equation}\label{eq: comparison_2}
\sup_{s\in\mathbb R}
\left|
\mathbb P_{\boldsymbol A_n}\big[\widetilde{T}_n\leq s\big]
-
\mathbb P_{\boldsymbol A_n}
\left[
\sup_{1\le i\le p_n}|Z_n(\tau_i)|\le s
\right]
\right|
=o(1).
\end{equation}

\noindent \underline{\bf Step 3.} We compare $\widetilde{T}_n$ with $\widehat{T}_n$. In particular, we aim to show 
\begin{equation}\label{eq: comparison_3}
\begin{aligned}
\sup_{s \in \mathbb{R}} \left| \mathbb P_{\boldsymbol A_n}\big[\widetilde{T}_n\leq s\big]
-
\mathbb P_{\boldsymbol A_n}
\big[ \widehat T_n \leq s
\big] \right| = o_P(1).
\end{aligned}
\end{equation}
Let $\epsilon_n \searrow 0$ be a sequence satisfying 
\[
  \bar{r}_{\mathbb{V}, n} \sqrt{\log p_n} + \eta_n  = o\left( \epsilon_n \right) \quad \text{and} \quad \epsilon_n \sqrt{\log p_n} = o(1).
\]
Such a sequence exists under the condition $\bar{r}_{\mathbb{V}, n} \log p_n + \eta_n \sqrt{\log p_n} = o(1)$. We verify below that $|\widehat{T}_n - \widetilde{T}_n| = o_P(\epsilon_n)$.

For any $s \in \mathbb{R}$, we have 
\[
\mathbb{P}_{\boldsymbol{A}_n}\big[\widehat{T}_n \leq s\big] \leq  \mathbb{P}_{\boldsymbol{A}_n}\big[ \widetilde{T}_n \leq s + \epsilon_n\big] + \mathbb{P}_{\boldsymbol{A}_n}\big[|\widetilde{T}_n - \widehat{T}_n| > \epsilon_n \big]. 
\]
Consequently, subtracting $\mathbb{P}_{\boldsymbol{A}_n}\big[ \widetilde{T}_n \leq s \big]$ from both hand sides gives
\[
\mathbb{P}_{\boldsymbol{A}_n} \big[\widehat{T}_n  \leq s\big] -  \mathbb{P}_{\boldsymbol{A}_n}\big[ \widetilde{T}_n \leq s \big]   \leq  \mathbb{P}_{\boldsymbol{A}_n}\big[ s < \widetilde{T}_n \leq s + \epsilon_n \big] + \mathbb{P}_{\boldsymbol{A}_n}\big[|\widetilde{T}_n - \widehat{T}_n| > \epsilon_n \big ]. 
\]
Interchanging the roles of \(\widehat{T}_n\) and \(\widetilde{T}_n\), applying the same argument gives
\[
\mathbb{P}_{\boldsymbol{A}_n}\big[\widetilde{T}_n \leq s\big]
-
\mathbb{P}_{\boldsymbol{A}_n} \big[\widehat{T}_n  \leq s\big]
\leq 
\mathbb{P}_{\boldsymbol{A}_n}\big[s - \epsilon_n< \widetilde{T}_n\leq s\big]
+
\mathbb{P}_{\boldsymbol{A}_n}\big[|\widehat{T}_n-\widetilde{T}_n|>\epsilon_n \big].
\]
Combining the two inequalities and applying \cref{lemma: grid_oracle_gaussian_approx}, it follows that 
\[
\begin{aligned}
\sup_{s \in \mathbb{R}} \left| \mathbb P_{\boldsymbol A_n}\big[\widetilde{T}_n\leq s\big]
-
\mathbb P_{\boldsymbol A_n}
\big[ \widehat T_n \leq s
\big] \right| & \leq \sup_{s \in \mathbb{R}} \mathbb{P}_{\boldsymbol{A}_n} \big[ |\widetilde{T}_n - s| \leq \epsilon_n  \big] + \mathbb{P}_{\boldsymbol A_n}\big[  |\widetilde{T}_n - \widehat T_n| >\epsilon_n  \big] \\
& \leq \sup_{s \in \mathbb{R}} \mathbb{P}_{\boldsymbol{A}_n} \left[ \left|\sup_{1\le i\le p_n}
\left|
Z_n(\tau_i)
\right|  - s \right| \leq \epsilon_n  \right] + o_P(1),
\end{aligned}
\]
provided $|\widetilde{T}_n - \widehat T_n|  = o_P(\epsilon_n)$. Define $\left\{ \upsilon_i: 1\leq i \leq 2p_n \right\}$ such that $\upsilon_i =  Z_n(\tau_i)$ for $1\leq i \leq p_n$, and $\upsilon_i =  - Z_n(\tau_{i-p_n})$ for $p_n + 1 \leq i \leq 2p_n$. By construction, $\sup_{1\le i\le p_n}
\left|Z_n(\tau_i)\right| = \sup_{1\leq i \leq 2p_n} \upsilon_i$. Applying Lemma 2.1 in  \cite{chernozhukov2013gaussian} and \cref{assumption: smallest_variance_diagonal} therefore gives 
\[
\begin{aligned}
\sup_{s \in \mathbb{R}} \left| \mathbb P_{\boldsymbol A_n}\big[\widetilde{T}_n\leq s\big]
-
\mathbb P_{\boldsymbol A_n}
\big[ \widehat T_n \leq s
\big] \right|   &  \leq    \sup_{s \in \mathbb{R}} \mathbb{P}_{\boldsymbol{A}_n} \left[ \left|\sup_{1\leq i \leq 2p_n} \upsilon_i  - s \right| \leq \epsilon_n  \right] + o_P(1) \\
& \lesssim   \epsilon_n \left( \sqrt{\log p_n} + \sqrt{ 1 \vee \log (  1 /\epsilon_n )  } \right) + o_P(1)\\
& = o_P(1).
\end{aligned}
\]

It remains to verify that $|\widetilde{T}_n - \widehat T_n|  = o_P(\epsilon_n)$. Let $T_n$ be the oracle version of $\widehat{T}_n$, obtained by replacing $\widehat{\sigma}_n(\tau)$ with  $\sigma_n(\tau)$: 
\[
T_n = \sup_{\tau \in \mathcal{T}}  \left|  \frac{\sqrt{n} \left( \widehat{q}(\tau) - q(\tau) \right) }{ \sigma_n(\tau) } \right|. 
\]
By \cref{lemma: feasible_variance_convergence}, $\sup_{\tau\in\mathcal T}\left|\widehat\sigma_n^2(\tau)-\sigma_n^2(\tau)\right|=O_P(\bar{r}_{\mathbb{V}, n} )$. Together with \cref{assumption: smallest_variance_diagonal}, this implies
\[
\sup_{\tau\in\mathcal T}
\left|
\frac{\widehat\sigma_n(\tau)}{\sigma_n(\tau)}-1
\right|
=
O_P(\bar{r}_{\mathbb{V}, n} ).
\]
Consequently,
\[
\big|\widehat T_n-T_n \big|\le
\sup_{\tau\in\mathcal T}
\left|
\frac{\sqrt n \left(\widehat q(\tau)-q(\tau)\right) }
{\sigma_n(\tau)}
\right| \times 
\sup_{\tau\in\mathcal T}
\left|
\frac{\sigma_n(\tau)}{\widehat\sigma_n(\tau)}-1
\right| = O_P\left( T_n \bar{r}_{\mathbb{V}, n}  \right).
\]
By  \cref{lemma: grid approximation},
\[
T_n
=
\sup_{\tau\in\mathcal T_n}
\left|
\frac1{\sqrt n}
\sum_{i=1}^n
\frac{\widetilde\psi_i(\tau)}{\sigma_n(\tau)}
\right|
+ O_P(\eta_n) =  \sup_{1\leq i\leq p_n}\big| \widetilde{\mathbb H}_{n}(\tau_i) \big| + O_P(\eta_n) ,
\]
According to \cref{lemma: grid_oracle_gaussian_approx}, the distribution of $\sup_{1\leq i\leq p_n}\big| \widetilde{\mathbb H}_{n}(\tau_i) \big|$ can be
approximated by $\sup_{1\le i\le p_n}\left|Z_n(\tau_i)\right|$. Lemma 2.1(b) of \cite{chernozhukov2013gaussian} and Markov's inequality give
\[
\sup_{1\le i\le p_n}|Z_n(\tau_i)|=O_P\left(\sqrt{\log p_n} \right) \Rightarrow
T_n=O_P\left(\sqrt{\log p_n}\right).
\]
Since \(\bar{r}_{\mathbb{V}, n} \log^2 p_n=o(1)\), it follows that
$\big|\widehat T_n-T_n\big|
=
O_P\left(\bar{r}_{\mathbb{V}, n} \sqrt{\log p_n} \right)$.  Moreover, \cref{lemma: grid approximation} gives $| T_n - \widetilde{T}_n | = O_P(\eta_n) $. Therefore, the triangle inequality and the definition of $\epsilon_n$ yield
\[
|\widetilde{T}_n - \widehat T_n|  \leq  \big|\widehat T_n-T_n\big|  +  | T_n - \widetilde{T}_n |  = O_P\left(  \bar{r}_{\mathbb{V}, n} \sqrt{\log p_n}  + \eta_n \right) = o_P(\epsilon_n).
\]

\noindent \underline{\bf Wrap up.} Combining \cref{eq: comparison_1,eq: comparison_2,eq: comparison_3} and applying the triangle inequality proves the theorem.
\end{proof}

\section{Auxiliary  Lemmas}

This section collects auxiliary lemmas used in the proofs of 
\cref{appendix:proof for section 3,appendix:proof for section 4}.

\begin{lemma}\label{lemma: Variance_Color_K}
Recall the definition of \(\widetilde{\sigma}_{n,k}^2\) in \cref{eq:sigma_nk}. Suppose that the conditions of \cref{lemma: Bahadur Representation-rate} hold. For every fixed \(M>0\), there exists a constant \(K>0\), independent of \(n\) and \(k\), such that, for each \(k=1,\ldots,\Delta_n+1\) and all sufficiently large \(n\),
\[
\mathbb{E}_{\boldsymbol A_n} \big[\widetilde{\sigma}_{n,k}^2  \big] \leq  K (\Delta_n+1) n^{-3/2}.
\]
\end{lemma}

\begin{proof}[Proof of \cref{lemma: Variance_Color_K}]
For each \(k=1,\ldots,\Delta_n+1\), define 
\[
\mathcal{J}_{n,k} = \mathbb{E}_{\boldsymbol A_n}\left[ \sup_{\tau \in \mathcal{T}, |u|\leq M} \left| \sum_{i \in \mathcal{C}_k} \Big( |R_i(u, \tau)| - \mathbb{E}\left[|R_i(u, \tau)| | \boldsymbol{A}_n\right] \Big)  \right| \right].
\]
Conditional on \(\boldsymbol A_n\), the variables indexed by \(i\in\mathcal C_k\) are mutually independent. Applying Lemma 2.3.1 in \cite{vaart2023empirical} therefore yields
\[
\mathcal{J}_{n,k}  \leq 2 \mathbb{E}_{\boldsymbol A_n}\left[ \mathbb{E}_\varepsilon \left[ \sup_{\tau \in \mathcal{T}, |u|\leq M} \left| \sum_{i \in \mathcal{C}_k} \varepsilon_i \left|R_i(u, \tau)\right| \right| \right]  \right].
\]
where $\{ \varepsilon_i \}_{i=1}^n$ are i.i.d. Rademacher variables independent of the data.

The bound in \cref{eq:squared_R} implies that $\lvert R_i(u,\tau)\rvert \le M/\sqrt{n}$ uniformly over $\tau \in \mathcal{T}$ and $|u|\leq M$. Hence, the constant function $M/\sqrt{n}$ is an envelope for the class $\mathcal{H}_n^\prime \equiv  \{Y_i \mapsto R_i(u,\tau): \tau \in \mathcal{T},  |u|\leq M \}$.
Since $\mathrm{VC}( \mathcal{H}_n^\prime  )$ is uniformly bounded, 
Corollary 2.2.8 in \cite{vaart2023empirical} implies that there are universal constants $A,K > 0$ such that
\[
\begin{aligned}
\mathbb{E}_\varepsilon \left[ \sup_{\tau \in \mathcal{T}, |u|\leq M} \left| \frac{1}{\sqrt{|\mathcal{C}_k|}} \sum_{i \in \mathcal{C}_k} \varepsilon_i |R_i(u, \tau)| \right| \right] 
& \leq K  \int_0^{\widetilde{\sigma}_{n,k}} \sqrt{2  \log \left( \frac{A /\sqrt{n} }{\epsilon} \right) }\mathrm{d}\epsilon   \\
& \leq \frac{2  K}{\sqrt{n}} \varphi ( \sqrt{n}\widetilde{\sigma}_{n,k} ) \lesssim 1/\sqrt{n},
\end{aligned}
\]
where $\varphi(x)=x\sqrt{2\log(A/x)}$. The last inequality follows because \(\sqrt n\,\widetilde{\sigma}_{n,k}\le M\) and \(\varphi\) is bounded on \([0,M]\), after choosing $M<A$. Consequently,
\[
\mathcal{J}_{n,k} \lesssim \sqrt{|\mathcal{C}_k|/n} = O(1).
\]
An argument analogous to that used in \cref{eq:Taylor_exp} yields
 \[
 \sup_{i\in\mathcal{C}_{k}} \sup_{\tau \in \mathcal{T}, |u|\leq M}  \mathbb{E}_{\boldsymbol A_n}\left[ \left|R_i(u, \tau) \right|     \right] = O\left(1/n\right).
 \]
 Therefore, we have
\begin{equation*}
\begin{aligned}
\mathbb{E}_{\boldsymbol A_n} \left[ \sup_{\tau \in \mathcal{T} , |u|\leq M } \sum_{i \in \mathcal{C}_k}  \left|R_i(u, \tau) \right| \right] &  \leq  \sup_{\tau \in \mathcal{T}, |u|\leq M } \sum_{i \in \mathcal{C}_k}  \mathbb{E}_{\boldsymbol A_n}\left[ \left|R_i(u, \tau) \right|     \right] \\
& + \mathbb{E}_{\boldsymbol A_n}\left[  \sup_{\tau \in \mathcal{T}, |u|\leq M } \left| \sum_{i \in \mathcal{C}_k}   \Big( |R_i(u, \tau)| - \mathbb{E}_{\boldsymbol A_n}\left[|R_i(u, \tau)| \right] \Big)  \right|   \right]  \\
&   \lesssim   |\mathcal{C}_k| / n  + \mathcal{J}_{n,k} = O(1).
\end{aligned}
\end{equation*}
Since $\lvert R_i(u,\tau)\rvert \le M/\sqrt{n}$, we have
\[
\begin{aligned}
\mathbb{E}_{\boldsymbol A_n}\left[ \sup_{\tau \in \mathcal{T}, |u|\leq M  }  \sum_{i\in \mathcal{C}_k } \left| R_i(u, \tau) \right|^2  \right] & \leq \frac{M}{\sqrt{n}}   \mathbb{E}_{\boldsymbol A_n}\left[ \sup_{\tau \in \mathcal{T}, |u|\leq M }  \sum_{i\in \mathcal{C}_k } \left| R_i(u, \tau) \right|    \right] \\
& \lesssim   \frac{M}{\sqrt{n}} \cdot  O(1)   = O \left( \frac{1}{\sqrt{n}} \right). 
\end{aligned}
\]
It follows that
\[
\mathbb{E}_{\boldsymbol A_n} \big[\widetilde{\sigma}_{n,k}^2  \big] = \frac{1}{|\mathcal{C}_k|} \mathbb{E}_{\boldsymbol A_n}\left[ \sup_{\tau \in \mathcal{T}, |u|\leq M}  \sum_{i\in \mathcal{C}_k } \left| R_i(u, \tau) \right|^2 \right]  = O \left(\frac{\Delta_n + 1}{n^{3/2}} \right).
\]
This completes the proof.
\end{proof}

\begin{lemma}
\label{lemma:L1n_uniform_bounded}
Suppose that the conditions of
\cref{lemma: Bahadur Representation-rate} hold. Recall that
\(L_{1,n}(\tau)\) is defined by
\[
L_{1,n}(\tau)
=
\frac{1}{\sqrt n}
\sum_{i=1}^n
\mathds{1}_i(w) \left(\tau-\mathds{1}\left\{Y_i\le q_{w}(\tau)\right\} \right).
\]
Then, conditional on  \(\{ \boldsymbol{A}_n \}_{n\ge1}\),
\[
\sup_{\tau\in\mathcal T}
\left|
L_{1,n}(\tau)
\right|
=
O_P(1).
\]
\end{lemma}

\begin{proof}
 By \cref{assumption:normal}, on the event
\(\{W_i=w\}\), we have $\mathds{1}\{Y_i\le q_{w}(\tau)\}
= \mathds{1}\{U_i\le \tau\}$.
Hence,
\[
L_{1,n}(\tau)
=
\frac{1}{\sqrt n}
\sum_{i=1}^n\gamma_i(\tau),
\quad\text{where}\quad
\gamma_i(\tau)
=
\mathds{1}_i(w)
\left(
\tau-\mathds{1}\{U_i\leq\tau\}
\right).
\]
Conditioning on \(\boldsymbol A_n\), the treatment assignment is independent of
the latent ranks, and \(U_i\) is uniformly distributed on \([0,1]\) conditional on
the degree by \cref{assumption:RCT,assumption:CE}. Therefore, $\mathbb E_{\boldsymbol A_n} \left[\gamma_i(\tau)  \right]=0$ for all $i \in [n]$ and $\tau \in \mathcal{T}$.
Let
\[
\mathcal H
=
\left\{
 (\bar{w}, \bar{u}) \mapsto
\mathds{1}\{\bar{w} = w\} \left(\tau-\mathds{1}\{\bar{u} \le\tau\} \right):
\tau\in\mathcal T
\right\}.
\]
It is easy to see that \(\mathcal H\) is a uniformly
bounded VC-type class. In particular, its envelope is bounded
by one, and there exist constants \(A,V<\infty\) such that, for every probability
measure \(Q\),
\[
\log N\left(\epsilon,\mathcal H,L_2(Q)\right)
\le
V\log(A/\epsilon),
\qquad 0<\epsilon<1.
\]

Let
\(\{\mathcal C_1,\ldots,\mathcal C_{\Delta_n+1}\}\)
denote the proper cover defined in \cref{sec:preliminary}. Because \(\boldsymbol A_n^{(2)}\) is a dependency graph, conditional
on \(\boldsymbol A_n\), the random vectors
\(\{(W_i,U_i):i\in\mathcal C_k\}\) are mutually independent within
each color class. Since \(\gamma_i(\cdot)\) is measurable with
respect to \((W_i,U_i)\), the processes
\(\{\gamma_i(\cdot):i\in\mathcal C_k\}\) are likewise mutually
independent conditional on \(\boldsymbol A_n\).

The triangle inequality gives
\[
\sup_{\tau\in\mathcal T}
\left|
\frac{1}{\sqrt n}\sum_{i=1}^n \gamma_i(\tau)
\right|
\le
\sum_{k=1}^{\Delta_{n}+1}
\sup_{\tau\in\mathcal T}
\left|
\frac{1}{\sqrt n}\sum_{i\in \mathcal C_k}\gamma_i(\tau)
\right|.
\]
Since the function class $\mathcal{H}$
is a uniformly bounded VC-type class,  using the same argument in  Theorem  2.14.1 in \cite{vaart2023empirical} gives that
\[
\mathbb E_{\boldsymbol A_n}
\left[
\sup_{\tau\in\mathcal T}
\left|
\sum_{i\in \mathcal C_k}\gamma_i(\tau)
\right| 
\right]
\lesssim
\sqrt{|\mathcal C_k|}.
\]
Therefore,
\[
\begin{aligned}
\mathbb E_{\boldsymbol A_n}
\left[
\sup_{\tau\in\mathcal T}
\left|
\frac{1}{\sqrt n}\sum_{i=1}^n \gamma_i(\tau)
\right| 
\right]
&\le
\sum_{k=1}^{\Delta_{n}+1}
\frac{1}{\sqrt n}
\mathbb E_{\boldsymbol A_n}
\left[
\sup_{\tau\in\mathcal T}
\left|
\sum_{i\in \mathcal C_k}\gamma_i(\tau)
\right| 
\right]
\\
&\lesssim
\sum_{k=1}^{\Delta_{n}+1}
\sqrt{\frac{|\mathcal C_k|}{n}}\leq 
\sqrt{\Delta_{n}+1}
\left(\sum_{k=1}^{\Delta_{n}+1} \frac{|\mathcal C_k|}{n}\right)^{1/2}
=
\sqrt {\Delta_{n}+1}.
\end{aligned}
\]
Since \(\Delta_n\) is uniformly bounded under \cref{ass:degree},
the preceding display and Markov's inequality imply that
\[
\sup_{\tau\in\mathcal T}
\left|
L_{1,n}(\tau)
\right| =  \sup_{\tau\in\mathcal T}
\left|
\frac{1}{\sqrt{n}}\sum_{i=1}^n \gamma_i(\tau)
\right| =
O_P(1).
\]
This completes the proof.
\end{proof}

\begin{lemma}
\label{lemma:uniform_argmin_localization}
Recall that $\widehat u_n(\tau)$ is defined in \cref{eq: u_n_definition}. Suppose that the conditions of \cref{lemma: Bahadur Representation-rate} hold. Then, conditional on 
\(\{ \boldsymbol{A}_n \}_{n\geq 1}\),
\[
\sup_{\tau\in\mathcal T} \left|\widehat u_n(\tau)\right|=O_P(1).
\]
\end{lemma}

\begin{proof}
\cref{eq: bounds_for_zeta_n,lemma:L1n_uniform_bounded} imply that
$\sup_{\tau\in\mathcal T}|u_n(\tau)|=O_P(1)$. Thus, for every \(\epsilon>0\), there exists a fixed \(M<\infty\) such that, for all sufficiently large \(n\), $\sup_{\tau\in\mathcal T}|u_n(\tau)|\le M/2$ holds with probability at least $1-\epsilon$. On this event, for any \(\tau\in\mathcal T\) and \(|u|\geq M\), $|u-u_n(\tau)|\ge M/2$. Thus, we can derive that
\[
Q_n(u,\tau)-Q_n(u_n(\tau),\tau)
=
\frac{1}{2}\zeta_n(w,\tau)\left|u-u_n(\tau)\right|^2
\geq \frac{1}{8} c_\zeta  M^2,
\]
for all $|u| \geq M$. \cref{eq:compact-uniform approximation} implies that, with probability tending to one,
\[
\sup_{\tau\in\mathcal T,\, |u|\le M}
\left|L_n(u,\tau)-Q_n(u,\tau) \right|
\le
\frac{M^2 c_\zeta}{32} .
\]
Combining the preceding two events, with probability at least \(1-\epsilon-o(1)\), for every \(\tau\in\mathcal T\) and \(u\in\{-M,M\}\),
\begin{equation}\label{eq: strict_boundary_inequality}
\begin{aligned}
L_n(u,\tau)-L_n(u_n(\tau),\tau)
&\ge
Q_n(u,\tau)-Q_n(u_n(\tau),\tau)
\\
&
-
\sup_{\tau\in\mathcal T, |u|\le M}
2\left|L_n(u,\tau)-Q_n(u,\tau)\right|
\\
&\ge \frac{M^2 c_\zeta}{16} >0.
\end{aligned}
\end{equation}
Because \(L_n(\cdot,\tau)\) is convex and $L_n\left( \pm M ,\tau\right)>L_n(u_n(\tau),\tau)$, every minimizer of \(L_n(\cdot,\tau)\) must lie in \((-M,M)\). Indeed, if a
minimizer \(\widehat u > M\) existed, there would be some $\lambda \in (0,1)$  such that \(  M =\lambda u_n(\tau) + (1-\lambda) \widehat u\). By convexity,
\[
L_n(M,\tau)
\le
\lambda L_n(u_n(\tau),\tau)
+
(1-\lambda)L_n(\widehat u,\tau)
\le
L_n(u_n(\tau),\tau),
\]
which contradicts \cref{eq: strict_boundary_inequality}.
The case \(\widehat u < -M\) follows analogously. Therefore,
\[
\widehat u_n(\tau)\in(-M,M).
\]

As a result, $\sup_{\tau\in\mathcal T}|\widehat u_n(\tau)|\le M$,
with probability at least \(1-\epsilon-o(1)\). Since, for every
\(\epsilon>0\), \(M\) can be chosen as a fixed constant independent
of \(n\), letting $\epsilon \rightarrow 0$ gives the desired result.
\end{proof}

\begin{lemma}\label{lemma: ULLN_U}
Suppose that \cref{ass:degree,assumption:local dependency} hold. Then, conditional on  $\left\{\boldsymbol{A}_n \right\}_{n\ge 1}$,
\[
\sup_{\tau \in \mathcal{T}}\left|\frac{1}{n}\sum_{i=1}^n  \mathds{1}\{ U_i \leq \tau \}  - \tau\right| = O_P \big( n^{-1/2} \big).
\]
\end{lemma}

\begin{proof}[Proof of \cref{lemma: ULLN_U}]
Define $F_n(\tau) = \frac{1}{n}\sum_{i=1}^n  \mathds{1}\{ U_i \leq \tau \}$. Recall the partition  $[n] = \mathcal{C}_1 \cup \mathcal{C}_2 \cup \dots \cup \mathcal{C}_{\Delta_{n}+1}$ introduced in \cref{sec:preliminary}. Conditional on \(\boldsymbol A_n\), the variables \(\{U_i:i\in\mathcal C_k\}\) are mutually independent and identically distributed within each color class. Moreover, $ \Delta_n$ is uniformly bounded by Assumptions \ref{ass:degree} and \ref{assumption:local dependency}. For each \(k=1,\ldots,\Delta_n+1\), define $F_{n,k}(\tau) = |\mathcal{C}_k|^{-1} \sum_{i \in  \mathcal{C}_k} \mathds{1}\{ U_i \leq \tau \}$. Then
\[
\sup_{\tau \in \mathcal{T} } | F_n(\tau) - \tau | \leq \sum_{k=1}^{\Delta_n + 1}  \left(  \frac{ |\mathcal{C}_k| }{n} \right) \sup_{\tau \in \mathcal{T} }  | F_{n,k} (\tau) - \tau  |.
\]
The Dvoretzky–Kiefer–Wolfowitz inequality gives, for any $\epsilon \geq 0$,
\[
\mathbb{P}_{\boldsymbol A_n}\left[ \sup_{\tau \in \mathcal{T} }  | F_{n,k} (\tau) - \tau  | > \epsilon  \right] \leq 2 e^{-2 |\mathcal{C}_k| \epsilon^2}.
\]
Applying Lemma 2.2.13 in \cite{durrett2019probability}  gives 
\[
\mathbb{E}_{\boldsymbol A_n} \left[  \sup_{\tau \in \mathcal{T} }  | F_{n,k} (\tau) - \tau  |  \right ] \leq \int_0^\infty 2 e^{-2 |\mathcal{C}_k| \epsilon^2} \mathrm{d}\epsilon =  \sqrt{ \frac{\pi}{2 |\mathcal{C}_k|} }.
\]
Therefore, we have 
\[
\mathbb{E}_{\boldsymbol A_n}\left[ \sup_{\tau \in \mathcal{T} } | F_n(\tau) - \tau |   \right] \leq \sum_{k=1}^{\Delta_n + 1}    \frac{ |\mathcal{C}_k| }{n} \sqrt{ \frac{\pi}{2 |\mathcal{C}_k|} } \leq \sqrt{ \frac{\pi (\Delta_n + 1)}{2n} }.
\]
By Markov’s inequality and \cref{ass:degree}, the desired result follows.
\end{proof}

\begin{lemma}\label{lemma:Oracle_variance}
Suppose \cref{ass:degree,assumption: density,assumption:local dependency} hold.  Then, conditional on $\{ \boldsymbol{A}_n \}_{n\geq 1}$,
\[
\sup_{\tau, \tau^\prime \in \mathcal{T}} \left| \widetilde{\mathbb{V}}_n(\tau,\tau')- \widetilde{\mathbb V}_{\boldsymbol{A}_n}(\tau,\tau')  \right| = O_P\big(n^{-1/2}\big).
\]
\end{lemma}

\begin{proof}[Proof of \cref{lemma:Oracle_variance}]

Recall from \cref{sec:preliminary} that $\boldsymbol A_n^{(2)}$ is a dependency graph for $\left( Y_i, W_i, U_i\right)_{i=1}^n$. Under \cref{ass:degree}, its maximum degree $\Delta_n$ is uniformly bounded. Define 
\[
\mathfrak{P}_n \equiv \left\{ (i,j)\in [n]^2: \ell_{\boldsymbol{A}_n}(i,j) \leq 2  \right \}.
\]
So, we have $n \leq |\mathfrak{P}_n| \le n(\Delta_n + 1) = O(n)$, adopting the convention that $\ell_{\boldsymbol{A}_n}(i, i) = 0$.
 
For each $\mathfrak{p}=(i,j)\in\mathfrak{P}_n$, define $Z_{\mathfrak{p}}(\tau,\tau')=\widetilde{\psi}_i(\tau)\widetilde{\psi}_j(\tau')$. Thus,
\[
\widetilde{\mathbb{V}}_n(\tau,\tau')- \widetilde{\mathbb V}_{\boldsymbol{A}_n}(\tau,\tau')=\frac{1}{n}\sum_{\mathfrak p\in\mathfrak P_n}\Big[ Z_{\mathfrak p}(\tau,\tau')-\mathbb E_{\boldsymbol{A}_n}\left[Z_{\mathfrak p}(\tau,\tau')\right] \Big].
\]
We construct a dependency graph $\mathfrak{G}_n$ on the vertex set $\mathfrak{P}_n$ as follows. For two distinct pairs $\mathfrak{p}_1=(i_1,j_1)$ and $\mathfrak{p}_2=(i_2,j_2)$ in $\mathfrak{P}_n$, we put an edge between $p_1$ and $p_2$ if either $\{i_1,j_1\}\cap\{i_2,j_2\}\neq\emptyset$, or there exist $k_1\in\{i_1,j_1\}$ and $k_2\in\{i_2,j_2\}$ such that $\boldsymbol{A}_n^{(2)}(k_1,k_2)=1$. For any pair $\mathfrak{p} \in \mathfrak{P}_n$, the degree of $\mathfrak{p}$ in $\mathfrak{G}_n$ is bounded by $2(\Delta_n + 1)^2$. Therefore, the maximum degree of $\mathfrak{G}_n$, denoted by $\Delta_{\mathfrak G_n}$, is uniformly bounded in $n$. The coloring argument in \cref{sec:preliminary} then yields a partition $\mathfrak{P}_n=\mathfrak C_{1}\cup\cdots\cup\mathfrak C_{\Delta_{\mathfrak{G}_n} + 1}$, such that within each $\mathfrak C_{k}$ the pair-level variables $Z_{\mathfrak{p}}$ are mutually independent conditional on $\boldsymbol{A}_n$.

Recall that $O_{i}=(Y_{i},W_{i})$, and  we define the following bounded function class
\[
\mathcal{G}_{n} = \left\{ (O_i, O_j) \mapsto \widetilde{\psi}_i(\tau)\widetilde{\psi}_j(\tau') : (\tau, \tau') \in \mathcal{T}^2 \right\}.
\]
For each \(w^\dagger\in\{w,w'\}\), the class \(\left\{
y\mapsto\mathds{1}\{y\le q_{w^\dagger}(\tau)\}:\tau\in\mathcal T\right\}\) is VC. In addition, $\zeta_n(w^\dagger,\tau)=f_{w^\dagger}(q_{w^\dagger}(\tau))\pi_{n}(w^\dagger)$ is uniformly bounded and Lipschitz by \cref{eq: bounds_for_zeta_n,assumption: density}. Lemma 2.6.18 in \cite{vaart2023empirical} therefore implies that \(\mathcal G_n\) is VC-type uniformly in \(n\) and admits a uniformly bounded envelope.


We now apply the symmetrization and VC maximal-inequality arguments used in Step 2 of the proof of \cref{lemma: Bahadur Representation-rate}.  For each $k=1,2,\dots,\Delta_{\mathfrak{G}_{n}}+1$, let
\[
\widehat{\sigma}_{n,k}^2 = \sup_{\tau, \tau' \in \mathcal{T}} \frac{1}{|\mathfrak C_{k}| } \sum_{\mathfrak p\in\mathfrak C_{k}} \left|Z_{\mathfrak p}(\tau, \tau^\prime) \right|^2.
\]
It is not difficult to see that  $\widehat{\sigma}_{n,k}$ is uniformly bounded by some universal constant $C >0$ not depending on $n$ and $k$. Theorem  2.14.1 in \cite{vaart2023empirical} then gives
\[
\begin{aligned}
\mathbb E_{\boldsymbol{A}_n}\left[ \sup_{\tau,\tau'\in\mathcal T}\left| \frac{1}{\sqrt{|\mathfrak C_{k}|}} \sum_{\mathfrak p\in\mathfrak C_{k}}\left \{Z_{\mathfrak p}(\tau,\tau')-\mathbb E_{\boldsymbol A_n} \left[Z_{\mathfrak p}(\tau,\tau')\right] \right\}\right| \right]
& \lesssim \int_0^{\widehat{\sigma}_{n,k} } \sqrt{ \log (A/\varepsilon) } \mathrm{d} \varepsilon, \\
& \lesssim \int_0^{C } \sqrt{ \log (A/\varepsilon) } \mathrm{d} \varepsilon.
\end{aligned}
\]
Summing over the color classes gives
\[
\begin{aligned}
&\mathbb E_{\boldsymbol{A}_n}\left[\sup_{\tau,\tau'\in\mathcal T}\left|\frac1n\sum_{\mathfrak p\in\mathfrak P_n}\{Z_{\mathfrak p}(\tau,\tau')-\mathbb E_{\boldsymbol A_n}[Z_{\mathfrak p}(\tau,\tau')]\}\right| \right]  \\
\le \ &\frac1n\sum_{k=1}^{\Delta_{\mathfrak{G}_n} + 1}\mathbb E_{\boldsymbol{A}_n}\left[\sup_{\tau,\tau'\in\mathcal T}\left|\sum_{\mathfrak p\in\mathfrak C_{k}}\{Z_{\mathfrak p}(\tau,\tau')-\mathbb E_{\boldsymbol A_n}[Z_{\mathfrak p}(\tau,\tau')]\}\right| \right] \\
\lesssim \ &\frac1n\sum_{k=1}^{\Delta_{\mathfrak{G}_n} + 1} \sqrt{|\mathfrak C_{k}|}\le\frac1n  \left[(\Delta_{\mathfrak{G}_n} + 1)\sum_{k=1}^{\Delta_{\mathfrak{G}_n} + 1}|\mathfrak C_{k}|\right]^{1/2}=\frac1n\sqrt{(\Delta_{\mathfrak{G}_n} + 1)|\mathfrak P_n|}=O(n^{-1/2}),
\end{aligned}
\]
where the last equality follows from $\Delta_{\mathfrak{G}_n} =O(1)$ and $|\mathfrak P_n|=O(n)$. Markov's inequality yields the stated uniform rate.
\end{proof}

\begin{lemma}\label{lemma: score_LLN}
Suppose \cref{ass:degree,assumption: density,assumption:local dependency,assumption:density_est_convergence} hold. Then, conditional on  $\{ \boldsymbol{A}_n \}_{n\geq 1}$,
\begin{equation}\label{eq:squared_convergence}
\sup_{\tau \in \mathcal{T}} \sqrt{\frac{1}{n} \sum_{i=1}^n \left| \widehat{\psi}_i(\tau) -  \widetilde{\psi}_i(\tau) \right|^2} = O_P \big(  r_{f,n} \vee n^{-1/4} \big).    
\end{equation}
\end{lemma}
\begin{proof}[Proof of \cref{lemma: score_LLN}]

For $w^\dagger \in \{w, w^\prime \}$, define $\widehat{\zeta}\left(w^\dagger, \tau\right) = \widehat{f}_{w^\dagger}\big(  \widehat{q}_{w^\dagger} (\tau)  \big) \widehat{\pi}(w^\dagger)$, and 
\[
 \widehat{\phi}_i\big(w^\dagger, \tau\big) = \frac{ \mathds{1}_i(w^\dagger)  }{\widehat{\zeta}\left(w^\dagger, \tau\right)  } \left(\tau - \mathds{1}\{ Y_i \leq \widehat{q}_{w^\dagger}(\tau) \} \right).
\]
To establish \cref{eq:squared_convergence}, it suffices to show that, for each $w^\dagger \in \{w, w^\prime\}$,
\[
\sup_{\tau \in \mathcal{T}} \frac{1}{n} \sum_{i=1}^n \left|  \widehat{\phi}_i(w^\dagger, \tau) -   \widetilde{\phi}_i(w^\dagger, \tau)\right|^2 = O_P \left(  r_{f,n}^2 \vee  n^{-1/2} \right).   
\]
We prove the result for $w^\dagger=w$; the argument for $w^\dagger=w'$ is identical. Consider the following decomposition:
\begin{equation}\label{eq:sum_phi}
\begin{aligned}
\frac{1}{n}\sum_{i=1}^n \left|\widehat{\phi}_i\left(w, \tau\right) -  \widetilde{\phi}_i\left(w, \tau\right) \right|^2 & \leq  \frac{1}{n}\sum_{i=1}^n    \mathds{1}_i(w)   \left| \widehat{\zeta}\left(w, \tau\right)^{-1}  - \zeta_n\left(w, \tau\right)^{-1} \right|^2     \\
 & +   \frac{1}{n}\sum_{i=1}^n  \frac{  \mathds{1}_i(w)    }{\zeta_{n}\left(w, \tau\right)^2 }   \left| \mathds{1}\{ Y_i \leq \widehat{q}_{w}(\tau) \} - \mathds{1}\{ Y_i \leq q_{w}(\tau) \}   \right|^2.
\end{aligned}
\end{equation}

By the triangle inequality, \cref{lemma:q_n_Nonasymptotics,assumption:density_est_convergence}, it follows that
\[
\begin{aligned}
\sup_{\tau \in \mathcal{T}} \left| \widehat{f}_w( \widehat{q}_w(\tau) ) -  f_w(q_w(\tau) ) \right| &\leq  \sup_{\tau \in \mathcal{T}} \left| \widehat{f}_w( \widehat{q}_w(\tau) ) -  f_w(\widehat{q}_w(\tau) ) \right|     \\
 & +   \sup_{\tau \in \mathcal{T}} \left|   f_w(\widehat{q}_w(\tau) ) - f_w(q_w(\tau) )  \right|\\
 & = O_P \left( n^{-1/2} +  r_{f,n} \right) = O_P \big(  r_{f,n}\big).
\end{aligned}
\]
The definitions of $\widehat{\zeta}\left(w,\tau\right)$ and $\zeta_{n}(w,\tau)$, together with \cref{lemma: pi_n_root_convergence}, yield
\[
\begin{aligned}
\sup_{\tau \in \mathcal{T}} \left| \widehat{\zeta}\left(w, \tau\right)   - \zeta_n(w, \tau) \right| 
& = O_P \big(  r_{f,n}\big).
\end{aligned}
\]
Since $\widehat{\zeta}\left(w,\tau\right)$ and $\zeta_{n}(w,\tau)$ are uniformly bounded away from zero with probability approaching one,
\[
\sup_{\tau \in \mathcal{T}} \left| \widehat{\zeta}\left(w, \tau\right)^{-1}   - \zeta_n(w, \tau)^{-1} \right|^2  =   O_P \big(  r_{f,n}^2 \big).
\]

By \cref{lemma:q_n_Nonasymptotics},  we have $\widehat{\epsilon}_n \equiv \sup_{\tau \in \mathcal{T}} \left| \widehat{q}_w(\tau) - q_w(\tau) \right| = O_P\big(n^{-1/2}\big)$.  Thus, we obtain the inequality
\[
 \left| \mathds{1}\{ Y_i \leq \widehat{q}_{w}(\tau) \} - \mathds{1} \left\{ Y_i \leq q_{w}(\tau) \}  \right| \leq \mathds{1}\{ |Y_i - q_w(\tau)| \leq \widehat{\epsilon}_n \right \}.
\]
Consequently, \cref{lemma: ULLN_quantile} gives
\[
\begin{aligned}
 &\sup_{\tau \in \mathcal{T}} \frac{1}{n} \sum_{i=1}^n \mathds{1}_i(w)    \left| \mathds{1}\{ Y_i \leq \widehat{q}_{w}(\tau) \} - \mathds{1}\{ Y_i \leq q_{w}(\tau) \}  \right|^2\\
 \leq \ & \sup_{\tau \in \mathcal{T}} \frac{1}{n} \sum_{i=1}^n   \mathds{1}_i(w)   \mathds{1} \left\{ |Y_i - q_w(\tau)| \leq \widehat{\epsilon}_n  \right\} = O_P \big(n^{-1/2} \big).
\end{aligned}
\]
Combining these bounds with \cref{eq:sum_phi} gives
\[
\sup_{\tau \in \mathcal{T}} \sqrt{\frac{1}{n}\sum_{i=1}^n \left|\widehat{\phi}_i\left(w, \tau\right) -  \widetilde{\phi}_i\left(w, \tau\right) \right|^2} = O_P \big( r_{f,n} \vee n^{-1/4} \big).
\]
This completes the proof.
\end{proof}

\begin{lemma}\label{lemma:local_modulus_gamma}
 Suppose  \cref{ass:degree,assumption:local dependency} hold. Then, conditional on  \(\{ \boldsymbol{A}_n \}_{n\ge1}\), for any $\delta \in \left(0, 1/2\right)$,
\[
\mathbb E_{\boldsymbol A_n}
\left[
\sup_{|\tau-\tau'|\le \delta}
\left|
\frac1{\sqrt n}
\sum_{i=1}^n
\left[\gamma_i(\tau)-\gamma_i(\tau')\right]
\right|
\right]
\lesssim
\sqrt{ \delta \log\left(e/\delta\right)} + \frac{ \log \left(e/\delta\right)}{\sqrt{n}}.
\]
\end{lemma}

\begin{proof}
Without loss of generality, we consider $\tau_1 < \tau_2$ such that $\tau_2 - \tau_1 < \delta$, since the case $\tau_2 < \tau_1$ follows immediately by symmetry.  Recall that
\begin{equation}\label{eq: f_delta_def}
f_{\tau_1,\tau_2}(\bar w,\bar u)=\mathds{1}\{\bar w=w\}\left[(\tau_2-\tau_1)-\mathds{1}\{\tau_1<\bar u\le \tau_2\}
\right],
\end{equation}
and
\[
\mathcal F_\delta
=
\left\{
f_{\tau_1,\tau_2}:
\tau_1<\tau_2<\tau_1+\delta,\ 
\tau_1,\tau_2\in\mathcal T
\right\},
\]
for any fixed $\delta >0$ defined in \cref{appendix:proof of proposition:NP_quantile_process_convergence}. By construction, $\gamma_i(\tau_2)-\gamma_i(\tau_1)=f_{\tau_1,\tau_2}(W_i,U_i)$. Then,
\[
\sup_{|\tau-\tau'|\le\delta}
\left|
\frac1{\sqrt n}
\sum_{i=1}^n
\left[\gamma_i(\tau)-\gamma_i(\tau')\right]
\right|
=
\sup_{f\in\mathcal F_\delta}
\left|
\frac1{\sqrt n}
\sum_{i=1}^n f(W_i,U_i)
\right|.
\]
\cref{eq:Markov,eq:integral_1} imply that there exist universal constants $A,K>0$ such that
\begin{equation}\label{eq: maximal_inequality_228}
\begin{aligned}
\mathbb{E}_{\boldsymbol{A}_n} \left[   \sup_{f\in \mathcal{F}_{\delta} }  \left|  \frac{1}{\sqrt{n}} \sum_{i=1}^n  f(W_i, U_i) \right|  \right] & \leq K  \mathbb{E}_{\boldsymbol{A}_n} \left[    \int_0^{ \sqrt{\Delta_n +1}  \sigma_n  } \sqrt{ 2 \log (A/\epsilon)  } \mathrm{d} \epsilon \right] \\
& = K\mathbb{E}_{\boldsymbol{A}_n} J\left(  \sqrt{\Delta_n +1}  \sigma_n \right),
\end{aligned}
\end{equation}
where $\sigma_n = \sup_{f \in \mathcal{F}_{\delta}} \sqrt{ n^{-1} \sum_{i=1}^n |f(W_i, U_i)|^2  }$, and $J\left(x\right)=\int_0^x \sqrt{2\log\left(A/\epsilon\right)}\mathrm{d}\epsilon$. Notice that
\begin{align*}
\sigma_{n}^2 & =  \sup_{f \in \mathcal{F}_{\delta} }  \frac{1}{n} \sum_{i=1}^n \left| f(W_i,U_i)  \right|^2  =\sup_{\tau_1<\tau_2<\tau_1+\delta} \frac{1}{n} \sum_{i=1}^n   \left| f_{\tau_1, \tau_2}(W_i,U_i)  \right|^2 \\ 
& \leq  \sup_{\tau_1<\tau_2<\tau_1+\delta}  \frac{1}{n} \sum_{i=1}^n    \left| (\tau_2 - \tau_1) -  \mathds{1}\{ \tau_1 < U_i \leq \tau_2 \} \right|^2 \\
& \leq \delta^2 +  \delta +  \sum_{k=1}^{\Delta_{n}+1} \frac{|\mathcal{C}_{k}|}{n} \sup_{ \tau_1 < \tau_2 < \tau_1 + \delta } \left|   \frac{1}{|\mathcal{C}_{k}|} \sum_{i\in\mathcal{C}_{k}} \mathds{1}\{ \tau_1 < U_i \leq \tau_2 \} - (\tau_2- \tau_1) \right| .
\end{align*}
Since $\{ U_i : i \in \mathcal{C}_k \}$ are i.i.d. conditional on $\boldsymbol{A}_n$, then applying \cref{lemma:local_interval_increment} implies that
\begin{equation}\label{eq: sigma_n_k}
\mathbb{E}_{\boldsymbol{A}_n}  \sigma_{n}^2   \leq  2 \delta + K\sum_{k=1}^{\Delta_{n}+1} \frac{|\mathcal{C}_{k}|}{n}\left(
\sqrt{\frac{\delta\log(e/\delta)}{|\mathcal{C}_k|}}
+
\frac{\log(e/\delta)}{|\mathcal{C}_k|}
\right)
 \leq  K\left(
\delta+ \frac{\log(e/\delta)}{n}
\right).
\end{equation}
Hence, by \cref{lem:entropy-integral-bound}, Jensen's
inequality and Cauchy-Schwarz inequality,
\[
\mathbb E_{\boldsymbol{A}_n} J\left(  \sqrt{\Delta_n +1}  \sigma_n \right)
\le
J\left( \sqrt{\Delta_n +1} \mathbb E_{\boldsymbol{A}_n}\sigma_{n}\right)
\le
J\left( 1 \wedge \sqrt{(\Delta_n +1)\mathbb E_{\boldsymbol{A}_n}\sigma_{n}^2}\right).
\]
Combining \eqref{eq: maximal_inequality_228},
\eqref{eq: sigma_n_k}, and \cref{lem:entropy-integral-bound}, we obtain
\begin{equation}\label{eq:one_color_gamma_bound}
\begin{aligned}
\mathbb{E}_{\boldsymbol{A}_n} \left[   \sup_{f\in \mathcal{F}_{\delta} }  \left|  \frac{1}{\sqrt{n}} \sum_{i=1}^n  f(W_i, U_i) \right|  \right] &\le
K J\left( \sqrt{(\Delta_n +1)\mathbb E_{\boldsymbol{A}_n}\sigma_{n}^2}\right) \\
& \lesssim\sqrt{\left(\delta+\log(e/\delta)/n\right)\log\left(\frac{e}{\delta+\log(e/\delta)/n}\right)}\\
& \leq_{(1)}\sqrt{\left(\delta+\log(e/\delta)/n\right)\log(e/\delta)
}  \\
&\leq \sqrt{\delta\log(e/\delta)}+ n^{-1/2}\log(e/\delta),
\end{aligned}
\end{equation}
where the inequality (1) holds since $\delta+ \log(e/\delta)/n
\ge \delta$. This proves the desired bound.
\end{proof}

\begin{lemma}\label{lem:entropy-integral-bound}
Let \(A\ge e\) be a constant and define $J\left(x\right)=\int_0^x \sqrt{2\log\left(A/\epsilon\right)}\mathrm{d}\epsilon$ for $x \in (0,1]$. Then \(J\) is increasing and concave on \((0,1]\). Moreover, there exists a
universal constant \(K>0\) such that, for all \(0<x\le 1\),
\[
J\left(x\right)\le Kx\sqrt{\log(A/x)}.
\]
\end{lemma}

\begin{proof}[Proof of \cref{lem:entropy-integral-bound}]
By the change of variables and $\int_0^1\sqrt{2\log(1/u)}\mathrm{d}u
=
\sqrt{\pi/2}$, we have
\[
\begin{aligned}
J\left(x\right)
&=
x\int_0^1
\sqrt{2\log(A/x)+2\log(1/u)}\mathrm{d} u \\
& \le
x\sqrt{2\log(A/x)}
+
x\int_0^1\sqrt{2\log(1/u)}\mathrm{d}u \\
& =
x\sqrt{2\log(A/x)}
+
x\sqrt{\pi/2}.
\end{aligned}
\]
Because \(A\ge e\) and \(x\le 1\), we have \(\log(A/x)\ge 1\). Hence, $J\left(x\right)\le K x\sqrt{\log(A/x)} $
for a universal constant \(K>0\). Next, by simple algebra, we have for \(0<x\le 1\)
\[
J'(x)=\sqrt{2\log(A/x)} \quad \text{and} \quad
J''(x)
=
-\frac{1}{x\sqrt{2\log(A/x)}}<0.
\]
 Therefore \(J\) is increasing and concave on \((0,1]\).
 \end{proof}

\begin{lemma}\label{lemma: ULLN_quantile}
Suppose \cref{ass:degree,assumption: density,assumption:local dependency}  hold. Let $\epsilon_n\geq0$ be any possibly random sequence satisfying $\epsilon_n = O_{P}(n^{-1/2})$.  Then, conditional on $\{ \boldsymbol{A}_n \}_{n\geq 1}$, we have that for $w^\dagger \in \{w,w^\prime \}$, 
\[
\sup_{\tau \in \mathcal{T}} \frac{1}{n} \sum_{i=1}^n \mathds{1}_i(w^\dagger) \mathds{1}\{ |Y_i - q_{w^\dagger}(\tau)| \leq \epsilon_n \} = O_P\big(n^{-1/2}\big).
\]
\end{lemma}

\begin{proof}

We define a random function $\mathbb{S}_{n}: \mathbb{R}_+ \rightarrow \mathbb{R}_+$ as
\[
\mathbb{S}_{n}(r)\equiv\sup_{\tau\in\mathcal T}\frac{1}{n}\sum_{i=1}^n\mathds{1}_i(w^\dagger)\mathds{1}\left\{\left|Y_i-q_{w^\dagger}(\tau)\right|\leq r\right\}.
\]
\noindent \underline{\textbf{Step 1.}}  
 We first assume that $\left(\epsilon_n \right)_{n=1}^\infty$ is a deterministic sequence. Under Assumptions \ref{assumption:RCT}, \ref{assumption:CE} and \ref{assumption: density}, it follows that there is a constant $c_{\mathrm{Lip}} > 0$ not depending on $\tau$ and $\epsilon$ such that
\[
\begin{aligned}
\mathbb{P}\left[ |Y_i - q_{w^\dagger}(\tau)| \leq \epsilon | W_i =w^\dagger, \boldsymbol{A}_n \right] & = F_{w^\dagger}(  q_{w^\dagger}(\tau) + \varepsilon ) - F_{w^\dagger} \left(  q_{w^\dagger}(\tau)- \epsilon \right ) \\
& \leq  c_{\mathrm{Lip}} \epsilon.
\end{aligned}
\]
For each $\tau \in \mathcal{T}$, define $a^{\pm}_{n}(\tau) = F_{w^\dagger}\left( q_{w^\dagger}(\tau) \pm \epsilon_n \right)$, where $F_{w^\dagger}$ is the conditional cdf of $Y_i | W_i =w^\dagger$. Since $q_{w^\dagger}(\tau) = F_{w^\dagger}^{-1}(\tau)$, then conditional on $W_i = w^\dagger$,
\[
\left\{  \left|Y_i - q_{w^\dagger}(\tau) \right| \leq \epsilon_n  \right\} = \{  a_n^-(\tau) \leq U_i \leq a_n^+(\tau)  \}
\]
By the boundedness of $f_w$ given by \cref{assumption: density}, for sufficiently large $n$, it follows that
\[
\sup_{\tau \in \mathcal{T}}\left| a_n^+(\tau) - a_n^-(\tau) \right| \leq c_{\mathrm{Lip}} \epsilon_n.
\]
Therefore, applying \cref{lemma: ULLN_U} gives that
\[
\begin{aligned}
\mathbb{S}_{n}(\epsilon_n) = &\sup_{\tau \in \mathcal{T}} \frac{1}{n} \sum_{i=1}^n \mathds{1}\left\{ W_i =w^\dagger, |Y_i - q_{w^\dagger}(\tau)|    \leq \epsilon_n \right \}
 \leq \sup_{|b - a| \leq \epsilon_n } \frac{1}{n} \sum_{i=1}^n \mathds{1}\{ a \leq  U_i \leq b \}   \\
 \leq    &  \sup_{|b - a| \leq \epsilon_n }  \left[ \left|\frac{1}{n} \sum_{i=1}^n  \mathds{1}\{  U_i \leq a \} - a \right|   +   \left|\frac{1}{n} \sum_{i=1}^n  \mathds{1}\{  U_i \leq b \} - b \right|   \right] + \epsilon_n 
 =    O_P \big(n^{-1/2} \big) + \epsilon_n .
\end{aligned}
\]
Since $\epsilon_n = O(n^{-1/2})$, we have $\mathbb{S}_{n}(\epsilon_n)= O_P(n^{-1/2})$.

\noindent \underline{\textbf{Step 2.}}   We next allow the sequence $\left(\epsilon_n\right)_{n=1}^\infty$ to be random with $\epsilon_n=O_P(n^{-1/2})$. Then, for any $\eta > 0$ sufficiently small, there exists a constant \(c_\eta<\infty\) such that, for all sufficiently large \(n\),
\[
\mathbb P\left[\epsilon_n>c_\eta /\sqrt{n}\right]\leq\frac{\eta}{2}.
\]
Applying the  result shown in Step 1, there exists a constant \(M_\eta<\infty\) such that, for all sufficiently large \(n\),
\[
\mathbb P\left[\sqrt n\,\mathbb{S}_{n}\left(c_\eta/\sqrt{n}  \right)>M_\eta\right]\leq\frac{\eta}{2}.
\]
Because \(r\mapsto \mathbb{S}_{n}(r)\) is non-decreasing, we have
$\mathbb{S}_{n}(\epsilon_n)\leq \mathbb{S}_{n}\left(c_\eta /\sqrt{n} \right)$ on the event $\left\{\epsilon_n\leq c_\eta /\sqrt{n}\right\}$.
It follows that for all sufficiently large \(n\):
\[
\begin{aligned}
\mathbb P\left[\sqrt n\,\mathbb{S}_{n}(\epsilon_n)>M_\eta\right]
&\leq\mathbb P\left[\epsilon_n>c_\eta /\sqrt{n}\right] +\mathbb P\left[\sqrt n\,\mathbb{S}_{n}(c_\eta/\sqrt{n})>M_\eta\right]\leq\eta
\end{aligned}
\]
Therefore, we have $\mathbb{S}_{n}(\epsilon_n)=O_P(n^{-1/2})$,
which proves the result.
\end{proof}

\begin{lemma}\label{lemma:q_n_Nonasymptotics}
Suppose \cref{ass:degree,assumption: density,assumption:local dependency} hold. Then, conditional on  $\{ \boldsymbol{A}_n \}_{n\geq 1}$, for each $w^\dagger \in \{ w, w^\prime \}$,
\[
\sup_{\tau \in \mathcal{T}} \left| \widehat{q}_{w^\dagger}(\tau) - q_{w^\dagger}(\tau) \right| = O_P\big( n^{-1/2} \big).
\]
\end{lemma}

\begin{proof}[Proof of \cref{lemma:q_n_Nonasymptotics}]

The Bahadur representation in \cref{lemma: Bahadur Representation-rate} implies that it suffices to show
\[
\sup_{\tau \in \mathcal{T}}\left| \frac{1}{\sqrt{n}} \sum_{i=1}^n \widetilde{\phi}_{i}\big( w^\dagger, \tau \big) \right|  = O_P(1).
\]   
Because the denominator $\zeta_n(w,\tau)$ is uniformly bounded away from zero by \cref{eq: bounds_for_zeta_n}, it remains to show
\begin{equation}\label{eq: w_dagger_O_P_1}
\sup_{\tau \in \mathcal{T}}\left| \frac{1}{\sqrt{n}} \sum_{i=1}^n \mathds{1}_i(w^\dagger) \left(\tau - \mathds{1}\{ Y_i \leq q_{w^\dagger}(\tau) \} \right) \right|  = O_P(1).    
\end{equation}
The conclusion in \cref{eq: w_dagger_O_P_1} is precisely the bound established in \cref{lemma:L1n_uniform_bounded}. This completes the proof.
\end{proof}

\begin{lemma}\label{lemma: pi_n_root_convergence}
Suppose \cref{ass:degree} holds. Then, conditional on  $\{ \boldsymbol{A}_n \}_{n\geq 1}$, for each $w^\dagger \in \{ w, w^\prime \}$,
\[
\left| \widehat{\pi} (w^\dagger) - \pi_{n}(w^\dagger) \right| = O_P\big( n^{-1/2} \big).
\]
\end{lemma}

\begin{proof}[Proof of \cref{lemma: pi_n_root_convergence}]
Recall that $\boldsymbol{A}_n^{(2)}$ is a dependency graph of $\{W_i\}_{i=1}^n$ with maximum degree $\Delta_n$ that is defined in \cref{sec:preliminary}. Since $\widehat{\pi} (w^\dagger) = n^{-1}\sum_{i=1}^n \mathds{1}_i(w^\dagger)$, then
\[
\begin{aligned}
\mathrm{Var}_{\boldsymbol{A}_n}\left( \widehat{\pi} (w^\dagger)  \right) &  =  \frac{1}{n^2}\sum_{i=1}^n \sum_{\ell_{\boldsymbol{A}_n}(i,j) \leq 2 } \mathrm{Cov}_{\boldsymbol{A}_n}\left( \mathds{1}_i(w^\dagger) ,  \mathds{1}_j(w^\dagger)  \right)  \leq \frac{\Delta_n + 1}{n}.
\end{aligned}
\]
By Chebyshev's Inequality, for any $\epsilon >0$, 
\[
\mathbb{P}_{\boldsymbol{A}_n}\left[  \left| \widehat{\pi}(w^\dagger) - \pi_{n}(w^\dagger) \right| > \epsilon \sqrt{ (\Delta_n +1) /n } \right] \lesssim \epsilon^{-2}.
\]
Because $\Delta_n$ is uniformly bounded under \cref{ass:degree}, it follows that
\[
\left| \widehat{\pi}(w^\dagger) - \pi_{n}(w^\dagger) \right| = O_P\left( \sqrt{ \frac{\Delta_n +1}{n} } \right) = O_P\big(n^{-1/2}\big).
\]
\end{proof}

\begin{lemma}
\label{lemma:local_interval_increment}
Let $U_i \overset{\mathrm{i.i.d.}}{\sim} \mathrm{Unif}(0,1)$.
Then, there is a universal constant $K >0$ such that for any \(0<\delta<1/2\) and $k=1,2,\dots,\Delta_{n}+1$,
\[
\mathbb{E}\left[\sup_{ \tau_2 - \tau_1 < \delta }\left|\frac{1}{|\mathcal C_k|}\sum_{i\in\mathcal C_k}\mathds 1 \left\{ \tau_1 <U_i\le\tau_2 \right\}- (\tau_1- \tau_2)\right|\right]\le K\left[\sqrt{\frac{\delta\log(e/\delta)}{|\mathcal C_k|}}+\frac{\log(e/\delta)}{|\mathcal C_k|}\right].
\]

\end{lemma}

\begin{proof}[Proof of \cref{lemma:local_interval_increment}]
Define a function class $\mathcal G_\delta$ as
\[
\mathcal G_\delta
=
\left\{u \mapsto \mathds \mathds{1}\{ \tau_1 <u \leq \tau_2 \}: 0\le \tau_1<\tau_2\le 1,\ \tau_2-\tau_1<\delta  \right\}.
\]
The class \(\mathcal G_\delta\) is a VC-type class  that is uniformly bounded. Moreover, $\sup_{g\in\mathcal G_\delta}\mathbb E \left|g(U_i)\right|^2 
\le \delta $. The desired result follows from Corollary 5.1 in \cite{chernozhukov2014gaussian}.
\end{proof}

\bibliographystyle{apalike}
\bibliography{reference.bib}

\end{document}